\documentclass[11pt]{article}

\usepackage[a4paper,margin=1in]{geometry}
\usepackage[T1]{fontenc}

\usepackage{mathtools} 
\mathtoolsset{showonlyrefs=true}
\usepackage{amssymb}
\usepackage{bm}
\usepackage{dsfont} 
\usepackage{mathrsfs}
\usepackage{upgreek}
\allowdisplaybreaks

\usepackage{algorithm}
\usepackage{algpseudocode}

\usepackage{graphicx}
\usepackage[font=small]{caption}
\usepackage{booktabs}
\usepackage[flushleft]{threeparttable}
\usepackage{tabularx}
\usepackage{float}

\usepackage[dvipsnames]{xcolor}
\usepackage[normalem]{ulem}
\usepackage{soul}
       
\usepackage{amsthm}
\theoremstyle{plain}
\newtheorem{theorem}{Theorem}
\newtheorem{lemma}{Lemma}
\newtheorem{proposition}{Proposition}

\theoremstyle{definition}

\newtheorem{remark}{Remark}
\newtheorem{assumption}{Assumption}

\usepackage[style=numeric, natbib, uniquelist=false, backend=biber, doi=false, url=false, maxcitenames = 2, date=year]{biblatex}
\usepackage[colorlinks=true, linkcolor=blue, urlcolor=cyan, citecolor=green]{hyperref}    

\usepackage{todonotes}

\usepackage{appendix}

\newcommand\Eb{\mathbb{E}}
\newcommand\Nb{\mathbb{N}}
\newcommand\Rb{\mathbb{R}}

\newcommand\Cc{\mathcal{C}}
\newcommand\Dc{\mathcal{D}}
\newcommand\Ec{\mathcal{E}}

\newcommand\Nc{\mathcal{N}}

\newcommand\Kc{\mathcal{K}}

\newcommand\Xc{\mathcal{X}}

\newcommand{\ee}{\mathrm{e}}
\newcommand{\BS}{\mathrm{BS}}
\newcommand{\VIX}{\mathrm{VIX}}
\newcommand{\Vega}{\mathrm{Vega}}
\newcommand{\Vomma}{\mathrm{Vomma}}

\newcommand{\He}{\mathrm{He}}

\newcommand{\ind}{\mathds{1}}
\newcommand{\VIXsqExp}{\Eb[\VIX_T^2]}
\newcommand{\Kbar}{\overline{K}}

\newcommand{\sbs}{\sigma_{\mathrm{BS}}}

\newcommand{\wtsp}{\widetilde{\sigma}_P}
\newcommand{\wtspj}{\widetilde{\sigma}_{P,j}}

\newcommand{\dd}{\mathop{}\!\mathrm{d}}
\newcommand{\ds}{\dd s}
\newcommand{\dt}{\dd t}
\newcommand{\du}{\dd u}
\newcommand{\dv}{\dd v}

\newcommand{\dy}{\dd y}

\newcommand{\dX}{\dd X}

\newcommand{\dW}{\dd W}

\newcommand{\ms}{\mathrm{ms}}
\newcommand{\s}{\mathrm{s}}

\DeclareMathOperator*{\argmin}{arg\,min}

\newcommand{\eqlnostar}[2]{\begin{align}\label{#1}#2\end{align}}
\newcommand{\eqstar}[1]{\begin{align*}#1\end{align*}}

\title{Implied Volatility Expansions for VIX Options in Forward Variance Models}
\author{Y.~Liao\thanks{University of Glasgow (\href{mailto:2636546L@student.gla.ac.uk}{2636546L@student.gla.ac.uk})}
\qquad
A.~Agarwal\thanks{Western University (\href{mailto:aagarw93@uwo.ca}{aagarw93@uwo.ca})}
\qquad
F.~Bourgey\thanks{NYU~(\href{mailto:fb2615@nyu.edu}{fb2615@nyu.edu})}
}
\date{}

\begin{document}
\maketitle

\begin{abstract}
\noindent 
We develop expansions for the implied volatility of VIX options within the class of forward variance models based on weak-approximation techniques for VIX option prices. In single-kernel models, such as the standard and rough Bergomi models, the expansion is fully explicit, and we prove that its error is of the same order as that of the underlying price approximation. In mixed-kernel models, we propose two expansions. The first one is fully explicit and relies on a Hermite expansion of the price, which we rewrite exactly as a Black--Scholes price with a shifted volatility, up to explicit residual terms. The second one uses the implied volatility of a mixed log-normal proxy, computed numerically, as its zeroth-order term, and we prove that it achieves the same order of accuracy as in the single-kernel case. The resulting formulas yield implied volatilities directly, rather than by inverting approximate option prices.

\end{abstract}

\section{Introduction}
\label{sec:introduction}
Accurate explicit implied volatility expansions allow industry practitioners to efficiently calibrate their models to market data. This operation is crucial for computing risk sensitivities of derivative trading books based on current market conditions. In this work, we develop implied volatility expansions for options on the S\&P~500 volatility index (VIX) within the class of forward variance models. Our implied volatility expansions are derived from VIX option price approximations developed by \cite{bourgey_weak_2023}.

The VIX was first introduced in 1993 by the Chicago Board Options Exchange (CBOE), originally as a weighted combination of Black--Scholes implied volatilities of near-the-money S\&P~100 options. In 2003, it was redesigned using S\&P~500 options across a wide range of strikes to measure the 30-day implied volatility of a log-contract on that index. In practice, the log-contract is replicated using market prices of listed vanilla options across a continuum of strikes. Following \cite{bergomi_smile_2004}, a natural way to model the joint dynamics of the S\&P 500 spot price $S_t$ and its implied volatility surface is to work directly with the term structure of implied variances of log-contracts, referred to as forward variances. As is well known, these forward variances are driftless under the risk-neutral measure and can be treated as additional state variables alongside $S_t$. The resulting models are referred to as \textit{forward variance models}. 

In this framework, the instantaneous forward variance curve $(\xi_t^u)_{t\le u}$ is defined by
\eqlnostar{eq:instantaneous_forward_variance}{
\xi_t^u = \frac{\dd}{\du} \Bigl((u-t) \, \hat{\sigma}_t (u)^2 \Bigr),
}
where $\hat{\sigma}_t (u)^2$ denotes the implied variance of the log-contract with maturity $u$ observed at time $t$ (equivalently, the fair strike of a variance swap with maturity $u$ at time $t$). In this work, we assume that the process $(\xi_t^u)_{t\le u}$ solves the following stochastic differential equation (SDE)
\eqlnostar{eq:instantaneous_forward_variance_sde}{
\dd\xi_t^u = \xi_t^u \, K^u(t) \dW_t,
}
where $K$ is a deterministic $L^2$-kernel function, and $\xi_0^u$ is a given initial forward variance curve. The solution to \eqref{eq:instantaneous_forward_variance_sde} is explicitly given by the log-normal process 
\eqstar{
\xi_t^u = \xi_0^u \exp \Bigl(- \frac{1}{2} \int_0^t K^u(s)^2 \ds + \int_0^t K^u(s) \dW_s\Bigr),
}
for all $u \geq t$. Equivalently, the instantaneous log-forward variance $X_t^u := \ln (\xi_t^u)$, with $u \geq t$, solves the SDE
\eqlnostar{eq:instantaneous_log-forward_variance}{
\dX_t^u = - \frac{1}{2} K^u(t)^2 \dt + K^u(t) \dW_t.
}

\begin{assumption}
\label{assump:initial_instantaneous_forward_variance}
The initial instantaneous forward variance curve $u \mapsto \xi_0^u$ is positive, bounded, and bounded away from zero.
\end{assumption}

\begin{assumption}
\label{assump:kernel_function}
The kernel $K$ in \eqref{eq:instantaneous_forward_variance_sde} satisfies $\int_0^T K^u(t)^2 \dt < \infty$ for every $u\in [T, T+1]$. Moreover, for any $p > 0$, there exists a positive constant $C_p$ such that
\eqlnostar{eq:integrable_assumption}{
\sup_{\Delta \in (0,1]} \frac{1}{\Delta} \int_T^{T+\Delta} \ee^{\,p \int_0^T K^u (t)^2 \dt} \du \le C_p.
}
\end{assumption}
\noindent
The upper bound $1$ for $\Delta$ in Assumption~\ref{assump:kernel_function} is an arbitrary normalization. In practice $\Delta = 30/365$, and the asymptotic results below concern the regime $\Delta \to 0$.
\noindent 
This class of models encompasses both the standard Bergomi model \citep{bergomi_smile_2005} and the rough Bergomi model \citep{bayer_pricing_2016}. In the standard Bergomi model, the kernel takes the exponential form
\eqlnostar{eq:standard_kernel}{
K^u (t) = \omega \, \ee^{- \kappa (u-t)},
}
where $\omega>0$ is the volatility of forward variances, and $\kappa>0$ is the mean-reversion speed. The rough Bergomi model corresponds to a fractional power-law kernel,
\eqlnostar{eq:power_kernel}{
K^u(t) = \eta \, (u-t)^{H - \frac{1}{2}},
}
where the Hurst exponent $H \in (0,\frac12)$ governs the degree of roughness (or memory decay), and $\eta > 0$ again represents the volatility of forward variances. Unlike \eqref{eq:standard_kernel}, the curve $u \mapsto \xi_t^u$ in the rough Bergomi model does not admit a finite-dimensional Markovian representation. Moreover, the mixed Bergomi models studied in Section~\ref{sec:mixed_kernel}, in which the forward variance is a convex combination of two log-normal components, do not satisfy \eqref{eq:instantaneous_forward_variance_sde}. The corresponding assumptions and results for mixed Bergomi  models are stated separately in Section~\ref{sec:mixed_kernel}.

\subsection{Our contributions}

Our contributions are fourfold. First, building on the weak VIX option price approximation of \cite{bourgey_weak_2023}, we derive an explicit implied volatility expansion for the standard and rough Bergomi models, together with an error bound of the same order as the underlying price expansion. 
This gives a direct implied volatility formula, rather than requiring numerical inversion of the price expansion. Second, we extend the construction to mixed-kernel forward variance models. We rewrite the mixed VIX option approximation by means of a Hermite expansion, and show that the resulting price is exactly a Black--Scholes price with a shifted volatility, up to explicit residual terms. Neglecting these residuals yields the fully explicit implied volatility expansion, explicit in the sense that it requires no BS inversion: its ingredients are one-dimensional integrals, computed by quadrature, and the solution of a scalar equation. Third, we derive a second mixed-kernel expansion, whose zeroth-order term is the implied volatility of the mixed log-normal proxy, computed numerically, and whose first-order correction is explicit and linear in the coefficients of the price expansion. We show that this expansion has an error of the same order as in the single-kernel case. These results complement existing Bergomi and rough-volatility asymptotic results for VIX futures, power payoffs, and volatility derivatives \citep{jacquier2018vix, lacombe2021asymptotics, guyon_vix_2022, alos_smile_2022}, as well as the asymptotic results for VIX implied volatilities in rough multi-factor volatility models obtained in \cite{jacquier2025rough}. Fourth, numerical tests in the standard, rough and mixed Bergomi models show that the explicit expansion is accurate and, in single-kernel models, about four times faster than numerical inversion of the weak price approximation. Mixed specifications produce positively skewed VIX smiles, while single-kernel smiles are nearly flat, a phenomenon consistent with the limitations documented in \cite{horvath_volatility_2020}. The accuracy of the fully explicit mixed-kernel expansion deteriorates when the two proxy volatilities are far apart. In these regimes, the expansion with a numerically computed zeroth-order term remains as accurate as the weak price approximation itself.

\subsection{Notations and definitions}
\label{sec:notations_and_definitions}

The Black--Scholes (BS) call and put price formulas, as functions of the log-spot $x \in \Rb$, log-strike $k \in \Rb$, and volatility $\sigma > 0$, are defined by
\begin{align}
C^{\BS} (x, k, \sigma) 
&:= \Eb \Bigl[ \bigl( \ee^x \ee^{-\frac{\sigma^2}{2}T +\sigma \sqrt{T} Z} - \ee^{k} \bigr)^{+} \Bigr]
= \ee^x \, \Phi (d_1(x, k, \sigma)) - \ee^k \, \Phi (d_2(x, k, \sigma)),
\label{eq:bs_call} \\
P^{\BS} (x, k, \sigma) 
&:= \Eb \Bigl[ \bigl( \ee^{k} - \ee^x \ee^{-\frac{\sigma^2}{2}T +\sigma \sqrt{T} Z} \bigr)^{+} \Bigr]
= \ee^k \Phi (-d_2(x, k, \sigma)) - \ee^x \Phi (-d_1(x, k, \sigma)),
\label{eq:bs_put}
\end{align}
where $Z \sim \Nc(0,1)$ and
\eqlnostar{eq:d_bs}{
d_1(x, k, \sigma) := \frac{x-k}{\sigma \sqrt{T}} + \frac{\sigma \sqrt{T}}{2}, 
\qquad 
d_2 (x, k, \sigma) := d_1 (x, k, \sigma) - \sigma \sqrt{T}.
}
Here, $\Phi(x) := \int_{-\infty}^{x} \phi(t) \dt$ denotes the cumulative distribution function (CDF) of the standard normal random variable, with probability density function (PDF) $\phi(t) := \ee^{-t^2/2} / \sqrt{2\pi}$.

The first- and second-order partial derivatives of $C^\BS$ with respect to the volatility $\sigma$ are given by (see, e.g., \cite[Proposition 32]{bompis_stochastic_2013})
\begin{align}
    \Vega (x, k, \sigma) 
    &:= \partial_\sigma C^{\BS} (x, k, \sigma) 
    = \ee^x \sqrt{T} \phi \bigl(d_1 (x, k, \sigma)\bigr),
    \label{eq:vega}
    \\
    \Vomma (x, k, \sigma) &:= \partial_\sigma^2 C^{\BS} (x, k, \sigma) =  \ee^x \sqrt{T} \phi (d_1 (x, k, \sigma)) \biggl( \frac{(x-k)^2}{\sigma^3 T} - \frac{\sigma T}{4} \biggr).
    \label{eq:vomma}
\end{align}
From \cite[Proposition 34]{bompis_stochastic_2013}, the following useful identities involving spatial derivatives hold
\eqlnostar{eq:identities_bs_price}{
( \partial_x^2 -  \partial_x ) C^\BS (x, k, \sigma) &= \frac{\Vega (x, k, \sigma)}{\sigma T}, \nonumber\\
( \partial_x^3 - \partial_x^2 ) C^\BS (x, k, \sigma) &= \Vega (x, k, \sigma) \Bigl(-\frac{x-k}{\sigma^3 T^2} + \frac{1}{2 \sigma T} \Bigr). 
}

\paragraph{VIX definition.} The VIX index at time $T$ is given by

\eqlnostar{eq:vix_single_kernel}{
\VIX_T := \sqrt{\frac{1}{\Delta} \int_T^{T+\Delta} \xi_T^u \du} 
= \sqrt{\frac{1}{\Delta} \int_T^{T+\Delta} \xi_0^u \, \ee^{Y_T^u} \du}, 
}
where $\Delta = 30/365$ corresponds to 30 days, and 
\eqstar{
Y_T^u := X_T^u - X_0^u 
= - \frac{1}{2} \int_0^T K^u(t)^2 \dt + \int_0^T K^u(t) \dW_t,
}
denotes the increment of the log-forward variance over the interval $[0,T]$.

\paragraph{VIX option prices.}
We assume zero interest rates and dividends. The VIX call, put, and futures prices for fixed strike and maturity, $\ee^k$ and $T$, respectively, are defined as
\eqstar{
C := \Eb \Bigl[ \bigl(\VIX_{T}-\ee^k \bigr)^{+} \Bigr], \qquad
P := \Eb \Bigl[ \bigl(\ee^k-\VIX_{T} \bigr)^{+} \Bigr], \qquad
F := \Eb[\VIX_{T}].
}

\paragraph{VIX implied volatility.}
The implied volatility of the VIX call price $C$ with log-strike $k$ and maturity $T$, denoted by $\sbs (k, T)$, is the unique positive solution of
\eqlnostar{eq:vix_iv}{
C^{\BS} (\ln(F), k, \sbs (k, T)) = C.
}

\paragraph{Hermite polynomials.}
Let $(\He_j)_{j\in \Nb}$ denote the probabilists' Hermite polynomials, defined by
\eqlnostar{eq:probabilist's_hermite_polynomial}{
\He_j (x) := (-1)^j \ee^{\frac{x^2}{2}} \partial_x^{j} \, \Bigl(\ee^{-\frac{x^2}{2}} \Bigr),  
\qquad 
\forall x \in \Rb.
}

The rest of the paper is organized as follows. 
In Sections \ref{sec:price_approx_single_kernel} and \ref{sec:price_approx_mixed_kernel}, we provide approximations for VIX option prices under single- and mixed-kernel forward variance models, respectively. 
In Section \ref{sec:vix_iv_proxy_single_kernel}, we present explicit implied volatility expansions for VIX options and assess their accuracy through numerical experiments in Section \ref{sec:numerical_tests_single_kernel} for both the standard and the rough Bergomi models. In Section \ref{sec:vix_iv_proxy_mixed_kernel}, we extend these expansions to mixed-kernel forward variance models and report the corresponding numerical results in Section \ref{sec:numerical_tests_mixed_kernel}. In Section \ref{sec:vix_iv_proxy_mixed_kernel_rigorous}, we derive a second mixed-kernel expansion, with a numerically computed zeroth-order term and a proven error bound, and test it numerically in Section \ref{sec:rigorous_mixed_numerics}. Finally, we conclude in Section \ref{sec:conclusion}. 

A GitHub repository reproducing the numerical results is available at \url{https://github.com/fbourgey/vix-expansion-bergomi}.

\section{Single-kernel forward variance models}
\label{sec:single_kernel}

In this section, we outline the key steps for deriving approximations of VIX option prices within the single-kernel forward variance models introduced in \cite{bourgey_weak_2023}. The same steps are used later in Section \ref{sec:price_approx_mixed_kernel} to derive price approximations in mixed-kernel forward variance models. These results form the basis for the subsequent expansions of VIX implied volatility under both model specifications.

\subsection{VIX price approximation}
\label{sec:price_approx_single_kernel}

Motivated by the idea of approximating an arithmetic mean by its geometric mean (see \cite[Section 2.1]{bourgey_weak_2023} for the intuition), we introduce a log-normal proxy for the squared VIX, defined by
\eqlnostar{eq:vix_proxy_single_kernel}{
\VIX_{T,P}^2 := \VIXsqExp \exp \, \biggl(
\frac{1}{\Delta} \int_{T}^{T+\Delta} \frac{\xi_0^u}{\VIXsqExp} Y_T^u \du
\biggr),
}
where $\VIXsqExp = \frac{1}{\Delta} \int_{T}^{T+\Delta} \xi_0^u \du$.

\begin{proposition}{(\cite[Proposition 2.3]{bourgey_weak_2023})}
\label{prop:proxy_vix}
The proxy $\VIX_{T,P}^2$ is log-normal, that is, $\ln ( \VIX_{T,P}^2 ) \sim \Nc (\mu_P, \sigma_P^2)$, where the mean and variance are given by
\eqlnostar{eq:mean}{
\mu_P := \ln (\VIXsqExp) - \frac{1}{2} \int_0^T \biggl(\frac{1}{\Delta} \int_{T}^{T+\Delta} \frac{\xi_0^u}{\VIXsqExp} K^u(t)^2 \du \biggr) \dt,
}
and
\eqlnostar{eq:variance}{
\sigma_P^2 := \int_0^T \biggl( \frac{1}{\Delta} \int_{T}^{T+\Delta} \frac{\xi_0^u}{\VIXsqExp} K^u(t) \du \biggr)^2 \dt.
}
\end{proposition}
We require the proxy $\VIX_{T,P}^2$ to be non-degenerate for which we make the following assumption.
\begin{assumption}
\label{assump:no_degenerate}
There exist $a_1, a_2 > 0$ such that, for any $p>0$, there exist positive constants $c$ and $\Delta_0 \le 1$ such that, for all $\Delta \leq \Delta_0$,
{\small
\begin{align}
\Gamma_{\Delta, T, p} &:= \biggl( \frac{1}{\Delta} \int_{T}^{T+\Delta} \frac{\xi_0^u}{\VIXsqExp} \bigg| 
\int_0^T \Bigl( K^u(t)^2 -\frac{1}{\Delta} \int_T^{T+\Delta} \frac{\xi_0^v}{\VIXsqExp} K^v(t)^2 \dv \Bigr)\dt 
\bigg|^p \du \biggr)^{\frac{1}{p}} 
\le c \Delta^{a_1}, \label{eq:gamma}\\
\Lambda_{\Delta, T, p} &:= \biggl( \frac{1}{\Delta} \int_{T}^{T+\Delta} \frac{\xi_0^u}{\VIXsqExp} \bigg|
\int_0^T \Bigl( K^u(t) - \frac{1}{\Delta} \int_T^{T+\Delta} \frac{\xi_0^v}{\VIXsqExp} K^v(t) \dv \Bigr)^2\dt
\bigg|^p \du \biggr)^{\frac{1}{p}}
\le c \Delta^{a_2}. \label{eq:delta}
\end{align}
}
Additionally, we assume
\eqstar{
\sup_{\Delta \le \Delta_0} |\mu_P| \le c, 
\qquad
\frac{1}{c} \le \inf_{\Delta \le \Delta_0} \sigma_P \le \sup_{\Delta \le \Delta_0} \sigma_P \le c.
}
\end{assumption}

\begin{remark}
Appendix~\ref{sec:sufficient_conditions} verifies Assumption~\ref{assump:no_degenerate} for the kernels \eqref{eq:standard_kernel} and \eqref{eq:power_kernel}, with $(a_1,a_2)=(1,2)$ and $(2H,2H)$, respectively.
\end{remark}

In view of the proxy, the VIX option price $\Eb [\varphi(\VIX_T^2)]$ is approximated by a Taylor expansion around $\VIX_{T,P}^2$, under the assumption that the payoff function $\varphi$ is sufficiently smooth. In what follows, we present an approximation formula for VIX option prices within the single-kernel forward variance models.

\begin{theorem}{(\cite[Theorem 2.7]{bourgey_weak_2023})}
\label{thm:vix_opt_price_approx_single_kernel}
Let $\varphi: \Rb \to \Rb$ be a $\theta$-H\"older continuous function for some $\theta \in (0,1]$. Under Assumptions \ref{assump:initial_instantaneous_forward_variance}, \ref{assump:kernel_function}, and \ref{assump:no_degenerate}, the price of a VIX option with payoff $\varphi(\VIX_T^2)$ admits the following approximation
\eqlnostar{eq:vix_opt_price_approx_single_kernel}{
\Eb \bigl[ \varphi (\VIX_T^2) \bigr] = \Eb \bigl[ \varphi (\VIX_{T,P}^2) \bigr] 
+ \sum_{i=1}^3 \gamma_i \, \partial_\epsilon^i \, \Eb \bigl[ \varphi \bigl( \VIX_{T,P}^2 \, \ee^\epsilon \bigr) \bigr] \big|_{\epsilon=0} + \Ec_\varphi,
}
where the remainder term $\Ec_\varphi$ satisfies the bound $|\Ec_\varphi| \le C \Delta^{3(a_1 \land \frac{a_2}{2})}$, with $C$ depending on $\varphi$ only through $\theta$ and its $\theta$-H\"older constant. The coefficients $(\gamma_i)_{i\in \{1,2,3\}}$ are given explicitly by
{\small
\eqlnostar{eq:coefficients}{
\gamma_1 &:= \frac{1}{8\Delta} \int_T^{T+\Delta} \frac{\xi_0^u}{\VIXsqExp} 
\Biggl( \int_0^T 
\Biggl( K^u(t)^2 - \frac{1}{\Delta} \int_T^{T+\Delta} \frac{\xi_0^v}{\VIXsqExp}  K^v(t)^2 \dv
\Biggr)\dt
\Biggr)^2 \du \nonumber \\
&\quad + \frac{1}{2\Delta} \int_T^{T+\Delta} \frac{\xi_0^u}{\VIXsqExp} \int_0^T 
\Biggl( K^u(t) - \frac{1}{\Delta} \int_T^{T+\Delta} \frac{\xi_0^v}{\VIXsqExp} K^v(t) \dv \Biggr)^2 \dt\,
\du, \nonumber \\[1.2ex] 
\gamma_2 &:= - \frac{1}{2\Delta} \int_T^{T+\Delta} \frac{\xi_0^u}{\VIXsqExp}
\Biggl( \int_0^T
\Biggl( K^u(t)^2 - \frac{1}{\Delta} \int_T^{T+\Delta} \frac{\xi_0^v}{\VIXsqExp} K^v(t)^2 \dv
\Biggr)\dt \Biggr) \\
&\quad \times
\Biggl( \int_0^T \frac{1}{\Delta} \int_T^{T+\Delta} \frac{\xi_0^v}{\VIXsqExp} K^v(t) \dv
\Biggl( K^u(t) - \frac{1}{\Delta} \int_T^{T+\Delta} \frac{\xi_0^v}{\VIXsqExp}  K^v(t) \dv
\Biggr)\dt
\Biggr) \du, \\[1.2ex]
\gamma_3 &:= \frac{1}{2\Delta}\! \int_T^{T+\Delta}\! \frac{\xi_0^u}{\VIXsqExp}
\Biggl(\! \int_0^T\! \frac{1}{\Delta}\! \int_T^{T+\Delta}\! \frac{\xi_0^v}{\VIXsqExp}  K^v(t)\! \dv
\Biggl(\! K^u(t) - \frac{1}{\Delta}\! \int_T^{T+\Delta}\! \frac{\xi_0^v}{\VIXsqExp} K^v(t) \dv \!
\Biggr)\! \dt \!
\Biggr)^2\! \du.
}
}
\end{theorem}

\begin{remark}
\label{re:payoff_func}
The payoff of a VIX call with log-strike $k$ is $\varphi(x) = (\sqrt{x^+} - \ee^k)^+$, a VIX put is $\varphi(x) = (\ee^k - \sqrt{x^+})^+$, and a VIX futures contract is $\varphi(x) = \sqrt{x^+}$, where $x^+ := \max(x,0)$ extends the payoffs to $\Rb$. In all these cases, the function $\varphi$ is $\frac{1}{2}$-H\"older continuous with a H\"older constant of $1$, which is independent of $k$.
\end{remark}

The following elementary bounds on the coefficients $(\gamma_i)_{i\in\{1,2,3\}}$ clarify the hierarchy of the terms in \eqref{eq:vix_opt_price_approx_single_kernel}
\begin{lemma}
\label{lem:order_of_coefficients}
Under Assumption \ref{assump:no_degenerate}, let $\varepsilon := a_1 \wedge \frac{a_2}{2}$. Then, for all $\Delta \le \Delta_0 \leq 1$,
\eqlnostar{eq:order_of_coefficients}{
|\gamma_1| \le \frac{c^2}{8}\Delta^{2a_1} + \frac{c}{2}\Delta^{a_2}, 
\qquad 
|\gamma_2| \le \frac{c^{5/2}}{2}\Delta^{a_1 + \frac{a_2}{2}},
\qquad
|\gamma_3| \le \frac{c^{3}}{2}\Delta^{a_2}.
}
In particular, there exists a constant $C_\gamma > 0$ such that $\max_{i\in\{1,2,3\}} |\gamma_i| \le C_\gamma \Delta^{2\varepsilon}$ for all $\Delta \le \Delta_0$, whereas, under the assumptions of Theorem~\ref{thm:vix_opt_price_approx_single_kernel}, the remainder in \eqref{eq:vix_opt_price_approx_single_kernel} satisfies $|\Ec_\varphi| \le C \Delta^{3\varepsilon}$.
\end{lemma}
\noindent The proof is in Appendix \ref{sec:proof of lem:order_of_coefficients}.

\paragraph{VIX option prices proxy.}
The proxy prices of a VIX call, put, and futures are defined as 
\eqstar{
C_{P,0} := \Eb \Bigl[ \bigl(\VIX_{T,P} - \ee^{k} \bigr)^{+} \Bigr], \qquad
P_{P,0} := \Eb \Bigl[ \bigl(\ee^{k} - \VIX_{T,P} \bigr)^{+} \Bigr], \qquad
F_{P,0} := \Eb[\VIX_{T,P}],
}
and the corrected proxy prices as
\eqstar{
C_P &:= C_{P,0}
+ \sum_{i=1}^3 \gamma_i \, \partial_\epsilon^i \Eb \Bigl[ \Bigl(\VIX_{T,P}\, \ee^{\frac{\epsilon}{2}} - \ee^k \Bigr)^+ \Bigr] \Big|_{\epsilon=0}, \\
P_P &:= P_{P,0}
+ \sum_{i=1}^3 \gamma_i \, \partial_\epsilon^i \Eb \Bigl[ \Bigl(\ee^k - \VIX_{T,P}\, \ee^{\frac{\epsilon}{2}} \Bigr)^+ \Bigr] \Big|_{\epsilon=0}, \\
F_P &:= F_{P,0} + \sum_{i=1}^3 \gamma_i \, \partial_\epsilon^i \Eb \Bigl[\VIX_{T,P}\, \ee^{\frac{\epsilon}{2}} \Bigr] \Big|_{\epsilon=0}.
}

\begin{remark}
\label{re:explicit_formula_for_coefficients}
In the standard Bergomi model, there are closed-form expressions for $\mu_P$, $\sigma_P$, and $(\gamma_i)_{i \in \{1,2,3\}}$, if the initial forward variance curve $u \mapsto \xi_0^u$ is constant on $u\in [T, T+\Delta]$. In contrast, in the rough Bergomi model, explicit expressions exist for $\mu_P$ and $\sigma_P$, while $(\gamma_i)_{i\in\{1,2,3\}}$ do not admit an explicit form even when the forward variance curve is flat. See \cite[Propositions 2.9 and 2.11]{bourgey_weak_2023} for details.
\end{remark}

\subsection{Implied volatility expansion}
\label{sec:vix_iv_proxy_single_kernel}

In view of Remark \ref{re:payoff_func}, the approximations for VIX call, put, and futures prices in \eqref{eq:vix_opt_price_approx_single_kernel} can be expressed as linear combinations of the BS price formulas \eqref{eq:bs_call}-\eqref{eq:bs_put} and their derivatives with respect to the log-spot price. Let $\wtsp := \frac{\sigma_P}{\sqrt{T}}$ and define the effective log-spot $x_P := \frac{\mu_P}{2} + \frac{\sigma_P^2}{8}$. Consider a call payoff $\varphi(x) = (\sqrt{x}-\ee^k)^+$.
From Proposition~\ref{prop:proxy_vix}, we have
\eqstar{
C_{P,0} = \Eb \bigl[ \varphi \bigl( \VIX_{T,P}^2 \bigr) \bigr] = C^{\BS} \bigl(x_P, k, \tfrac{\wtsp}{2} \bigr).
}
The $i$th-order derivative, $i \in \{1,2,3\},$ in \eqref{eq:vix_opt_price_approx_single_kernel} is given by
\eqstar{
\partial_\epsilon^i \Eb \bigl[ \varphi \bigl( \VIX_{T,P}^2 \, \ee^\epsilon \bigr) \bigr] \big|_{\epsilon=0}
= \partial_\epsilon^i C^{\BS} \bigl(x_P + \tfrac{\epsilon}{2}, k, \tfrac{\wtsp}{2} \bigr) \big|_{\epsilon = 0} 
= 2^{-i} \partial_x^i C^{\BS} \bigl(x_P, k, \tfrac{\wtsp}{2} \bigr), 
}
where
\eqlnostar{eq:bs_call_derivatives}{
\partial_x^i C^{\BS}\! (x, k, \sigma) 
= 
\ee^x \Phi (d_1(x, k, \sigma)) 
+ \ind_{i\ge 2}\, \ee^x \phi (d_1(x, k, \sigma)) \sum_{j=1}^{i-1} \binom{i-1}{j}\, (-1)^{j-1} \frac{\He_{j-1} (d_1(x, k, \sigma))}{\sigma^j T^{\frac{j}{2}}}.
}
Thus, the VIX call price approximation in \eqref{eq:vix_opt_price_approx_single_kernel} rewrites as
\eqlnostar{eq:vix_call_proxy_single_kernel}{
C_P = 
C^{\BS}\bigl(x_P, k, \tfrac{\wtsp}{2}\bigr)
+ \sum_{i=1}^3 2^{-i} \gamma_i\, \partial_x^i C^{\BS} \bigl(x_P, k, \tfrac{\wtsp}{2}\bigr).
}
We can analogously rewrite the price approximation in \eqref{eq:vix_opt_price_approx_single_kernel} for VIX puts, using the combination of the BS put price formula in \eqref{eq:bs_put} and its derivatives. The VIX put price approximation formula is then given by
\eqlnostar{eq:vix_put_proxy_single_kernel}{
P_P = 
P^{\BS} \bigl(x_P, k, \tfrac{\wtsp}{2}\bigr)
+ \sum_{i=1}^3 2^{-i} \gamma_i\, \partial_x^i P^{\BS} \bigl(x_P, k, \tfrac{\wtsp}{2}\bigr),
}
where
\eqlnostar{eq:bs_put_derivatives}{
\partial_x^i P^{\BS}\! (x, k, \sigma) 
= \!
-\ee^x \Phi\! (-d_1(x, k, \sigma))\! 
+\! \ind_{i\ge 2} \ee^x \phi\! (d_1(x, k, \sigma))\! \sum_{j=1}^{i-1} \binom{i-1}{j}\! (-1)^{j-1} \frac{\He_{j-1} (d_1(x, k, \sigma))}{\sigma^j T^{\frac{j}{2}}}.
}
For VIX futures prices, applying the same Taylor expansion to the payoff $\sqrt{x}$ and dropping the strike contribution yields
\eqlnostar{eq:vix_fut_proxy_single_kernel}{
F_P = \ee^{x_P} \Bigl( 1 + \frac{\gamma_1}{2} + \frac{\gamma_2}{4} + \frac{\gamma_3}{8} \Bigr) = \ee^{x_P} (1 + \delta),
\qquad
\delta := \frac{\gamma_1}{2} + \frac{\gamma_2}{4} + \frac{\gamma_3}{8}.
}
In the above, the put-call parity holds, that is, $C_P - P_P = F_P - \ee^k$. Consequently, when the futures price \eqref{eq:vix_fut_proxy_single_kernel} is used as the spot, the implied volatility extracted from VIX call prices \eqref{eq:vix_call_proxy_single_kernel} coincides with that extracted from VIX put prices \eqref{eq:vix_put_proxy_single_kernel}.

Inspired by \cite{bompis_analytical_2018}, we derive an approximation for the volatility implied by VIX call prices in \eqref{eq:vix_call_proxy_single_kernel} by using relations between BS-Greeks and sensitivities of VIX option prices with respect to log-spot $x$ and volatility $\sigma$ in \eqref{eq:identities_bs_price}. 

\begin{theorem}
\label{thm:vix_iv_single_kernel}
Let Assumptions \ref{assump:initial_instantaneous_forward_variance}, \ref{assump:kernel_function} and \ref{assump:no_degenerate} hold, let $\varepsilon := a_1 \wedge \frac{a_2}{2}$, and define the \emph{VIX implied volatility expansion}
\eqlnostar{eq:vix_iv_proxy_single_kernel}{
\widehat{\sigma}_P(k,T) := \frac{\wtsp}{2} 
+ \frac{\gamma_2}{2 \wtsp T} 
+ \frac{3 \gamma_3}{8 \wtsp T} 
- \frac{\gamma_3}{\wtsp^3 T^2} (x_P-k), 
}
where $\wtsp = \frac{\sigma_P}{\sqrt{T}}$, $x_P = \frac{\mu_P}{2} + \frac{\sigma_P^2}{8}$, and the coefficients $(\gamma_i)_{i\in\{1,2,3\}}$ are given in \eqref{eq:coefficients}. Then, for every compact set $\Kc \subset \Rb$ of log-strikes, there exist constants $C>0$ and $0 < \Delta_1 \leq \Delta_0$, depending only on $\Kc$, $T$ and the constants in Assumptions \ref{assump:initial_instantaneous_forward_variance}--\ref{assump:no_degenerate}, such that, for all $\Delta \le \Delta_1$ and all $k \in \Kc$,
\eqlnostar{eq:vix_iv_error_bound_single_kernel}{
\bigl| \sbs(k,T) - \widehat{\sigma}_P(k,T) \bigr| \le C \Delta^{3\varepsilon}.
}
\end{theorem}
\noindent The proof is provided in Appendix~\ref{sec:proof of thm:vix_iv_single_kernel}.

\begin{remark}
\label{re:structure_of_error}
In the proof we show that the expansion $\widehat{\sigma}_P$ is the exact BS implied volatility, with spot $F_P$, of the corrected proxy price $C_P$ up to a term of order $\Delta^{4\varepsilon}$, which is of higher order than the remainder $\Ec_\varphi$ of the price expansion itself. The dominant contribution to the error \eqref{eq:vix_iv_error_bound_single_kernel} therefore comes from the price expansion \eqref{eq:vix_opt_price_approx_single_kernel} and not from the derivation of the implied volatility formula. 
\end{remark}

\subsection{Numerical tests}
\label{sec:numerical_tests_single_kernel}

In this section, we evaluate the accuracy of the VIX implied volatility expansion \eqref{eq:vix_iv_proxy_single_kernel} (referred to as \textit{Expansion}) under the standard and rough Bergomi models, and compare it with implied volatilities obtained from the weak approximation in Theorem \ref{thm:vix_opt_price_approx_single_kernel} (referred to as \textit{Approx}). In the standard Bergomi model, reference values use $256$ Gauss--Legendre nodes over the delivery dates and an adaptive Gauss--Kronrod quadrature (QUADPACK, via \texttt{scipy.integrate.quad}) in the Gaussian factor (\textit{quadrature}). In the rough Bergomi model, we use randomized quasi-Monte Carlo with the substitution $u = T + \Delta v^4$, $16$ independent scrambled Sobol sequences of $2^{18}$ paths, and the geometric proxy as a control variate (\textit{randomized QMC}); shaded bands in the figures are $95\%$ confidence intervals across scrambles. The coefficients $\mu_P$, $\sigma_P$, and $(\gamma_i)_{i\in\{1,2,3\}}$ use $30$ Gauss--Legendre nodes under the substitution $t = T(1-(1-s)^4)$, $u = T + \Delta s^4$, $s \in [0,1]$, which resolves the kernel singularity at $u = t = T$.

We consider VIX maturities $T \in \{1, 3, 6\ \text{months}\}$ with $10$ evenly spaced log-moneyness $\ln(\ee^k/F)$ in $[-0.1, 0.4]$, and assume a flat forward variance $\xi_0^u = 0.24^2$. Figure \ref{fig:single_case_bergomi} shows VIX smiles in the standard Bergomi model for $\omega=2$, $ \kappa=0.25$, and $\omega=8$, $\kappa=10$, with parameters from \cite[Figure 3]{guyon_vix_2022}.
Figure \ref{fig:single_case_rough_bergomi} displays the rough Bergomi model with $H=0.1$, $\eta=1$, from \cite[Section 2]{bourgey_weak_2023}, and $H=0.23$, $\eta=1.02$, from \cite{bourgey2024smile}\footnote{Note that $\eta\sqrt{2H}$ is used instead of $\eta$ in \cite{bourgey2024smile}.}.
Relative errors are defined as $(\sigma^{\mathrm{approx}} - \sigma^{\mathrm{ref}})/\sigma^{\mathrm{ref}}$. In all four cases, the implied volatility expansion performs well, with relative errors below $0.6\%$, and of the same order as those of \textit{Approx}, in agreement with Remark~\ref{re:structure_of_error}. The maximum absolute relative errors of \textit{Expansion} are $2.8\times10^{-7}\%$ and $0.521\%$ in the two standard cases, and $0.424\%$ and $0.034\%$ in the two rough cases. Those of \textit{Approx} are $3.0\times10^{-7}\%$, $0.584\%$, $0.453\%$, and $0.036\%$, respectively. In the rough cases, the $95\%$ confidence half-widths of the reference are below $4\times10^{-4}\%$.

\textit{Expansion} has a slightly smaller absolute error than \textit{Approx} at $101$ of the $120$ points. The two methods differ by the terms which are of second order in $(\delta, \Delta\sigma)$ as shown in the proof of Theorem~\ref{thm:vix_iv_single_kernel}. In these examples, these terms turn out to be of the same magnitude as, and of opposite sign to, the error $\Ec_\varphi$ of the weak price approximation, so that the two errors partially cancel. Since the implied volatilities of the two methods differ by at most about $6$ bps in relative terms, the ordering between the two methods is not significant. Moreover, both \textit{Approx} and \textit{Expansion} underestimate the reference implied volatility at all strikes and maturities in the four cases. Since the remainder $\Ec_\varphi$ is dominated by the first neglected term of the expansion, and since the smiles are nearly flat, this term has essentially the same sign for all strikes.

Runtimes are medians of $51$ runs ($5$ for the references) in Python~3.14 (NumPy~2.4, SciPy~1.17) on a MacBook Air (Apple M5, $24$\,GB). For ten-strike smiles at $T = 3$ months, the references take $10.05$--$10.36\,\ms$ (\textit{quadrature}) and $8.19$--$8.20\,\s$ (\textit{randomized QMC}), against $0.24$--$0.25\,\ms$ for \textit{Approx} and $0.059$--$0.063\,\ms$ for \textit{Expansion}, which is about four times faster.

\begin{figure}[H]
    \centering
    \hspace{-0.5cm}
    \includegraphics[width=0.5\linewidth]{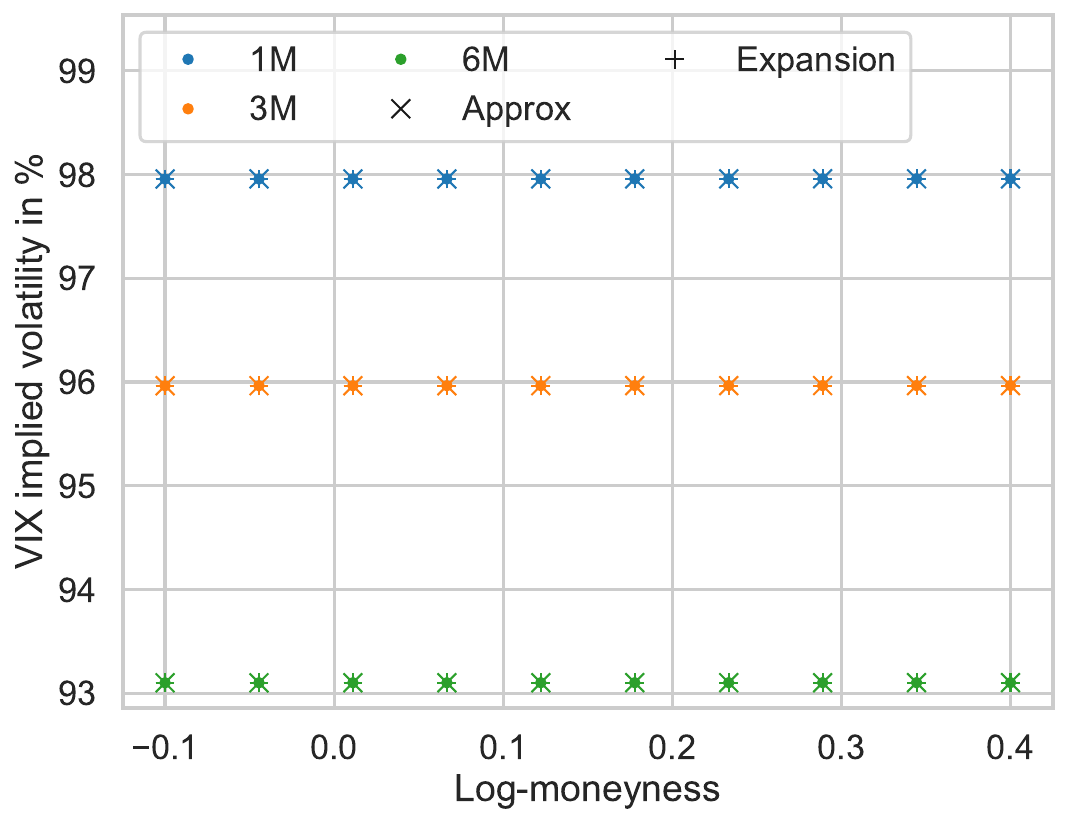}
    \includegraphics[width=0.5\linewidth]{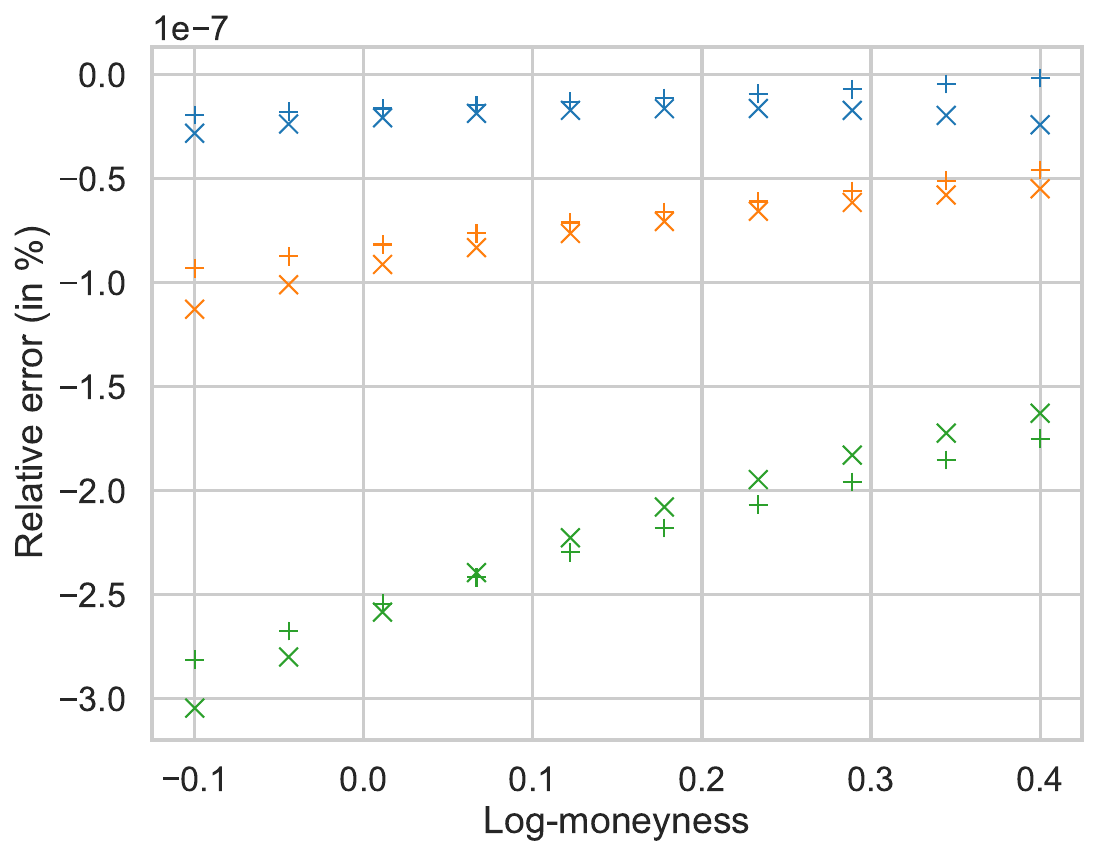}\\
    \hspace{-0.5cm}
    \includegraphics[width=0.5\linewidth]{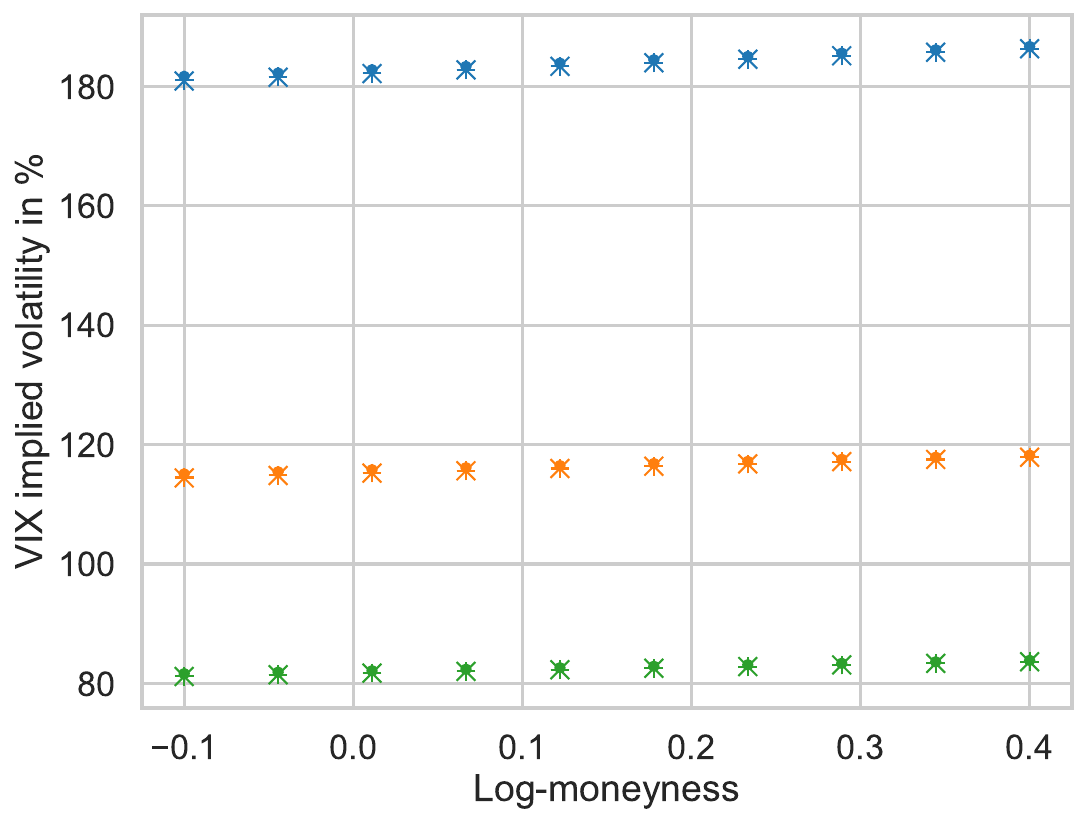}
    \includegraphics[width=0.5\linewidth]{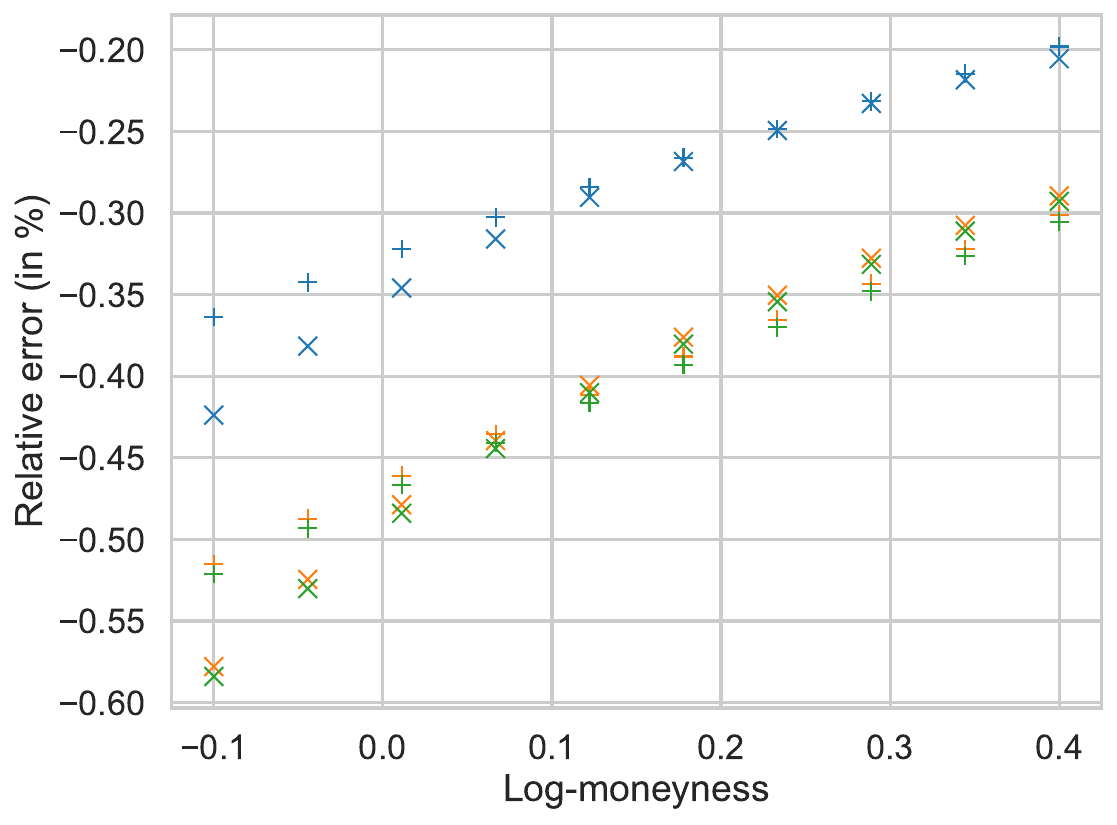}
    \caption{VIX smiles in the standard Bergomi model. Top: $\omega=2.0$, $\kappa=0.25$. Bottom: $\omega=8.0$, $\kappa=10.0$. Right: corresponding relative errors. Reference VIX futures $F$ for the 1-, 3-, and 6-month maturities: 0.2306, 0.2139, 0.1932 (top) and 0.2073, 0.2009, 0.2007 (bottom).}
    \label{fig:single_case_bergomi}
\end{figure}

\begin{figure}[H]
    \centering
    \hspace{-0.5cm}
    \includegraphics[width=0.5\linewidth]{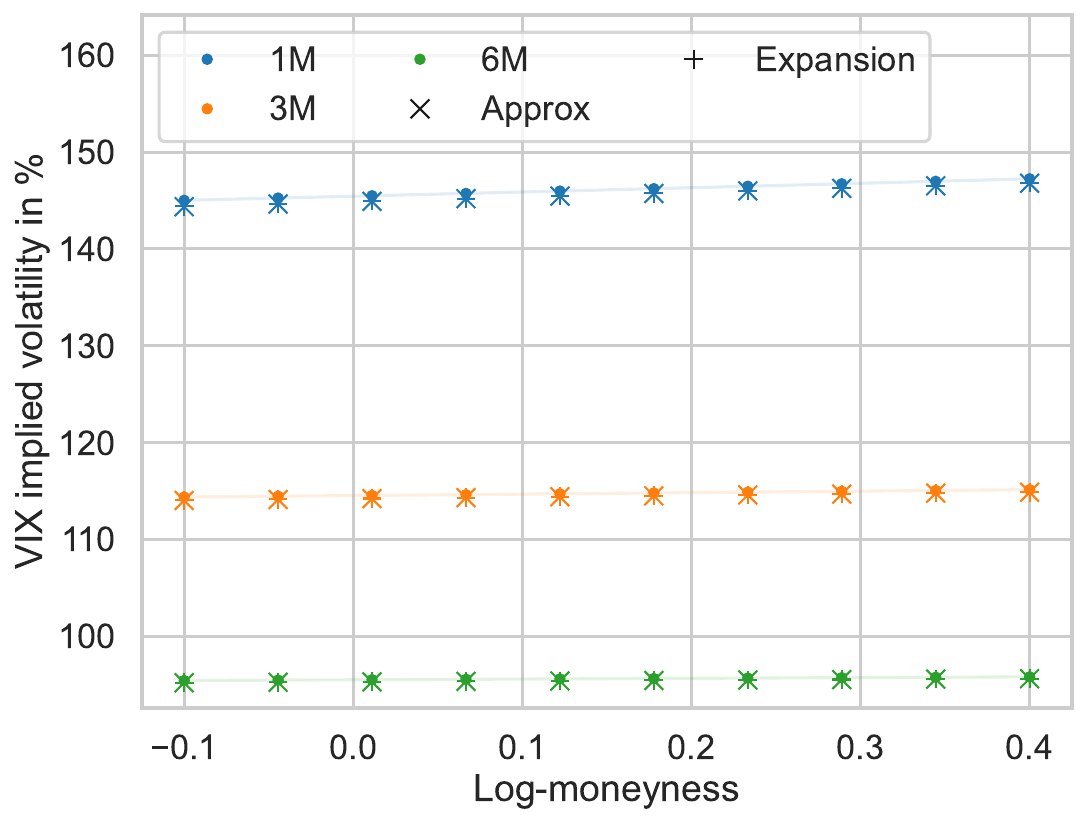}
    \includegraphics[width=0.5\linewidth]{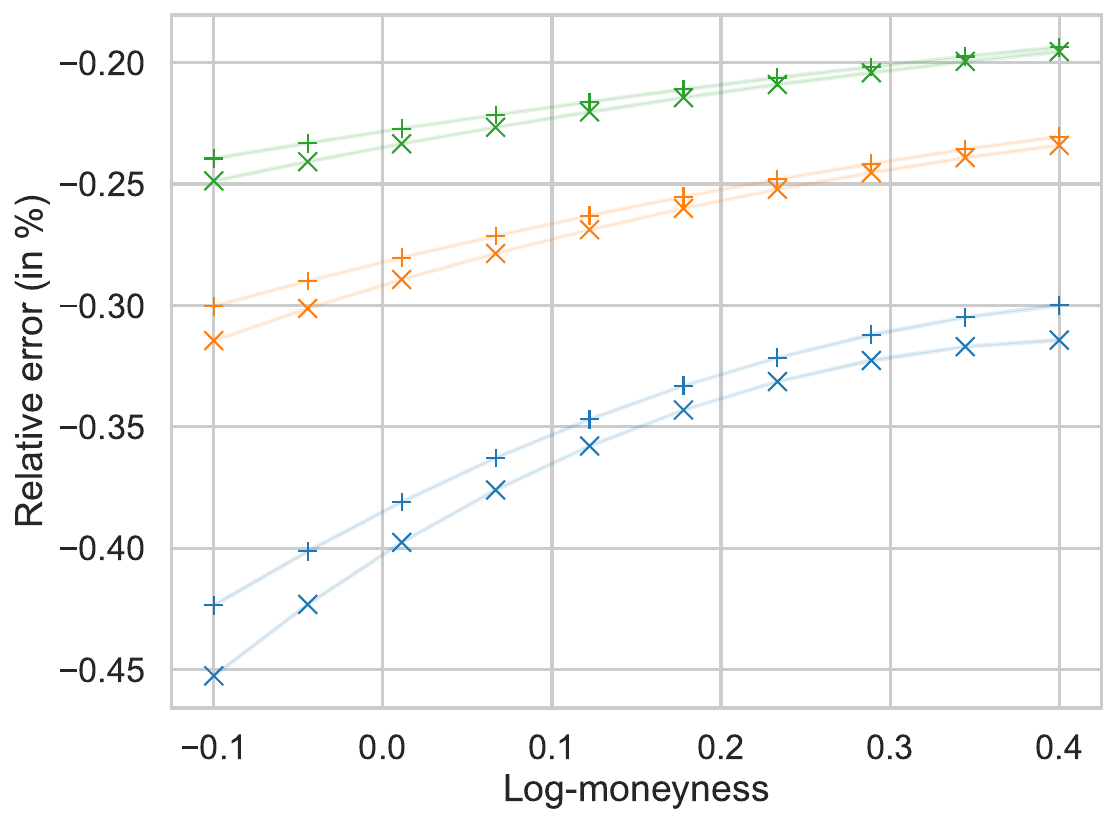}\\
    \hspace{-0.5cm}
    \includegraphics[width=0.5\linewidth]{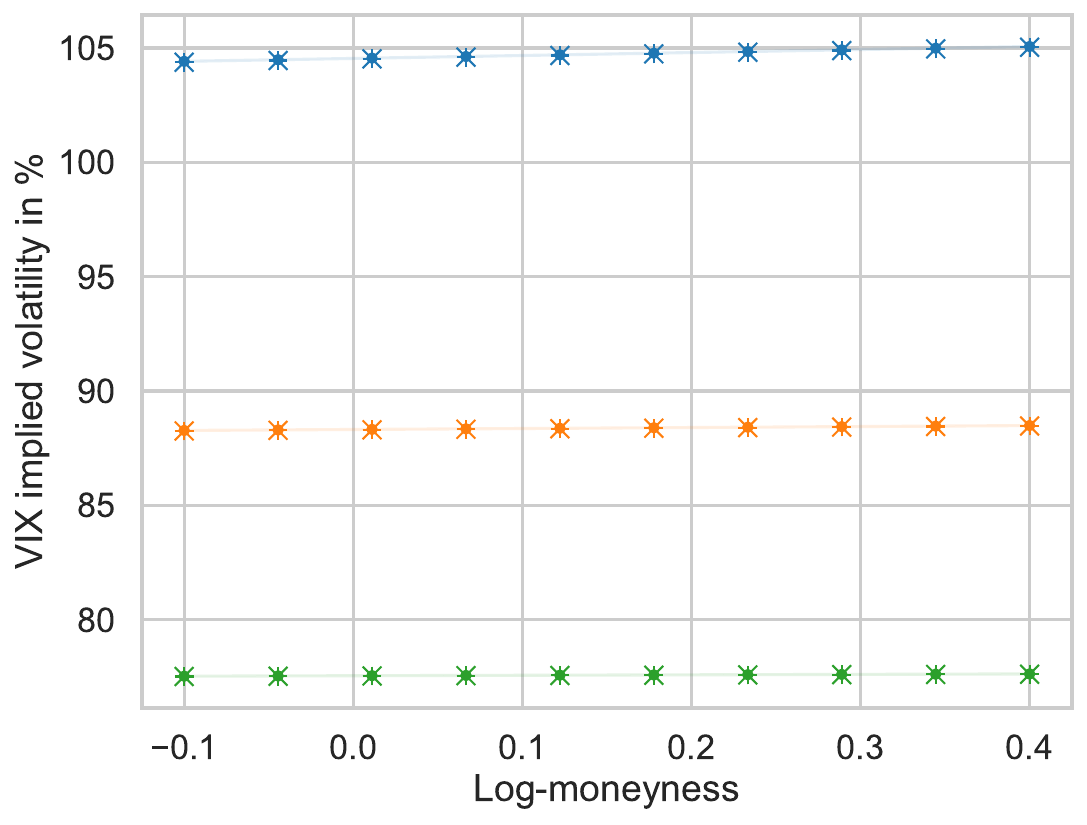}
    \includegraphics[width=0.5\linewidth]{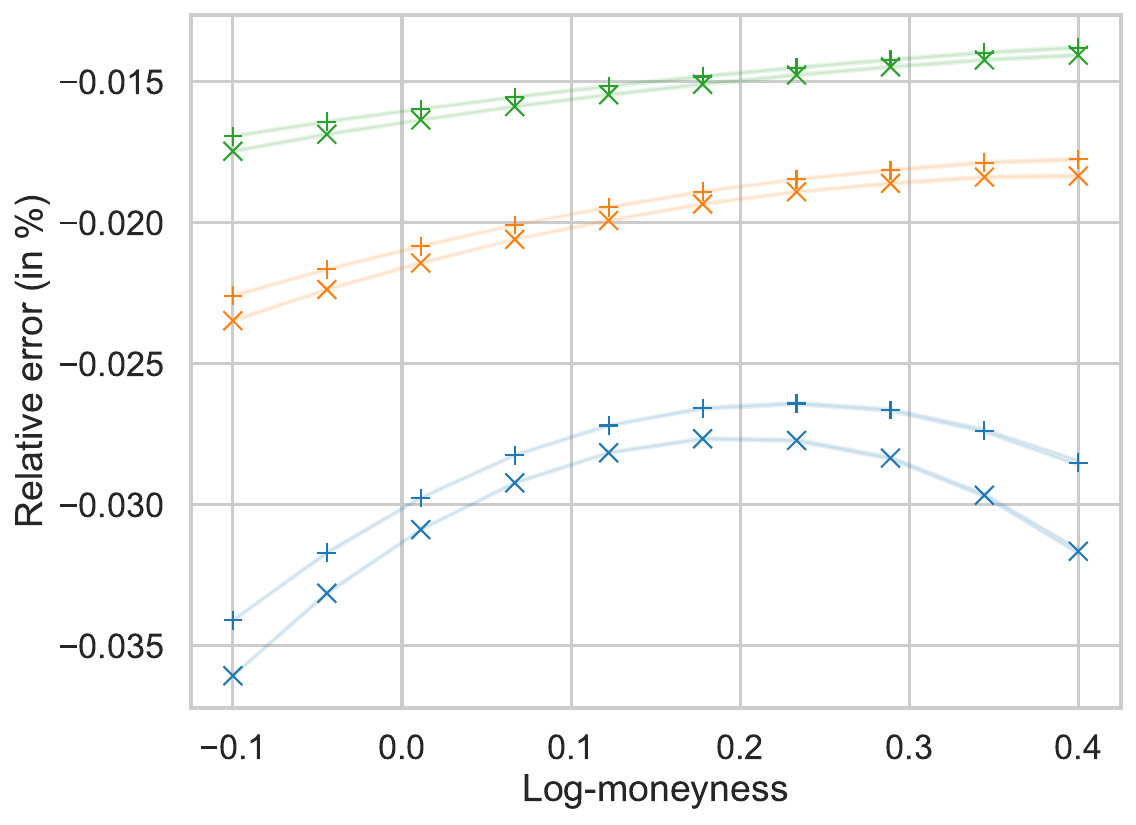}
    \caption{VIX smiles in the rough Bergomi model. Top: $H=0.1$, $\eta=1.0$. Bottom: $H=0.23$, $\eta=1.02$. Right: corresponding relative errors. Reference VIX futures $F$ for the 1-, 3-, and 6-month maturities: 0.2194, 0.2032, 0.1906 (top) and 0.2293, 0.2176, 0.2064 (bottom).}
    \label{fig:single_case_rough_bergomi}
\end{figure}

The VIX smile generated by \eqref{eq:instantaneous_forward_variance_sde} has been consistently reported to be nearly flat in \cite{bergomi_smile_2005, bayer_pricing_2016, horvath_volatility_2020}, which is also evident from Figure \ref{fig:single_case_bergomi} and \ref{fig:single_case_rough_bergomi}. This behaviour is consistent with $\VIX_T^2$ being close to log-normal in these examples, which explains the accuracy of the log-normal proxy $\VIX_{T,P}^2$. The resulting nearly log-normal distribution of $\VIX_T$ is consistent with the almost flat implied volatility smiles.
In contrast, market VIX smiles exhibit a pronounced positive skew, motivating the introduction of more flexible specifications, in particular mixed-kernel forward variance models.
In the next section, we apply the same methodology to derive an explicit expansion for the implied volatility of VIX options within the extended framework.

\section{Mixed-kernel forward variance models}
\label{sec:mixed_kernel}

Mixed-kernel forward variance models introduced in \cite{bergomi_smile_2008} are more capable of reproducing observed market VIX smiles than single-kernel forward variance models. In this case, the forward variance process $(\xi_t^u)_{t\le u}$ is modelled as a convex mixture of two log-normal components
\eqlnostar{eq:instantaneous_forward_variance_mixed}{
\xi_T^u = \xi_0^u \sum_{j=1}^{2} \lambda_j \exp(Y_{T,j}^u),
\qquad
Y_{T, j}^u := - \frac{1}{2} \int_0^T K_j^u(t)^2 \dt + \int_0^T K_j^u(t) \dW_t,
}
where $\lambda_1 = \lambda \in (0,1)$ is a mixing parameter, and $\lambda_2 = 1 - \lambda$. The cases $\lambda \in \{0,1\}$ reduce to the single-kernel model of Section~\ref{sec:single_kernel}.
The mixed formulation in \eqref{eq:instantaneous_forward_variance_mixed} encompasses various forward variance models depending on the choice of the kernels $K_j^u(t)$.
For kernels defined by $K_j^u(t) := \omega_j \, \ee^{-\kappa(u-t)}$, \eqref{eq:instantaneous_forward_variance_mixed} reduces to the mixed standard Bergomi model. 
Alternatively, for $K_j^u (t) := \eta_j \, (u-t)^{H-\frac{1}{2}}$, one obtains the mixed rough Bergomi model proposed in \cite{de_marco_volatility_2018,guyon_joint_2018}. Throughout this section we restrict to positive proportional kernels $K_j^u(t)=\alpha_jK_0^u(t)$, $\alpha_j>0, j \in \{1,2\}$, with a common $\kappa$ (standard model) or $H$ (rough model) parameter in the two models above. Specifically, $\alpha_j=\omega_j$ and $K_0^u(t)=\ee^{-\kappa(u-t)}$ in the standard model, while $\alpha_j=\eta_j$ and $K_0^u(t)=(u-t)^{H-1/2}$ in the rough model. With proportional kernels, both proxy components depend on the same Gaussian random variable. This fails if the two kernels have different values of $\kappa$ or $H$. In this case, the logarithms of the proxy components are jointly Gaussian with correlation strictly less than one. 

Within this extended framework with proportional kernels, $\VIX_T^2$ is a convex combination of integral means
\eqlnostar{eq:vix_mixed_kernel}{
\VIX_T^2 = \sum_{j=1}^{2} \lambda_j \VIX_{T,j}^2,
\qquad
\VIX_{T,j}^2 = \frac{1}{\Delta} \int_T^{T+\Delta} \xi_0^u \exp (Y_{T,j}^u) \du.
}

\subsection{VIX option price approximation}
\label{sec:price_approx_mixed_kernel}

Following the geometric-mean approximation used earlier in the single-kernel case, each arithmetic mean of exponential terms in \eqref{eq:vix_mixed_kernel} is approximated by its geometric mean. This leads to the following proxy for $\VIX_T^2$
\eqlnostar{eq:proxy_mixed_kernel}{
\VIX_{T,P}^2 = \sum_{j=1}^{2} \lambda_j \VIX_{T,P,j}^2,
\qquad
\VIX_{T,P,j}^2 = \VIXsqExp \exp \biggl(
\frac{1}{\Delta} \int_{T}^{T+\Delta} \frac{\xi_0^u}{\VIXsqExp} Y_{T,j}^u \du
\biggr).
}
By Proposition \ref{prop:proxy_vix}, we have
\eqstar{
\ln ( \VIX_{T,P,j}^2 ) \sim \Nc \bigl(\mu_{P,j}, \sigma_{P,j}^2\bigr),
\qquad
j \in \{1, 2\},
}
where
\eqlnostar{eq:mean_mixed}{
\mu_{P,j} := \ln (\VIXsqExp) - \frac{1}{2} \int_0^T \biggl(
\frac{1}{\Delta} \int_{T}^{T+\Delta} \frac{\xi_0^u}{\VIXsqExp} K_j^u(t)^2 \du
\biggr) \dt,
}
and
\eqlnostar{eq:variance_mixed}{
\sigma_{P,j}^2 := \int_0^T \biggl( 
\frac{1}{\Delta} \int_{T}^{T+\Delta} \frac{\xi_0^u}{\VIXsqExp} K_j^u(t) \du
\biggr)^2 \dt.
}
The following result is specialized to the positive proportional kernels introduced above.
\begin{theorem}(\cite[Theorem 3.1]{bourgey_weak_2023})
\label{thm:vix_opt_price_approx_mixed_kernel}
Let the initial curve satisfy Assumption~\ref{assump:initial_instantaneous_forward_variance}, and let Assumptions~\ref{assump:kernel_function} and \ref{assump:no_degenerate} hold for each of the positive proportional kernels $K_1$ and $K_2$, with common exponents $a_1$ \eqref{eq:gamma}, $a_2$ \eqref{eq:delta} and a common $\Delta_0$. Fix $\lambda\in(0,1)$ and let $\varphi\in\Cc_b^2$. Then, we have
\eqlnostar{eq:vix_opt_price_approx_mixed_kernel}{
\Eb \bigl[\varphi(\VIX_T^2) \bigr] = \Eb \bigl[\varphi(\VIX_{T,P}^2) \bigr]
+ \sum_{j=1}^2 \sum_{i=1}^3 \gamma_{i,j}\,P_{i,j}
+ \Ec_\varphi,
}
where the remainder term satisfies $|\Ec_\varphi|\le C\,\Delta^{3(a_1\wedge a_2/2)}$, for a constant $C>0$ independent of $\Delta$. The coefficients
$(\gamma_{i,j})_{i\in\{1,2,3\},\,j\in\{1,2\}}$ are the same as in Theorem~\ref{thm:vix_opt_price_approx_single_kernel}, replacing $K$ by $K_j$.
\noindent
For these proportional kernels, the main term reads
\eqlnostar{eq:main_term_mixed_kernel}{
\Eb \bigl[\varphi(\VIX_{T,P}^2) \bigr] = \Eb \biggl[ \varphi \Bigl(
\sum_{j=1}^{2} \lambda_j \,\ee^{\mu_{P,j}+\sigma_{P,j}\,Z}
\Bigr) \biggr],
\qquad 
Z \sim \Nc(0,1),
}
and the correction terms are given, for $i\in\{1,2,3\}$ and $j\in\{1,2\}$, by
\eqlnostar{eq:correction_terms_mixed_kernel}{
P_{i,j} = \partial_\epsilon^{i-1} \Eb \Bigl[ \Psi_j(X_j + \epsilon) \Bigr] \Big|_{\epsilon=0},
\qquad
X_j := \mu_{P,j} - L + \sigma_{P,j}\,Z,
\qquad
L := \ln \bigl(\VIXsqExp \bigr),
}
where
\eqstar{
\Psi_j(x) = \lambda_j\,\ee^{x}\, \VIXsqExp\, \varphi^\prime \Bigl(\VIXsqExp\,S_j(x)\Bigr),
}
and
\begin{align}
S_1(x) &:= \lambda_1\,\ee^{x} 
+ \lambda_2\, \ee^{(\frac{\alpha_1}{\alpha_2}-1)(L-\mu_{P,2})+\frac{\alpha_2}{\alpha_1}x},
\label{eq:S1} \\
S_2(x) &:= \lambda_1\, \ee^{(\frac{\alpha_2}{\alpha_1}-1)(L-\mu_{P,1})+\frac{\alpha_1}{\alpha_2}x}
+ \lambda_2\,\ee^{x}.
\label{eq:S2}
\end{align}
Thus, \eqref{eq:S1} and \eqref{eq:S2} apply to both models, using $\alpha_j=\omega_j$ for the standard or $\alpha_j=\eta_j$ for the rough.
\end{theorem}

Analogously to the single-kernel case, we denote by $C_P$, $P_P$, and $F_P$ the weak approximations of the VIX call, put, and futures prices, that is, the right-hand side of \eqref{eq:vix_opt_price_approx_mixed_kernel} without the remainder $\Ec_\varphi$. Throughout, we assume that Assumption~\ref{assump:no_degenerate} holds for each of the two kernels $K_1$ and $K_2$ (see Appendix \ref{sec:sufficient_conditions}), with common exponents $a_1$, $a_2$, and we keep the notation $\varepsilon = a_1 \wedge \frac{a_2}{2}$. The first term $\Eb[\varphi(\VIX_{T,P}^2)]$ in \eqref{eq:vix_opt_price_approx_mixed_kernel} is referred to as the \emph{zeroth-order term} of the expansion. By applying  Lemma~\ref{lem:order_of_coefficients} again to each kernel $K_j$, we can verify that the coefficients $\gamma_{i,j}$ are of order $\Delta^{2\varepsilon}$. Therefore, the correction terms $\sum_{i,j}\gamma_{i,j}P_{i,j}$ are of order $\Delta^{2\varepsilon}$, while the remainder $\Ec_\varphi$ is of order $\Delta^{3\varepsilon}$.
\begin{remark}
\label{re:non_smooth_payoff}
The VIX call, put, and futures payoffs are only $\frac{1}{2}$-H\"older continuous (in particular, not $\Cc^2$), and the call and futures payoffs are moreover unbounded, so they are not covered by the smooth-payoff statement above. In Lemma~\ref{lem:gaussian_rep_corrections}, we make the correction terms well-defined for these payoffs, but do not prove a bound on the pricing remainder. Instead, we assume that $|\Ec_\varphi|\le C\Delta^{3(a_1\wedge a_2/2)}$ holds for them, uniformly for log-strikes in a compact set. Our subsequent results on the volatility expansion error rely on this additional weak-error assumption.
\end{remark}

Let us define 
\eqlnostar{eq:func_h}{
h(z) &:= \sum_{j=1}^{2} \lambda_j \,\ee^{\mu_{P,j}+\sigma_{P,j}\,z}, \qquad \pi_j(z) := \frac{\lambda_j\, \ee^{\mu_{P,j} + \sigma_{P,j} z}}{h(z)}, \quad  z \in \Rb,\\
\label{eq:weights_q_mixed}
\qquad q_j(z) &:= \gamma_{1,j} + \gamma_{2,j}\, \frac{z}{\sigma_{P,j}} + \gamma_{3,j}\, \frac{z^2 - 1}{\sigma_{P,j}^2}, \qquad 
q(z) := \pi_1(z)\, q_1(z) + \pi_2(z)\, q_2(z). 
}
Thus, for $Z\sim\Nc(0,1)$, $\VIX_{T,P}^2$ has the law of $h(Z).$
\begin{lemma}
\label{lem:gaussian_rep_corrections}
Let $\varphi:(0,\infty)\to\Rb$ be locally absolutely continuous, with $|\varphi'(x)|\le c_\varphi(1+x^{-1/2})$ for almost every $x>0$. For $i\in\{1,2,3\}$ and $j\in\{1,2\}$, the correction terms \eqref{eq:correction_terms_mixed_kernel} are well-defined and satisfy
\eqlnostar{eq:gaussian_rep_P}{
P_{i,j} = \sigma_{P,j}^{1-i}\, \Eb\Bigl[ \He_{i-1}(Z)\, \lambda_j\, \ee^{\mu_{P,j} + \sigma_{P,j} Z}\, \varphi'\bigl(h(Z)\bigr) \Bigr],
}
for $Z\sim\Nc(0,1).$ Hence,
\eqlnostar{eq:gaussian_rep_corrections}{
\sum_{j=1}^2 \sum_{i=1}^3 \gamma_{i,j} P_{i,j} = \Eb\Bigl[ q(Z)\, h(Z)\, \varphi'\bigl(h(Z)\bigr) \Bigr].
}
\end{lemma}
\noindent The proof is given in Appendix~\ref{sec:gaussian_rep_corrections}. Note that the technical conditions on the payoff $\varphi$ in Lemma~\ref{lem:gaussian_rep_corrections} are satisfied by $\varphi\in\Cc_b^2$ and by the VIX call, put, and futures payoffs.


\subsection{Implied volatility expansion with explicit zeroth-order term}
\label{sec:vix_iv_proxy_mixed_kernel}

In the mixed-kernel forward variance model, the VIX option price approximation given in \eqref{eq:vix_opt_price_approx_mixed_kernel} does not admit a closed-form expression. Thus, the formula cannot be rewritten directly in the form of call \eqref{eq:vix_call_proxy_single_kernel} and put \eqref{eq:vix_put_proxy_single_kernel} price approximations as in the single-kernel case. However, by Lemma~\ref{lem:gaussian_rep_corrections}, they can still be expressed as one-dimensional Gaussian integrals that resemble a BS price formula.

Let $A := h^{-1} (\ee^{2k})$, where $h$ is defined in \eqref{eq:func_h} and $\ee^k$ denotes the VIX option strike. Since $h$ is an increasing bijection from $\Rb$ onto $(0,\infty)$, $A$ is well defined and is computed by solving a scalar equation. This also gives that $h(z) > \ee^{2k}$ if and only if $z > A$. For VIX calls, puts, and futures, respectively,
\eqstar{
x \varphi'(x) = \tfrac12 \sqrt{x}\, \ind_{x > \ee^{2k}},
\qquad
x \varphi'(x) = -\tfrac12 \sqrt{x}\, \ind_{x < \ee^{2k}},
\qquad
x \varphi'(x) = \tfrac12 \sqrt{x}.
}
Hence, Lemma~\ref{lem:gaussian_rep_corrections} gives the following first-order approximations for calls and futures based on the mixed log-normal proxy
\eqlnostar{eq:weak_approx_gaussian_rep}{
C_P = C_{P,0}(k) + C_{P,1}(k),
\qquad
F_P = F_{P,0} + F_{P,1},
}
where
\eqlnostar{eq:zeroth_first_order_mixed}{
C_{P,0}(k) &:= \Eb\Bigl[ \bigl(\sqrt{h(Z)} - \ee^k\bigr)^+ \Bigr],
\qquad
F_{P,0} := \Eb\Bigl[ \sqrt{h(Z)} \Bigr], \nonumber\\
C_{P,1}(k) &:= \frac{1}{2}\, \Eb\Bigl[ q(Z) \sqrt{h(Z)}\, \ind_{Z > A} \Bigr],
\qquad
F_{P,1} := \frac{1}{2}\, \Eb\Bigl[ q(Z) \sqrt{h(Z)} \Bigr].
}
The put approximation is $P_P=\Eb[(\ee^k-\sqrt{h(Z)})^+]-\frac12\Eb[q(Z)\sqrt{h(Z)}\,\ind_{Z<A}]$. Thus $C_P-P_P=F_P-\ee^k$, and approximations of call and put give the same implied volatility. The quantities in \eqref{eq:zeroth_first_order_mixed} are one-dimensional Gaussian integrals. The next result rewrites them by factoring out one of the two log-normal components of $h$.

\begin{proposition}
\label{prop:rewrite_vix_opt_price_approx_mixed_kernel}
Let $j \in \{1,2\}$, $j' := 3 - j$, and $A_j := A - \frac{\sigma_{P,j}}{2}$. Define
\eqstar{
x_{P,j} := \frac{\mu_{P,j}}{2} + \frac{\sigma_{P,j}^2}{8},
\
\Delta\sigma_{P,j} := (-1)^j (\sigma_{P,1} - \sigma_{P,2}),
\
C_j := \Bigl(\frac{\lambda_1}{\lambda_2} \Bigr)^{(-1)^j} \ee^{(-1)^j (\mu_{P,1} - \mu_{P,2}) + \frac{\sigma_{P,j}}{2} \Delta\sigma_{P,j}},
}
and, for $y \in \Rb$,
\eqlnostar{eq:func_g}{
\begin{aligned}
g^j_0(y) &:= \bigl(1 + C_j\, \ee^{\Delta \sigma_{P,j} y} \bigr)^{\frac{1}{2}},
\\
g^j_1(y) &:= \frac{1}{2}\, q\Bigl(y + \frac{\sigma_{P,j}}{2}\Bigr)\, g^j_0(y)
= \frac{q_j\bigl(y + \frac{\sigma_{P,j}}{2}\bigr) + q_{j'}\bigl(y + \frac{\sigma_{P,j}}{2}\bigr) \bigl(g^j_0(y)^2 - 1\bigr)}{2\, g^j_0(y)},
\end{aligned}
}
with $q_j$ and $q$ given in \eqref{eq:weights_q_mixed}. Then, the zeroth- and first-order terms in \eqref{eq:zeroth_first_order_mixed} can be written as
\begin{align}
C_{P,0}(k)
&= \sqrt{\lambda_j}\, \ee^{x_{P,j}} \int_{A_j}^\infty g^j_0(y)\, \phi(y) \dy - \ee^{k} \Phi(-A),
\qquad
C_{P,1}(k) = \sqrt{\lambda_j}\, \ee^{x_{P,j}} \int_{A_j}^\infty g^j_1(y)\, \phi(y) \dy,
\label{eq:vix_call_main_term_mixed_kernel}
\\
F_{P,i}
&= \sqrt{\lambda_j}\, \ee^{x_{P,j}} \int_\Rb g^j_i(y)\, \phi(y) \dy,
\qquad
i \in \{0,1\},
\label{eq:vix_fut_main_term_mixed_kernel}
\end{align}
and the VIX put price approximation can be written as
\eqlnostar{eq:vix_put_main_term_mixed_kernel}{
P_P = \ee^{k} \Phi(A) - \sqrt{\lambda_j}\, \ee^{x_{P,j}} \int_{-\infty}^{A_j} \bigl( g^j_0(y) + g^j_1(y) \bigr) \phi(y) \dy.
}
\end{proposition}
\noindent
The representation of the zeroth-order terms $C_{P,0}(k)$ and $F_{P,0}$ is \cite[Corollary 6.3.2]{bourgey2020stochastic}. The proof of Proposition~\ref{prop:rewrite_vix_opt_price_approx_mixed_kernel} is given in Appendix~\ref{sec:proof of prop:rewrite_vix_opt_price_approx_mixed_kernel}. The index $j$ records which of the two log-normal components has been factored out; the left-hand sides do not depend on $j$. The first-order factor $g^j_1$ is a combination of $g^j_0$ and $1/g^j_0$ with coefficients that are polynomials of degree two in $y$ and linear in the coefficients $(\gamma_{i,l})_{i\in\{1,2,3\},\,l\in\{1,2\}}$.

Although Proposition~\ref{prop:rewrite_vix_opt_price_approx_mixed_kernel} shows that the VIX option price approximations admit a representation that closely resembles the BS price formula, this representation is not itself a BS price corresponding to a single effective log-spot, strike, and volatility. Therefore, the technique of Theorem~\ref{thm:vix_iv_single_kernel} cannot be applied directly to derive an implied volatility expansion in mixed-kernel forward variance models. To overcome this difficulty, we compare the representation with the BS price formula and expand, in Hermite polynomials, only the non-Gaussian factors $g^j_0$ and $g^j_1$ in the integrands, rather than the integral representation as a whole. This procedure enables us to rewrite the approximation as a combination of a genuine BS price formula and its Vega, up to explicit residual terms (Proposition~\ref{prop:bs_representation_mixed_kernel}). Subsequently, we use the technique of Theorem~\ref{thm:vix_iv_single_kernel} to derive the implied volatility expansion. Contrary to the single-kernel case, however, the residual terms are not asymptotically small as $\Delta \to 0$. For this reason, we do not state an error bound for the resulting expansion, and assess its accuracy
numerically in Section~\ref{sec:numerical_tests_mixed_kernel}.

\paragraph{Hermite expansion of the non-Gaussian factors.} Since $1 \le g^j_0(y) \le 1 + \sqrt{C_j}\, \ee^{\Delta\sigma_{P,j} y/2}$ and $|g^j_1(y)| \le \frac12 \bigl( |q_j(y + \frac{\sigma_{P,j}}{2})| + |q_{j'}(y + \frac{\sigma_{P,j}}{2})|\, g^j_0(y) \bigr)$, where $q_j$ and $q_{j'}$ are polynomials, the functions $g^j_0$ and $g^j_1$ belong to $L^2(\phi) := L^2(\Rb, \phi(y)\dy)$. Since the Hermite polynomials $(\He_n)_{n\in\Nb}$ form an orthogonal basis of $L^2(\phi)$, with $\|\He_n\|^2_{L^2(\phi)} = n!$, the functions $g^j_i$, $i \in \{0,1\}$, admit the Hermite expansions
\eqlnostar{eq:hermite_expan_of_func_g}{
g^j_i(y) = \sum_{n=0}^\infty \omega^{(i)}_{n,j} \He_n (y),
\qquad
i \in \{0,1\},\ j \in \{1,2\},
}
where the series converge in $L^2(\phi)$ and the weights $\omega^{(i)}_{n,j}$ are given by
\eqlnostar{eq:weight_func}{
\omega^{(i)}_{n,j} := \frac{1}{n!} \int_{-\infty}^\infty g^j_i(y) \He_n(y) \phi (y) \dy,
\qquad
i \in \{0,1\},\ j \in \{1,2\},\ n \in \Nb.
}
By Parseval's identity, the $L^2(\phi)$-error of the truncation of \eqref{eq:hermite_expan_of_func_g} at order $N$ is explicit and non-increasing in $N$:
\eqlnostar{eq:hermite_truncation_error}{
\varepsilon_{i,j}(N)^2 := \Bigl\| g^j_i - \sum_{n=0}^N \omega^{(i)}_{n,j}\He_n \Bigr\|^2_{L^2(\phi)}
= \sum_{n > N} n!\, \bigl(\omega^{(i)}_{n,j}\bigr)^2 \xrightarrow[N \to \infty]{} 0.
}
\begin{remark}
\label{re:why_hermite}
Our choice of the Hermite polynomials is dictated by the structure of \eqref{eq:vix_call_main_term_mixed_kernel}. The Hermite polynomials are orthogonal precisely with respect to $\phi$, and their half-line integrals are explicit,
\eqlnostar{eq:hermite_half_line_integral}{
\int_x^\infty \He_n (y) \phi(y) \dy =
\begin{cases}
\He_{n-1}(x) \phi(x), \qquad &n \ge 1,\\
\Phi(-x), \qquad &n=0,
\end{cases}
\qquad \forall x \in \Rb.
}
Substituting \eqref{eq:hermite_expan_of_func_g} into \eqref{eq:vix_call_main_term_mixed_kernel} therefore produces a term proportional to $\Phi(\cdot)$, which matches the structure of the BS formula \eqref{eq:bs_call}, plus terms of the form $\He_{n-1}(\cdot)\phi(\cdot)$, which match the structure of the spatial derivatives of the BS price in \eqref{eq:bs_call_derivatives}. Other families of orthogonal polynomials, such as the Laguerre, Legendre, or Chebyshev polynomials, are orthogonal with respect to different weights and do not enjoy these two properties simultaneously.
\end{remark}

\paragraph{$N$th-order Hermite approximation.}
For $N \in \Nb$, $i \in \{0,1\}$, $j\in\{1,2\}$, and $x \in \Rb$, define
\eqlnostar{eq:Integral_0}{
I^j_{i,N}(x) &:= \omega^{(i)}_{0,j} \Phi(-x) + \sum_{n=1}^N \omega^{(i)}_{n,j} \He_{n-1}(x) \phi(x), \nonumber\\
\rho^j_{i,N}(x) &:= \int_x^\infty \Bigl( g^j_i(y) - \sum_{n=0}^N \omega^{(i)}_{n,j}\He_n(y) \Bigr) \phi(y) \dy.
}
By \eqref{eq:hermite_half_line_integral}, we have the exact identity
\eqlnostar{eq:Integral_0_exact}{
\int_{x}^\infty g^j_i(y) \phi(y) \dy = I^j_{i,N}(x) + \rho^j_{i,N}(x),
\qquad
\sup_{x\in\Rb} |\rho^j_{i,N}(x)| \le \varepsilon_{i,j}(N),
}
where the bound follows from the Cauchy--Schwarz inequality in $L^2(\phi)$. We call \emph{$N$th-order Hermite approximation} of any quantity defined through the integrals in Proposition~\ref{prop:rewrite_vix_opt_price_approx_mixed_kernel}, the expression obtained by replacing $g^j_0$ and $g^j_1$ with their truncated Hermite series $\sum_{n=0}^N \omega^{(i)}_{n,j}\He_n$, that is, by dropping the residuals $\rho^j_{i,N}$ in \eqref{eq:Integral_0_exact}.
\par

\begin{proposition}
\label{prop:vix_opt_price_hermite_expan}
Let $N \in \Nb$ and $j \in \{1,2\}$. The $N$th-order Hermite approximation of the weak VIX call price approximation $C_P$ in \eqref{eq:vix_opt_price_approx_mixed_kernel}, under the mixed standard and mixed rough Bergomi models, is given by
\eqlnostar{eq:vix_call_hermite_expan}{
C^j_{P,N} := \sqrt{\lambda_j}\, \ee^{x_{P,j}} \sum_{i=0}^{1} I^j_{i,N} (A_j) - \ee^{k} \Phi(-A),
}
with $A_j = A - \frac{\sigma_{P,j}}{2}$ as in Proposition~\ref{prop:rewrite_vix_opt_price_approx_mixed_kernel}. For the VIX put price,
\eqlnostar{eq:vix_put_hermite_expan}{
P^j_{P,N} := \ee^{k} \Phi (A)
- \sqrt{\lambda_j}\, \ee^{x_{P,j}} \sum_{i=0}^{1} \bigl( \omega^{(i)}_{0,j} - I^j_{i,N} (A_j) \bigr),
}
and the VIX futures price satisfies, for every $N$,
\eqlnostar{eq:vix_fut_hermite}{
F_P = \sqrt{\lambda_j}\, \ee^{x_{P,j}} \bigl( \omega^{(0)}_{0,j} + \omega^{(1)}_{0,j} \bigr).
}
Moreover, the truncation residual is explicit:
\eqlnostar{eq:hermite_residual_price}{
\begin{gathered}
C_P = C^j_{P,N} + \rho^j_N,
\qquad
P_P = P^j_{P,N} + \rho^j_N,
\\
\rho^j_N := \sqrt{\lambda_j}\, \ee^{x_{P,j}} \sum_{i=0}^1 \rho^j_{i,N}(A_j),
\qquad
|\rho^j_N| \le \sqrt{\lambda_j}\, \ee^{x_{P,j}} \sum_{i=0}^1 \varepsilon_{i,j}(N).
\end{gathered}
}
In particular, $C^j_{P,N} \to C_P$ and $P^j_{P,N} \to P_P$ as $N \to \infty$, and $C^j_{P,N} - P^j_{P,N} = F_P - \ee^k$ for every $N$.
\end{proposition}
The proof is given in Appendix~\ref{sec: proof of prop:vix_opt_price_hermite_expan}. Note that, while $C_P$ does not depend on $j$, the truncated quantities $C^1_{P,N}$ and $C^2_{P,N}$ differ in general, since $g^1_i$ and $g^2_i$ have different Hermite expansions.

\begin{remark}
Assume that $K_1 = K_2$ and $\lambda \in (0,1)$, so that the mixed-kernel model reduces to the single-kernel model. Then, $\sigma_{P,1} = \sigma_{P,2} = \sigma_P$, $\mu_{P,1} = \mu_{P,2} = \mu_P$, and $\gamma_{i,1} = \gamma_{i,2} = \gamma_i$ for $i \in \{1,2,3\}$. In this case, $\Delta\sigma_{P,j} = 0$, so that $g^j_0 \equiv \lambda_j^{-1/2}$, and $q = q_1 = q_2$ since $\pi_1 + \pi_2 = 1$. Writing $(y + \frac{\sigma_P}{2})^2 - 1 = \He_2(y) + \sigma_P \He_1(y) + \frac{\sigma_P^2}{4}$, we get
\eqstar{
g^j_1(y) = \frac{1}{2\sqrt{\lambda_j}}\, q_1\Bigl(y + \frac{\sigma_P}{2}\Bigr)
= \frac{1}{\sqrt{\lambda_j}} \Bigl( \delta + \frac{\gamma_2 + \gamma_3}{2 \sigma_P} \He_1(y) + \frac{\gamma_3}{2 \sigma_P^2} \He_2(y) \Bigr),
}
with $\delta$ as in \eqref{eq:vix_fut_proxy_single_kernel}. Thus, the Hermite series of $g^j_0$ and $g^j_1$ terminate after $n = 0$ and $n = 2$, respectively, and the formula of the $N$th-order Hermite approximation for the VIX call price in \eqref{eq:vix_call_hermite_expan} becomes, for every $N \ge 2$ and $j \in \{1,2\}$,
\eqstar{
C^j_{P,N} = (1 + \delta)\, \ee^{x_P} \Phi\Bigl(-A + \frac{\sigma_P}{2}\Bigr) - \ee^k \Phi(-A)
+ \ee^{x_P} \phi\Bigl(A - \frac{\sigma_P}{2}\Bigr) \Bigl( \frac{\gamma_2 + \gamma_3}{2 \sigma_P} + \frac{\gamma_3}{2 \sigma_P^2} \Bigl( A - \frac{\sigma_P}{2} \Bigr) \Bigr),
}
with $x_P = \frac{\mu_P}{2} + \frac{\sigma_P^2}{8}$. Since $A = (2k - \mu_P)/\sigma_P$, we have $d_1(x_P, k, \frac{\wtsp}{2}) = -A + \frac{\sigma_P}{2}$, and \eqref{eq:bs_call_derivatives} shows that this coincides with the single-kernel VIX call price approximation \eqref{eq:vix_call_proxy_single_kernel}.
\end{remark}

By \eqref{eq:vix_fut_hermite}, the $N$th-order Hermite approximation of the VIX call price in \eqref{eq:vix_call_hermite_expan} can be further written as
\eqstar{
C^j_{P,N}
= F_P\, \Phi(-A_j) - \ee^{k} \Phi(-A)
+ \sqrt{\lambda_j}\, \ee^{x_{P,j}} \phi(A_j) \sum_{n=1}^N \sum_{i=0}^1 \omega^{(i)}_{n,j} \He_{n-1}(A_j).
}
However, the leading term is not yet in the form of the BS price formula as in \eqref{eq:bs_call}. Indeed, letting $\wtspj := \sigma_{P,j}/\sqrt{T}$, for the triple $(\ln(F_P), k, \wtspj/2)$, the quantities $d_1$ and $d_2$ in \eqref{eq:d_bs} are in general not equal to $-A_j$ and $-A$, respectively. Thus, to obtain a BS price representation, we introduce auxiliary log-spot and log-strike variables defined as follows, for $\theta\in [0,1]$ and $j \in \{1,2\}$,
\begin{align}
x_{\theta,j} &:= \ln(F_P) + \theta\, \delta_j,
\label{eq:x_theta}\\
k_{\theta,j} &:= k - (1-\theta)\, \delta_j,
\label{eq:k_theta}\\
\delta_j &:= k - \ln(F_P) - \frac{A \wtspj \sqrt{T}}{2} + \frac{\wtspj^2 T}{8}.
\label{eq:delta_j}
\end{align}
With these choices, $x_{\theta,j} - k_{\theta,j} = \ln(F_P) - k + \delta_j = -\frac{A \wtspj \sqrt{T}}{2} + \frac{\wtspj^2 T}{8}$, and a direct computation shows that
\eqstar{
d_1 \bigl( x_{\theta,j}, k_{\theta,j}, \tfrac{\wtspj}{2} \bigr)
= -A + \frac{\wtspj \sqrt{T}}{2} = -A_j,
&&
d_2 \bigl(x_{\theta,j}, k_{\theta,j}, \tfrac{\wtspj}{2} \bigr) = -A,
}
independently of $\theta$. Hence, the $N$th-order Hermite approximation of the VIX call price in Proposition \ref{prop:vix_opt_price_hermite_expan}  can be rewritten, exactly, as
\eqlnostar{eq:vix_call_hermite_expan_bs}{
C^j_{P,N} &= C^\BS \bigl(x_{\theta,j}, k_{\theta,j}, \tfrac{\wtspj}{2} \bigr)
+ (F_P - \ee^{x_{\theta,j}}) \Phi (-A_j)
+ (\ee^{k_{\theta,j}} - \ee^{k}) \Phi(-A) \nonumber\\
&\quad+ \sqrt{\lambda_j}\, \ee^{x_{P,j}} \phi(A_j) \sum_{n=1}^N \sum_{i=0}^1 \omega^{(i)}_{n,j} \He_{n-1}(A_j),
}
where the second and third terms are correction terms arising from replacing the original inputs $(\ln(F_P), k)$ by the auxiliary log-spot and log-strike $(x_{\theta,j}, k_{\theta,j})$. The following result is the mixed-kernel counterpart of Theorem~\ref{thm:vix_iv_single_kernel}.

\begin{proposition}
\label{prop:bs_representation_mixed_kernel}
Let $N \in \Nb$, $j \in \{1,2\}$, $\theta \in [0,1]$, and let $x_{\theta,j}$, $k_{\theta,j}$, and $\delta_j$ be as in \eqref{eq:x_theta}-\eqref{eq:delta_j}. Define the volatility correction and the VIX implied volatility expansion
\eqlnostar{eq:vix_iv_proxy_mixed_kernel}{
\Delta\sigma_{\theta,j,N}(k) &:= \frac{\sqrt{\lambda_j} \ee^{x_{P,j}}}{\ee^{x_{\theta,j}} \sqrt{T}} \sum_{n=1}^N \sum_{i=0}^1 \omega^{(i)}_{n,j} \He_{n-1} (A_j), \nonumber \\
\widehat{\sigma}^{j}_{\theta,N}(k,T) &:= \frac{\wtspj}{2} + \Delta\sigma_{\theta,j,N}(k),
}
and assume that $\widehat{\sigma}^{j}_{\theta,N}(k,T) > 0$. Also, define the shift residual
\eqlnostar{eq:shift_residual}{
R^{\mathrm{shift}}_{\theta,j} 
&:= \delta_{x,j}\, \partial_x C^{\BS}\bigl(x_{\theta,j}, k_{\theta,j}, \tfrac{\wtspj}{2}\bigr) 
+ \delta_{k,j}\, \partial_k C^{\BS}\bigl(x_{\theta,j}, k_{\theta,j}, \tfrac{\wtspj}{2}\bigr) \nonumber \\
&= (F_P - \ee^{x_{\theta,j}}) \Phi (-A_j)
+ (\ee^{k_{\theta,j}} - \ee^{k}) \Phi(-A),
}
where $\delta_{x,j} := \ee^{-\theta \delta_j} -1$ and $\delta_{k,j} := \ee^{(1-\theta) \delta_j } -1$. Then, the following holds
\eqlnostar{eq:bs_representation_exact}{
C^j_{P,N} = C^{\BS}\bigl(x_{\theta,j}, k_{\theta,j}, \widehat{\sigma}^{j}_{\theta,N}(k,T)\bigr) + R^{\mathrm{shift}}_{\theta,j} + R^{\mathrm{vol}}_{\theta,j,N},
}
where the volatility residual $R^{\mathrm{vol}}_{\theta,j,N}$ satisfies the following
\eqlnostar{eq:vol_residual_bound}{
|R^{\mathrm{vol}}_{\theta,j,N}| \le \frac{1}{2}\, \Delta\sigma_{\theta,j,N}(k)^2 \sup_{\sigma \in [\frac{\wtspj}{2} \wedge \widehat{\sigma}^{j}_{\theta,N}, \frac{\wtspj}{2} \vee \widehat{\sigma}^{j}_{\theta,N}]} \bigl| \Vomma(x_{\theta,j}, k_{\theta,j}, \sigma) \bigr|,
}
with $\Vomma$ given in \eqref{eq:vomma}. Consequently, by Theorem~\ref{thm:vix_opt_price_approx_mixed_kernel} and Proposition~\ref{prop:vix_opt_price_hermite_expan}, the VIX call price satisfies
\eqlnostar{eq:bs_representation_true_price}{
C = C^{\BS}\bigl(x_{\theta,j}, k_{\theta,j}, \widehat{\sigma}^{j}_{\theta,N}(k,T)\bigr) + R^{\mathrm{shift}}_{\theta,j} + R^{\mathrm{vol}}_{\theta,j,N} + \rho^j_N + \Ec_\varphi,
}
where $\rho^j_N$ is the Hermite truncation residual \eqref{eq:hermite_residual_price} and $\Ec_\varphi$ the remainder of the weak approximation \eqref{eq:vix_opt_price_approx_mixed_kernel}.
\end{proposition}

\begin{proof}
Since $\partial_x C^{\BS}(x,k,\sigma) = \ee^x\Phi(d_1(x,k,\sigma))$ and $\partial_k C^{\BS}(x,k,\sigma) = -\ee^k\Phi(d_2(x,k,\sigma))$, the identities $d_1(x_{\theta,j},k_{\theta,j},\wtspj/2) = -A_j$ and $d_2(x_{\theta,j},k_{\theta,j},\wtspj/2) = -A$, together with $F_P = \ee^{x_{\theta,j}}\ee^{-\theta\delta_j}$ and $\ee^k = \ee^{k_{\theta,j}}\ee^{(1-\theta)\delta_j}$, which follow from \eqref{eq:x_theta}-\eqref{eq:k_theta}, show that the second and third terms in \eqref{eq:vix_call_hermite_expan_bs} equal $R^{\mathrm{shift}}_{\theta,j}$. 

Moreover, by \eqref{eq:vega} and the symmetry of $\phi$, the last term in \eqref{eq:vix_call_hermite_expan_bs} equals\\$\Delta\sigma_{\theta,j,N}(k)\,\Vega(x_{\theta,j}, k_{\theta,j}, \wtspj/2)$. Hence
\eqlnostar{eq:vix_call_hermite_expan_bs_re}{
C^j_{P,N} = C^{\BS}\bigl(x_{\theta,j}, k_{\theta,j}, \tfrac{\wtspj}{2} \bigr)
+ \Delta\sigma_{\theta,j,N}(k)\,\Vega\bigl(x_{\theta,j}, k_{\theta,j}, \tfrac{\wtspj}{2}\bigr) 
+ R^{\mathrm{shift}}_{\theta,j}.
}
Finally, we apply Taylor's theorem to  $\sigma \mapsto C^{\BS}(x_{\theta,j}, k_{\theta,j}, \sigma)$ between $\wtspj/2$ and $\widehat{\sigma}^{j}_{\theta,N}$, to get the following 
\eqstar{
C^{\BS}\bigl(x_{\theta,j}, k_{\theta,j}, \widehat{\sigma}^{j}_{\theta,N}\bigr) 
&= C^{\BS}\bigl(x_{\theta,j}, k_{\theta,j}, \tfrac{\wtspj}{2}\bigr) + \Delta\sigma_{\theta,j,N}\, \Vega\bigl(x_{\theta,j}, k_{\theta,j}, \tfrac{\wtspj}{2}\bigr) \\
&\quad + \frac12 \Delta\sigma_{\theta,j,N}^2\, \Vomma(x_{\theta,j}, k_{\theta,j}, \bar\sigma),
}
for some $\bar\sigma$ between $\wtspj/2$ and $\widehat{\sigma}^{j}_{\theta,N}$. Setting $R^{\mathrm{vol}}_{\theta,j,N} := -\frac12 \Delta\sigma_{\theta,j,N}^2\, \Vomma(x_{\theta,j}, k_{\theta,j}, \bar\sigma)$ and substituting the above results into \eqref{eq:vix_call_hermite_expan_bs_re} yields \eqref{eq:bs_representation_exact} and \eqref{eq:vol_residual_bound}. Identity \eqref{eq:bs_representation_true_price} follows from \eqref{eq:vix_opt_price_approx_mixed_kernel} and \eqref{eq:hermite_residual_price}.
\end{proof}

\paragraph{The implied volatility expansion.}
If we neglect the residuals $R^{\mathrm{shift}}_{\theta,j}$, $R^{\mathrm{vol}}_{\theta,j,N}$, $\rho^j_N$, and $\Ec_\varphi$ in \eqref{eq:bs_representation_true_price},  as in the single-kernel case, we can use the BS volatility computed with the auxiliary log-spot and log-strike $(x_{\theta,j}, k_{\theta,j})$ as the approximation for the market implied volatility $\sbs(k,T)$ computed with $(\ln(F), k)$ 
\eqlnostar{eq:vix_iv_mixed_kernel_approx}{
\sbs(k,T) \approx \widehat{\sigma}^{j}_{\theta,N}(k,T)
&= \frac{\wtspj}{2} 
+ \frac{\sqrt{\lambda_j} \ee^{x_{P,j}}}{\ee^{x_{\theta,j}} \sqrt{T}} \sum_{n=1}^N \sum_{i=0}^1 \omega^{(i)}_{n,j} \He_{n-1} (A_j).
}
This is the VIX implied volatility expansion with an explicit zeroth-order term in the mixed-kernel case. Let us discuss the sizes of the neglected terms.
\begin{itemize}
\item The weak-approximation remainder $\Ec_\varphi$ is of order $\Delta^{3\varepsilon}$ under the weak-error assumption of Remark~\ref{re:non_smooth_payoff}.
\item The Hermite truncation residual $\rho^j_N$ is controlled by $N$ through \eqref{eq:hermite_residual_price} and can be made arbitrarily small. 
\item The volatility residual $R^{\mathrm{vol}}_{\theta,j,N}$ is of second order in the correction $\Delta\sigma_{\theta,j,N}$. Contrary to the single-kernel case, $\Delta\sigma_{\theta,j,N}$ does not vanish as $\Delta \to 0$.
\item The shift residual $R^{\mathrm{shift}}_{\theta,j}$ is of first order in $\delta_j$ defined in \eqref{eq:delta_j}. In the single-kernel limit (say $\lambda_2 = 0$ with $j=1$), $\delta_1 = -\ln(1+\delta) = O(\Delta^{2\varepsilon})$, with $\delta$ as in \eqref{eq:vix_fut_proxy_single_kernel}, so that $R^{\mathrm{shift}}_{\theta,1}$ is negligible, consistently with Theorem~\ref{thm:vix_iv_single_kernel}. In the genuinely mixed case, $\delta_j$ does not vanish as $\Delta\to 0$ and depends on both $\sigma_{P,1}$ and $\sigma_{P,2}$ through $A$. Together with the volatility residual, this residual is responsible for the loss of accuracy of \eqref{eq:vix_iv_mixed_kernel_approx} when $\sigma_{P,1}$ and $\sigma_{P,2}$ are far apart.
In Remark~\ref{re:extended_vix_iv_mixed_kernel}, we present an alternative expansion which absorbs the strike part of the shift by a Taylor expansion in $k$, at the price of an additional residual of order $\delta_j^2$ and of a less stable formula.
\end{itemize}
Since $\delta_j$ and $\Delta\sigma_{\theta,j,N}$ are not asymptotically small in $\Delta$, we do not claim a rate of convergence for \eqref{eq:vix_iv_mixed_kernel_approx}. An alternative mixed-kernel expansion with a quantitative error bound, at the cost of a numerically computed implicit zeroth-order term, is given in Section~\ref{sec:vix_iv_proxy_mixed_kernel_rigorous}.
\par

\begin{remark}
\label{re:choice_of_sigma_base}
By Proposition~\ref{prop:rewrite_vix_opt_price_approx_mixed_kernel} and Proposition~\ref{prop:bs_representation_mixed_kernel}, both the VIX option price and the implied volatility expansion can be written in two forms, $j=1$ and $j=2$. Although the two representations coincide for the (untruncated) price approximations, they lead to different implied volatility expansions, since the zeroth-order term in \eqref{eq:vix_iv_mixed_kernel_approx}, namely $\widetilde{\sigma}_{P,1}/2$ or $\widetilde{\sigma}_{P,2}/2$, and the neglected residuals differ.
Moreover, for each level of log-strike, we can choose the value of $\theta$ in \eqref{eq:x_theta}-\eqref{eq:k_theta} that provides the smallest correction term for each $j$, i.e.
\eqstar{
\theta^* = \underset{\theta\in[0,1]}{\argmin}\, \bigl| \Delta\sigma_{\theta,j,N}(k) \bigr|.
}
We do so to improve the accuracy of the expansion in \eqref{eq:vix_iv_mixed_kernel_approx}. 
For a fixed $j$, the only $\theta$-dependent quantity in $\Delta\sigma_{\theta,j,N}(k)$ is $\ee^{-x_{\theta,j}}$. Therefore, minimizing with respect to $\theta$ is equivalent to maximizing $x_{\theta,j}$ over $[0,1]$, which leads to
\eqstar{
\theta^*=
\begin{cases}
1, &\delta_j > 0,\\
0, &\delta_j < 0,\\
\text{any\ } \theta \in [0,1], &\delta_j = 0. 
\end{cases}
}
Note that this choice of $\theta$ is a heuristic. It minimizes the volatility correction, and hence the prefactor $\Delta\sigma_{\theta,j,N}(k)^2$ in the bound \eqref{eq:vol_residual_bound} on $R^{\mathrm{vol}}_{\theta,j,N}$. 
\end{remark}

\begin{remark}
\label{re:extended_vix_iv_mixed_kernel}
An alternative implied volatility expansion is obtained by choosing $\theta=0$ in \eqref{eq:x_theta}-\eqref{eq:k_theta}, which gives 
\eqstar{
x_{0,j} = \ln (F_P), 
\qquad
k_{0,j} = k - \delta_j = \ln(F_P) + \frac{A \wtspj \sqrt{T}}{2} - \frac{\wtspj ^2 T}{8},
\qquad
j \in \{1,2\},
}
so that $\delta_{x,j} = 0$ and the shift residual reduces to its strike part, namely 
\eqstar{
R^{\mathrm{shift}}_{0,j} = \delta_{k,j}\, \partial_k C^\BS\bigl(\ln(F_P), k_{0,j}, \tfrac{\wtspj}{2}\bigr).
}
Instead of neglecting this term, one may absorb it by a Taylor expansion in the log-strike: since $k = k_{0,j} + \delta_j$ and $\delta_{k,j} = \ee^{\delta_j} - 1 = \delta_j + O(\delta_j^2)$,
\eqstar{
C^\BS\bigl(\ln(F_P), k_{0,j}, \tfrac{\wtspj}{2}\bigr) + \delta_{k,j}\, \partial_k C^\BS\bigl(\ln(F_P), k_{0,j}, \tfrac{\wtspj}{2}\bigr)
&= C^\BS\bigl(\ln(F_P), k, \tfrac{\wtspj}{2}\bigr) + R^{k}_{j}, \\
|R^{k}_{j}| &\le \widetilde{C}_j\, \delta_j^2,
}
where $\widetilde{C}_j$ depends on $|\delta_j|$ and on the first two derivatives of $C^\BS$ with respect to $k$ on the segment between $k_{0,j}$ and $k$. Combining with \eqref{eq:vix_call_hermite_expan_bs_re} for $\theta = 0$, and expressing the Hermite correction term as a $\Vega$ correction at the original log-strike $k$, we obtain the exact identity
\eqlnostar{eq:vix_call_proxy_intermediate}{
C^j_{P,N} = C^{\BS} \bigl(\ln (F_P), k, \tfrac{\wtspj}{2} \bigr)
+ \Delta\sigma^{\mathrm{ext}}_{j,N}(k)\,\Vega \bigl(\ln (F_P), k, \tfrac{\wtspj}{2} \bigr) + R^k_j,
}
where
\eqstar{
\Delta\sigma^{\mathrm{ext}}_{j,N}(k) :=
\frac{\sqrt{\lambda_j}\, \ee^{x_{P,j}}\, \phi(A_j)} {\Vega \bigl(\ln (F_P), k, \tfrac{\wtspj}{2} \bigr)}
\sum_{n=1}^N \sum_{i=0}^1 \omega^{(i)}_{n,j} \He_{n-1} (A_j).
}
Proceeding as in Proposition~\ref{prop:bs_representation_mixed_kernel}, that is, absorbing the $\Vega$ term into the volatility argument up to a residual of second order in $\Delta\sigma^{\mathrm{ext}}_{j,N}$, and neglecting this residual together with $R^k_j$, $\rho^j_N$, and $\Ec_\varphi$, leads to the alternative VIX implied volatility expansion
\eqlnostar{eq:vix_iv_proxy_mixed_kernel_extended}{
\sbs(k,T) \approx \widehat{\sigma}^{j,\mathrm{ext}}_{N}(k,T) := \frac{\wtspj}{2} + \Delta\sigma^{\mathrm{ext}}_{j,N}(k),
}
which, contrary to \eqref{eq:vix_iv_mixed_kernel_approx}, is directly expressed with respect to the pair $(\ln(F_P), k)$ and does not require the identification of the auxiliary pair $(x_{\theta,j}, k_{\theta,j})$ with $(\ln(F), k)$. This apparent advantage comes at two costs. First, the neglected residual $R^k_j$ is of second order in $\delta_j$, which, as discussed above, is not small in the mixed case. Second, the correction $\Delta\sigma^{\mathrm{ext}}_{j,N}$ contains the ratio 
\eqstar{
\phi(A_j)/\Vega(\ln(F_P), k, \wtspj/2) = \phi(A_j)/\bigl(F_P \sqrt{T}\,\phi(d_1(\ln(F_P), k, \wtspj/2))\bigr),
}
that is, a ratio of two Gaussian densities evaluated at different points, which becomes numerically unstable away from the money. Consequently, \eqref{eq:vix_iv_proxy_mixed_kernel_extended} is expected to be less accurate than \eqref{eq:vix_iv_mixed_kernel_approx}, which is confirmed by the numerical tests reported in Appendix~\ref{sec:numerical_tests_vix_iv_extended_mixed_kernel}.
\end{remark}

\begin{remark}[Moment-matched log-normal approximation]
\label{re:vix_iv_mixed_kernel_moment_matched}
In the mixed-kernel forward variance models, the squared VIX index proxy $\VIX_{T,P}^2$ in \eqref{eq:proxy_mixed_kernel} is given by a linear combination of two log-normal random variables $\VIX_{T,P,j}^2$, $j \in \{1,2\}$. A natural approximation choice is to use a single log-normal random variable with the first two moments matched with the first two moments of $\VIX_{T,P}^2$. Applying the VIX implied volatility expansion in \eqref{eq:vix_iv_proxy_single_kernel} from Theorem \ref{thm:vix_iv_single_kernel}, originally derived for the single-kernel case, we then obtain an approximation for the implied volatility of VIX options under the mixed-kernel framework. However, moment matching does not provide the kernel-dependent coefficients $\gamma_i$ of Theorem~\ref{thm:vix_opt_price_approx_single_kernel}, so that only the zeroth-order term of this expansion is available, and the resulting smile is flat. It therefore cannot reproduce the pronounced skew generated by the mixed models. Numerically, this approach performs significantly worse than the expansion \eqref{eq:vix_iv_mixed_kernel_approx}, with maximum relative errors above $60\%$ in Scenarios~1, 2, and~4 at every maturity, and between $12\%$ and $28\%$ in Scenario~3. \end{remark}

\subsubsection{Numerical tests}
\label{sec:numerical_tests_mixed_kernel}

Since the VIX implied volatility expansion \eqref{eq:vix_iv_mixed_kernel_approx} is based on a Hermite series expansion for the weak option price approximation in Proposition~\ref{prop:vix_opt_price_hermite_expan}, we first test the accuracy of the $N$th-order Hermite expansion of the VIX call price itself in \eqref{eq:vix_call_hermite_expan} with respect to its corresponding weak approximation in Theorem \ref{thm:vix_opt_price_approx_mixed_kernel}. However, since it is only an intermediate step towards computing the implied volatility expansions, we relegate the corresponding numerical procedure and results to Appendix \ref{sec:numerical_tests_vix_opt_mixed_kernel}. In short, the $N$th-order Hermite expansion of the VIX call price is accurate across all the parameter cases considered in Tables \ref{tab:parameters_order_mb}-\ref{tab:parameters_order_mrb}, with relative errors on prices below $1.4\%$ in all cases. 
\par

To assess the accuracy of our VIX implied volatility expansion (referred to as \textit{Expansion}) in \eqref{eq:vix_iv_mixed_kernel_approx}, with the choice of $\theta$ as explained in Remark~\ref{re:choice_of_sigma_base}, we perform numerical tests for maturities of 1, 3, and 6 months, and under two parameter scenarios each in the mixed standard and mixed rough Bergomi models. The truncation order $N$ of the Hermite expansion of $g^j_0$ in \eqref{eq:hermite_expan_of_func_g} is selected as follows. Recall from \eqref{eq:hermite_truncation_error} that the mean squared error of the truncated expansion under the standard Gaussian measure,
\eqstar{
\mathrm{MSE}_j(N) := \Eb_{y \sim \mathcal{N} (0,1)} \biggl[ \Bigl( \sum_{n=0}^N \omega^{(0)}_{n,j} \He_n (y) - g^j_0(y) \Bigr)^2 \biggr]
= \varepsilon_{0,j}(N)^2 = \sum_{n>N} n!\, \bigl(\omega^{(0)}_{n,j}\bigr)^2,
\quad
j \in \{1,2\},
}
with $\omega^{(0)}_{n,j}$ defined in \eqref{eq:weight_func}, is non-increasing in $N$ and tends to zero. We select the smallest $N$ for which $\mathrm{MSE}_j(N)$ falls below a prescribed tolerance $\epsilon_{\mathrm{tol}} = 10^{-6}$. Since $\|g^j_0\|^2_{L^2(\phi)} = 1 + C_j\, \ee^{\Delta\sigma_{P,j}^2/2}$ is explicit, $\mathrm{MSE}_j(N) = \|g^j_0\|^2_{L^2(\phi)} - \sum_{n=0}^N n!\, (\omega^{(0)}_{n,j})^2$ is evaluated through Parseval's identity, with the weights $\omega^{(i)}_{n,j}$ computed by Gauss--Legendre quadrature with $4000$ nodes on $[-40, 40]$. The resulting values of $N$, which are used for both $g^j_0$ and $g^j_1$, are reported in Tables~\ref{tab:parameters_order_mb} and \ref{tab:parameters_order_mrb}. Note that this criterion only controls the truncation error of $g^j_0$, and not the full bound $\rho^j_N$ in \eqref{eq:hermite_residual_price}, which also involves the truncation error $\varepsilon_{1,j}(N)$ of the first-order factor $g^j_1$. The resulting accuracy of the Hermite approximation of VIX call prices is assessed numerically in Appendix~\ref{sec:numerical_tests_vix_opt_mixed_kernel}.
The reference values (\textit{quadrature} and \textit{randomized QMC}) and the coefficients $\mu_{P,j}$, $\sigma_{P,j}$, and $\gamma_{i,j}$ are computed as in Section~\ref{sec:numerical_tests_single_kernel}. We also compute the implied volatilities obtained from the weak approximation of VIX option prices in Theorem~\ref{thm:vix_opt_price_approx_mixed_kernel} (referred to as \textit{Approx}), with the Gaussian integrals of \eqref{eq:weak_approx_gaussian_rep} computed as in Section~\ref{sec:rigorous_mixed_numerics}. 

We set $\kappa = 1$ and $H = 0.1$ in the mixed standard and mixed rough Bergomi models, respectively, with a flat initial forward variance curve $\xi_0^u = \xi_0 = 0.24^2$ for all $u \in [T, T+\Delta]$, and consider ten evenly spaced values of log-moneyness ranging from $-0.1$ to $0.4$. The remaining parameters and the selected order $N$ of the Hermite expansion for $g^j_0$ in \eqref{eq:hermite_expan_of_func_g} for these two models are given in Tables \ref{tab:parameters_order_mb} and \ref{tab:parameters_order_mrb}. Moreover, unless otherwise stated, we choose the index $j$ corresponding to the lower of the two values $\sigma_{P,1}$ and $\sigma_{P,2}$ for the implied volatility expansion in \eqref{eq:vix_iv_mixed_kernel_approx}. This choice is empirical, and both choices of $j$ are reported for Scenarios~2 and 4.

\begin{table}[H]
\centering
\footnotesize
\captionsetup{font = footnotesize, skip = 5pt}
\caption{Model parameters and selected orders $N$ of the Hermite expansion for $g^j_0$, for $j = 1$ and $j = 2$, in the mixed standard Bergomi model}
\begin{tabular}{lcccc}
\hline
Scenario & $(\omega_1, \omega_2, \lambda)$ & $T$ (months) & $N$ ($j=1$) & $N$ ($j=2$) \\
\hline
  &              & 1 & 10 & 10 \\
1 & (10, 2, 0.2) & 3 & 22 & 14 \\
  &              & 6 & 36 & 17 \\
\hline
  &              & 1 & 7  & 7 \\
2 & (0.5, 6, 0.3) & 3 & 12 & 13 \\
  &              & 6 & 17 & 19 \\
\hline
\end{tabular}
\label{tab:parameters_order_mb}
\end{table}

\begin{table}[H]
\centering
\footnotesize
\captionsetup{font = footnotesize, skip = 5pt}
\caption{Model parameters and selected orders $N$ of the Hermite expansion for $g^j_0$, for $j = 1$ and $j = 2$, in the mixed rough Bergomi model}
\begin{tabular}{lcccc}
\hline
Scenario & $(\eta_1, \eta_2, \lambda)$ & $T$ (months) & $N$ ($j=1$) & $N$ ($j=2$) \\
\hline
  &              & 1 & 4 & 4 \\
3 & (1.4, 0.7, 0.3) & 3 & 5 & 4 \\
  &              & 6 & 6 & 4 \\
\hline
  &              & 1 & 7  & 8 \\
4 & (2, 0.2, 0.4) & 3 & 11 & 9 \\
  &              & 6 & 14 & 12 \\
\hline
\end{tabular}
\label{tab:parameters_order_mrb}
\end{table}

In the mixed standard Bergomi model, we can observe from Figures \ref{fig:mixed_case_bergomi_1} and \ref{fig:mixed_case_bergomi_2} that our implied volatility expansion is accurate in Scenario 1, with relative errors below $7\%$ (and below $1\%$ for the 3- and 6-month maturities), whereas the performance deteriorates in Scenario 2, where relative errors reach about $40\%$ for the expansion around $\sigma_{P,1}$ and $50\%$ for the expansion around $\sigma_{P,2}$ at the 6-month maturity. Since the kernels are proportional, $\sigma_{P,1}/\sigma_{P,2} = \omega_1/\omega_2$, so that the larger proxy volatility is $5$ times the smaller one in Scenario 1 and $12$ times the smaller one in Scenario 2. Thus, although $\sigma_{P,1}$ and $\sigma_{P,2}$ are not close in Scenario 1, they are much further apart in Scenario 2. Moreover, the accuracy in Scenario 1 relies on the choice of $j$: the expansion around the lower proxy volatility $\sigma_{P,2}$, which also corresponds to the component with the larger weight $1-\lambda = 0.8$, is accurate, whereas the expansion around $\sigma_{P,1}$ (not shown) produces relative errors exceeding $100\%$. When the two proxy volatilities are as far apart as in Scenario 2, neither choice of $j$ yields an accurate expansion for the longer maturities. This behaviour is similarly repeated in the mixed rough Bergomi model, where we can see from Figures \ref{fig:mixed_case_rbergomi_1} and \ref{fig:mixed_case_rbergomi_2} that the expansion formula is accurate in Scenario 3, where $\eta_1/\eta_2 = 2$, with relative errors below $0.6\%$ (below $5\%$ for the expansion around $\sigma_{P,1}$), but performs poorly in Scenario 4, where $\eta_1/\eta_2 = 10$, with relative errors exceeding $20\%$ for the 3- and 6-month maturities.

This behaviour is in fact not driven by the Hermite truncation of the functions $g^j_0$ and $g^j_1$ in \eqref{eq:hermite_expan_of_func_g}, that is, by the residual $\rho^j_N$ in \eqref{eq:bs_representation_true_price}. As we numerically show in Appendix \ref{sec:numerical_tests_vix_opt_mixed_kernel}, the Hermite expansion formulas for VIX option prices are accurate (relative errors below $1.4\%$) in all scenarios, including Scenarios 2 and 4. 
The loss of accuracy of the implied volatility expansion emanates instead from the volatility residual $R^{\mathrm{vol}}_{\theta,j,N}$ and the shift residual $R^{\mathrm{shift}}_{\theta,j}$ discussed in Section~\ref{sec:vix_iv_proxy_mixed_kernel}, which are neglected in \eqref{eq:vix_iv_mixed_kernel_approx}. Neglecting them is required to obtain an explicit implied volatility formula. As shown in Section~\ref{sec:vix_iv_proxy_mixed_kernel}, they are negligible in the single-kernel limit, but do not vanish in the mixed-kernel case, and grow with the distance between $\sigma_{P,1}$ and $\sigma_{P,2}$, which enters $\delta_j$ and $\Delta\sigma_{\theta,j,N}$ through $A$ and the Hermite weights, respectively. In light of this issue, we also provide an alternative derivation of the VIX implied volatility in Remark~\ref{re:extended_vix_iv_mixed_kernel}, which absorbs the strike part of the shift by a Taylor expansion instead of discarding it. However, the resulting expansion \eqref{eq:vix_iv_proxy_mixed_kernel_extended} exhibits poorer numerical performance. This loss of accuracy is discussed in Remark~\ref{re:extended_vix_iv_mixed_kernel} and further confirmed by the numerical experiments reported in Appendix~\ref{sec:numerical_tests_vix_iv_extended_mixed_kernel}.
Having said that, the loss of accuracy in Scenarios 2 and 4 mainly affects the longer maturities. For the 1-month maturity, the expansion around the lower proxy volatility yields relative errors below $10\%$ in Scenario 2 and below $13\%$ in Scenario 4, and the expansion around the higher proxy volatility $\sigma_{P,2}$ yields relative errors below $5\%$ in Scenario 2. For the 3- and 6-month maturities, however, the errors are large for both in-the-money (ITM) and out-of-the-money (OTM) strikes, and their location depends on the choice of $j$. In Scenario 2, the expansion around $\sigma_{P,1}$ has its largest errors slightly OTM (about $41\%$ at log-moneyness $0.18$ for the 6-month maturity), with relative errors of $20$--$27\%$ also for ITM strikes, whereas the expansion around $\sigma_{P,2}$ has its largest errors for ITM strikes (about $-50\%$ at log-moneyness $-0.04$ for the 6-month maturity). In Scenario 4, the expansion around the lower proxy volatility $\sigma_{P,2}$ has its largest errors slightly OTM (about $26\%$ at log-moneyness $0.18$ for the 6-month maturity), whereas the expansion around $\sigma_{P,1}$ has large errors near the money (about $-53\%$ at log-moneyness $0.01$), for OTM strikes (about $51\%$ at log-moneyness $0.4$), and for ITM strikes (about $78\%$ at log-moneyness $-0.1$) for the 6-month maturity. Therefore, in these parameter regimes, the poor accuracy is not confined to OTM strikes.

With the timing protocol of Section~\ref{sec:numerical_tests_single_kernel}, the median runtime of \textit{Expansion} for a ten-strike three-month smile is $0.41$--$0.48\,\ms$, including the coefficients, the Hermite weights, and the exercise boundary $A$ for each strike, but excluding the selection of $N$. This is comparable to the runtime of \textit{Approx}, $0.44$--$0.46\,\ms$ (Table~\ref{tab:mixed_main_timings}), so that, contrary to the single-kernel case, the explicit expansion brings no speed-up over the BS inversion of the weak price approximation in the mixed-kernel case. The references take $20$--$21\,\ms$ (\textit{quadrature}) and about $11\,\s$ (\textit{randomized QMC}).

\begin{figure}[H]
    \centering
    \hspace{-0.5cm}
    \includegraphics[width=0.5\linewidth]{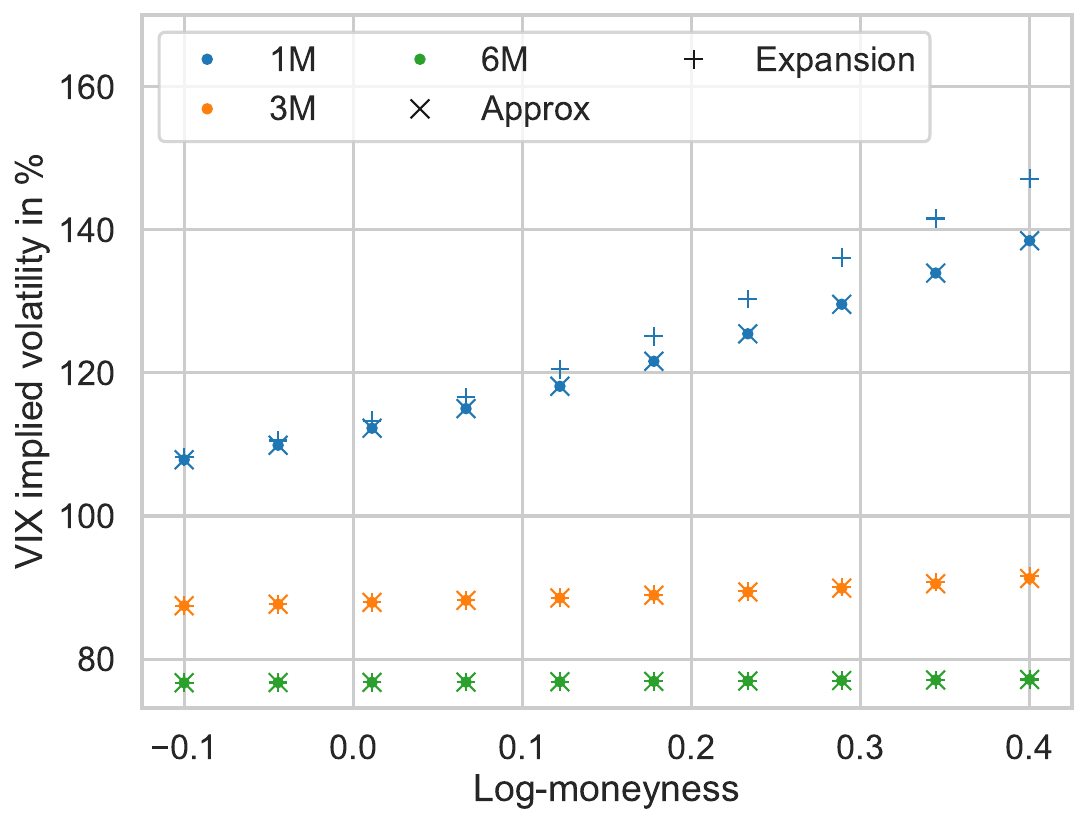}
    \includegraphics[width=0.5\linewidth]{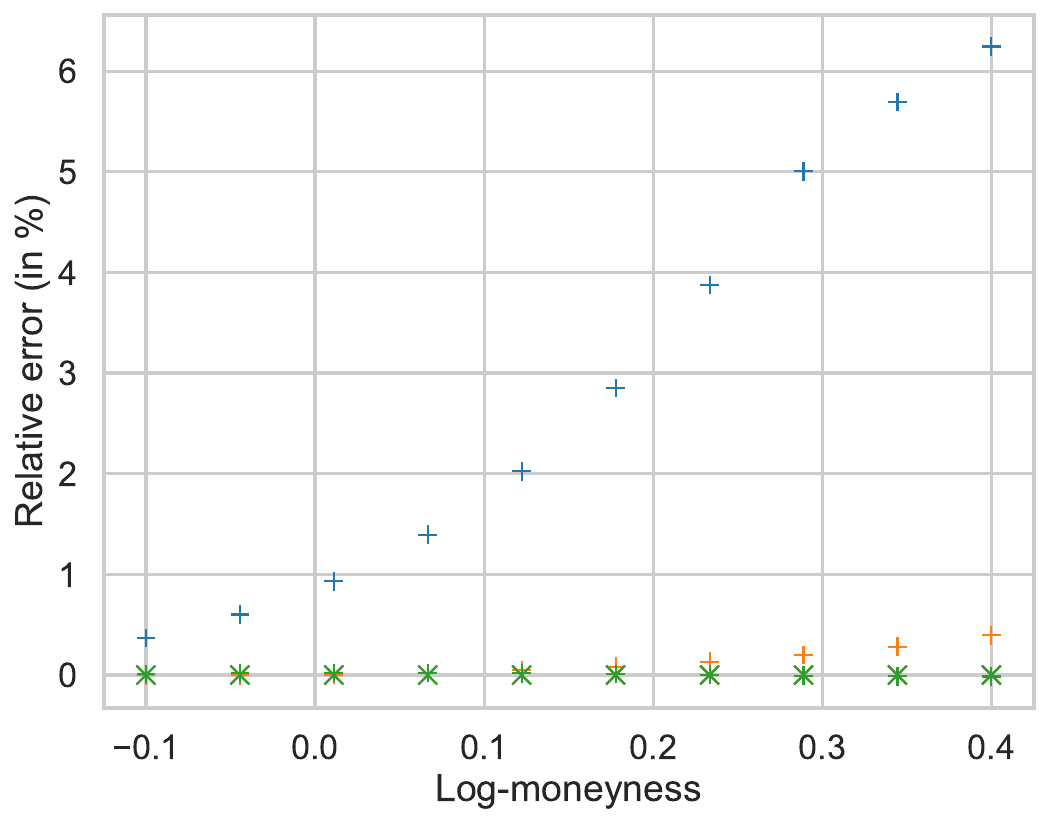}
    \caption{VIX smiles of Scenario 1 in the mixed standard Bergomi model. Expansion around $\sigma_{P,2}$ (lower). Reference VIX futures $F$ for the 1-, 3-, and 6-month maturities: 0.2178, 0.1986, 0.1862.}
    \label{fig:mixed_case_bergomi_1}
\end{figure}

\begin{figure}[H]
    \centering
    \hspace{-0.5cm}
    \includegraphics[width=0.5\linewidth]{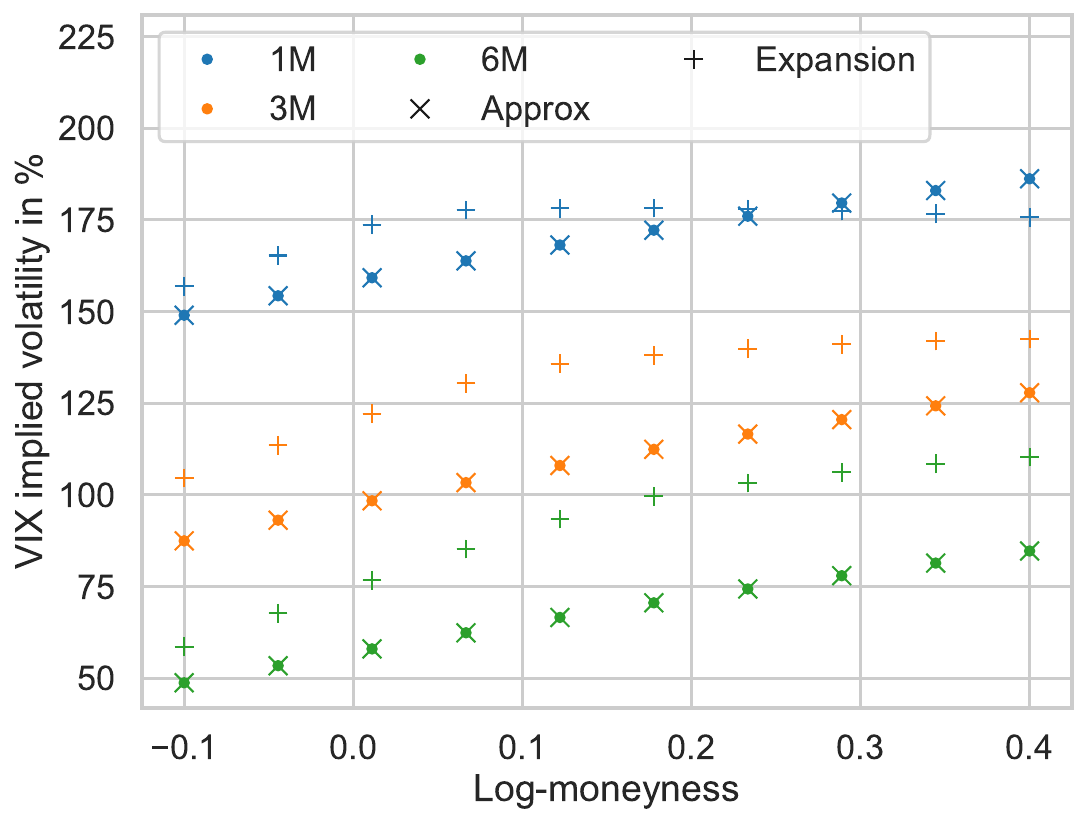}
    \includegraphics[width=0.5\linewidth]{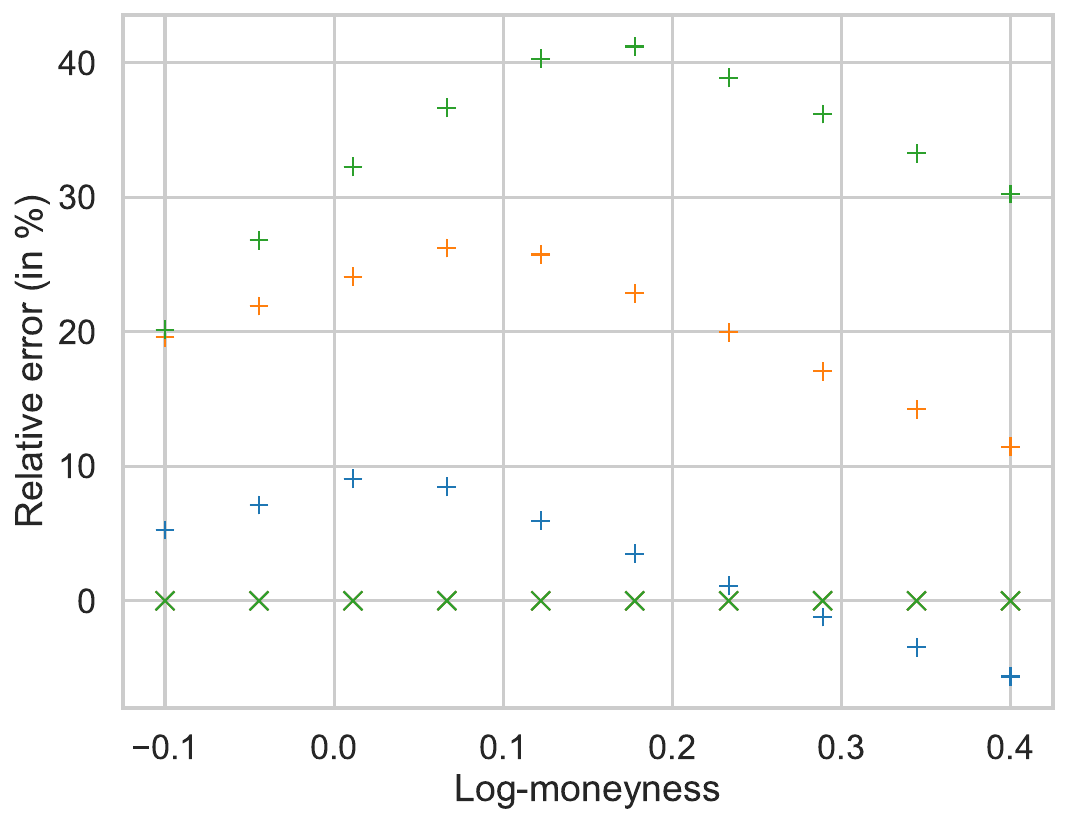}\\
    \hspace{-0.5cm}
    \includegraphics[width=0.5\linewidth]{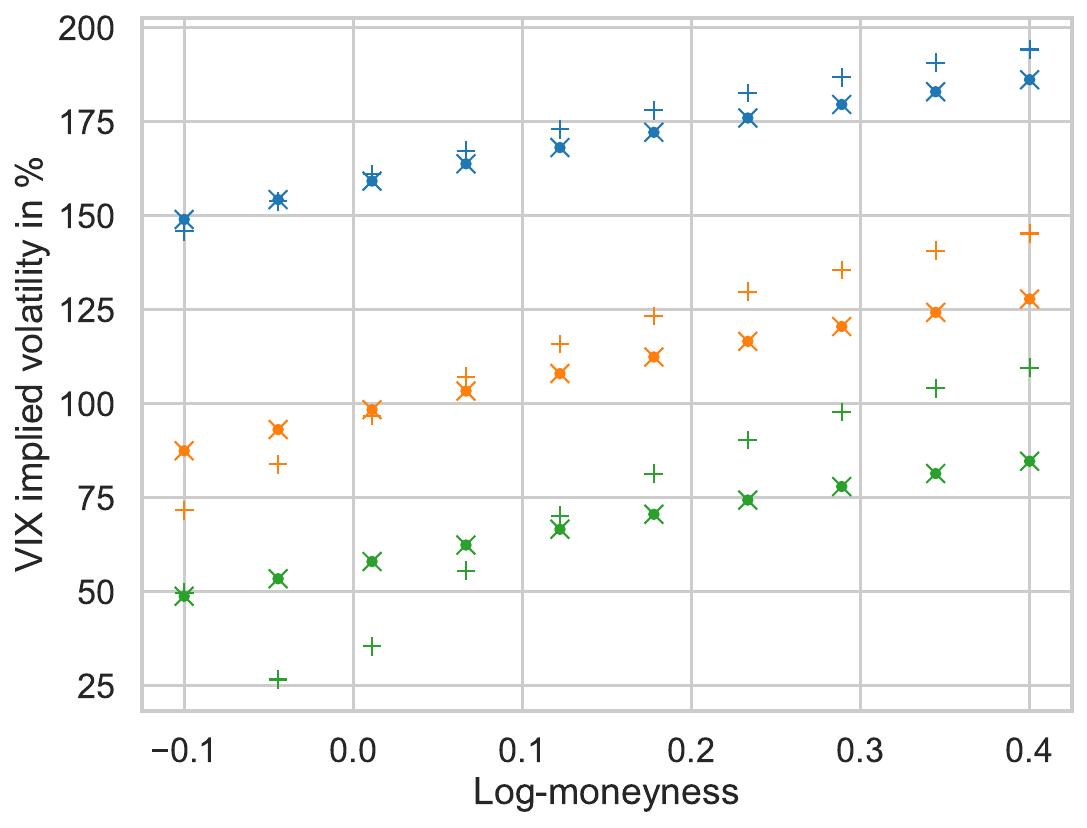}
    \includegraphics[width=0.5\linewidth]{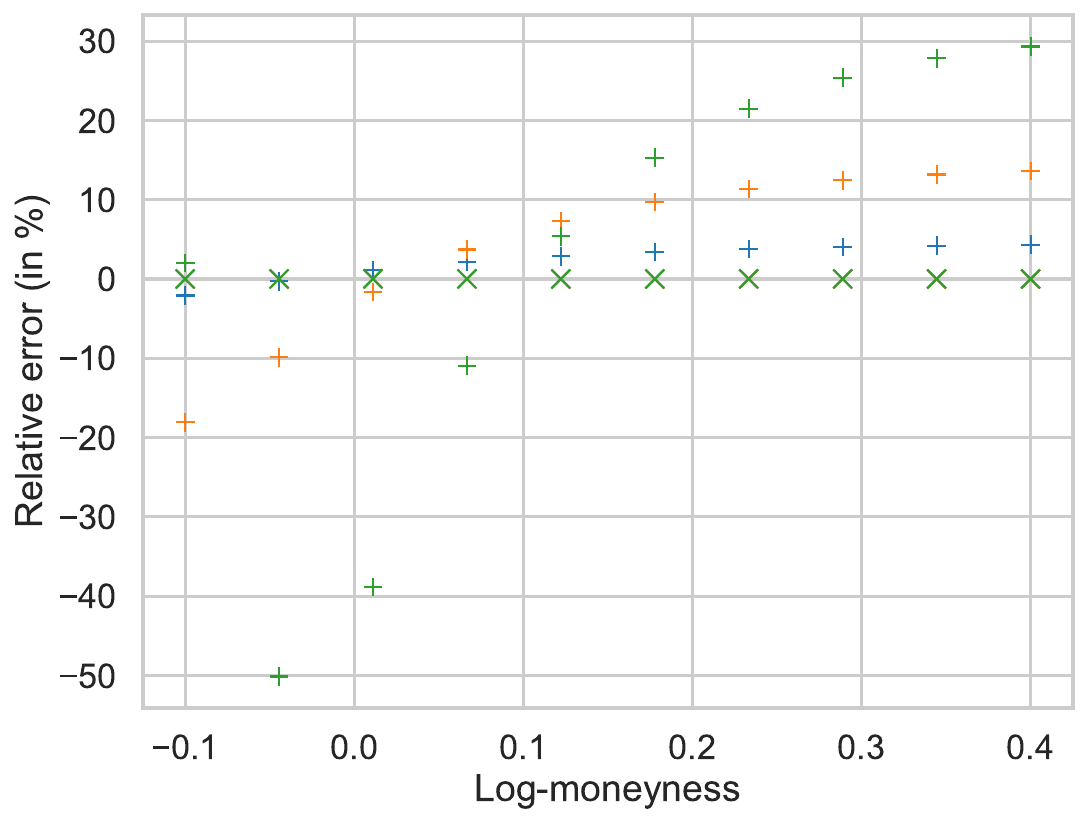}
    \caption{VIX smiles of Scenario 2 in the mixed standard Bergomi model. Top: Expansion around $\sigma_{P,1}$ (lower). Bottom: Expansion around $\sigma_{P,2}$ (higher). Reference VIX futures $F$ for the 1-, 3-, and 6-month maturities: 0.2072, 0.1752, 0.1566.}
    \label{fig:mixed_case_bergomi_2}
\end{figure}

\begin{figure}[H]
    \centering
    \hspace{-0.5cm}
    \includegraphics[width=0.5\linewidth]{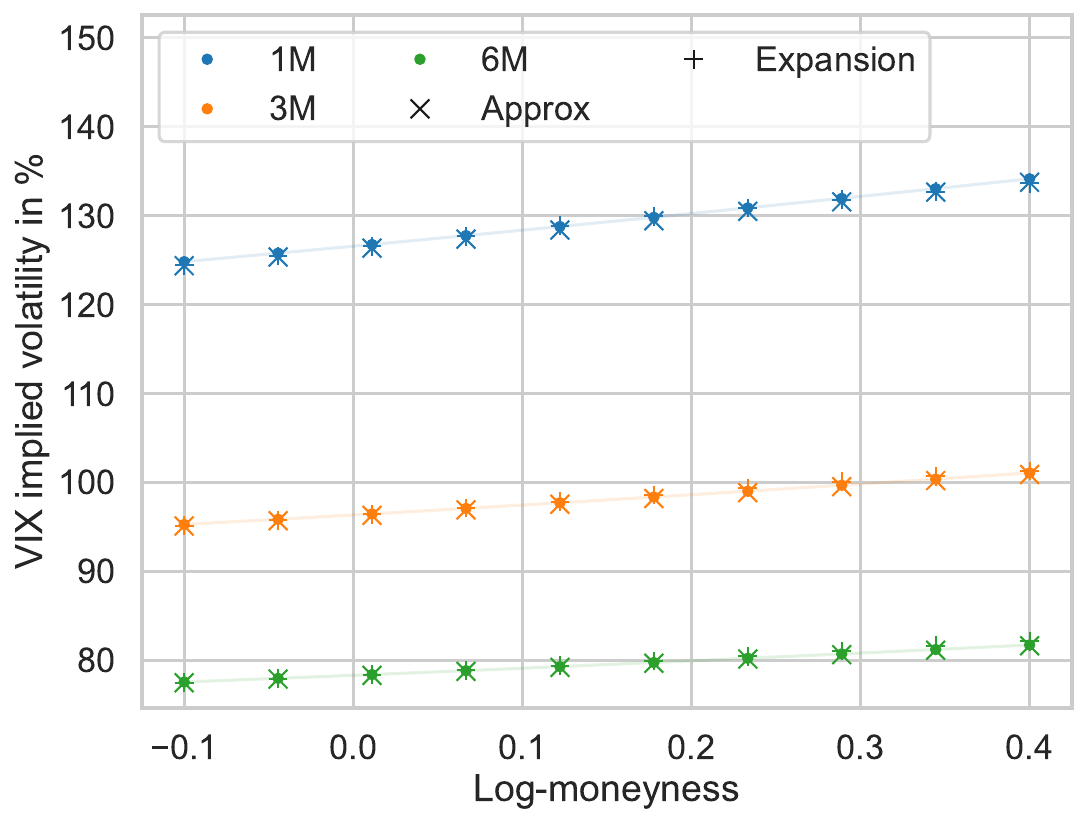}
    \includegraphics[width=0.5\linewidth]{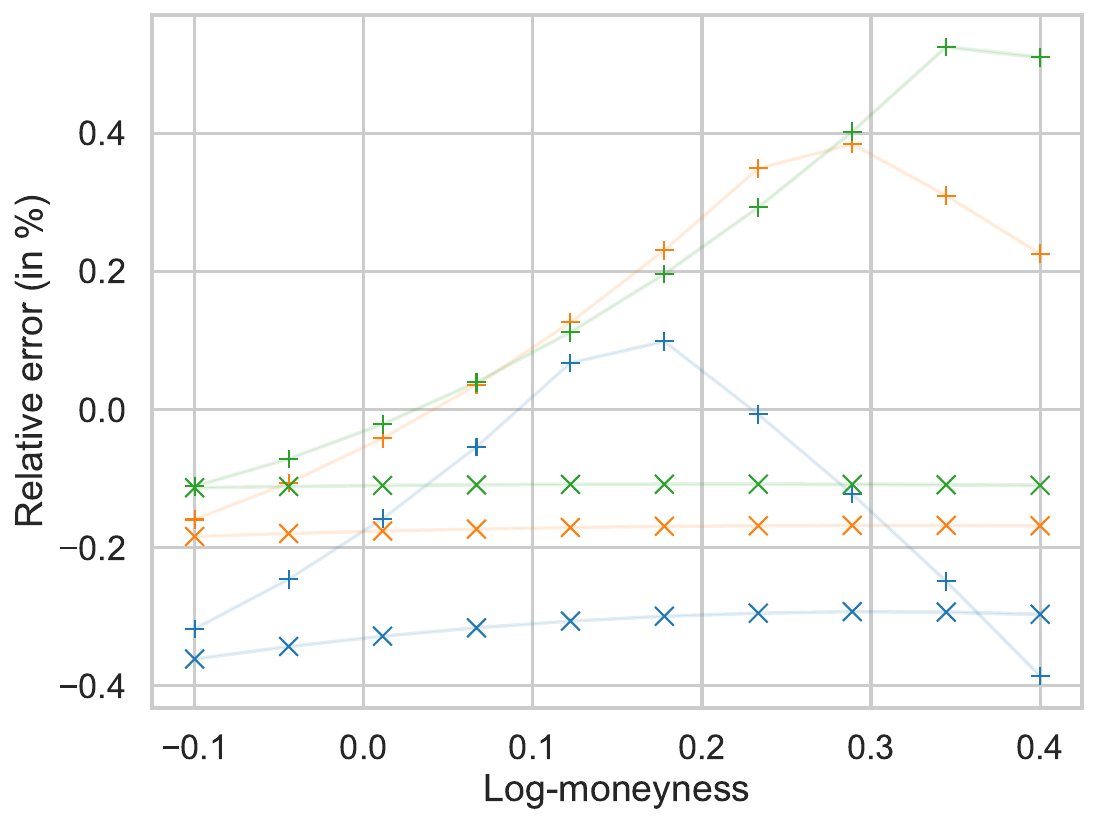}
    \caption{VIX smiles of Scenario 3 in the mixed rough Bergomi model. Expansion around $\sigma_{P,2}$ (lower). Reference VIX futures $F$ for the 1-, 3-, and 6-month maturities: 0.2233, 0.2106, 0.2010.}
    \label{fig:mixed_case_rbergomi_1}
\end{figure}

\begin{figure}[H]
    \centering
    \hspace{-0.5cm}
    \includegraphics[width=0.5\linewidth]{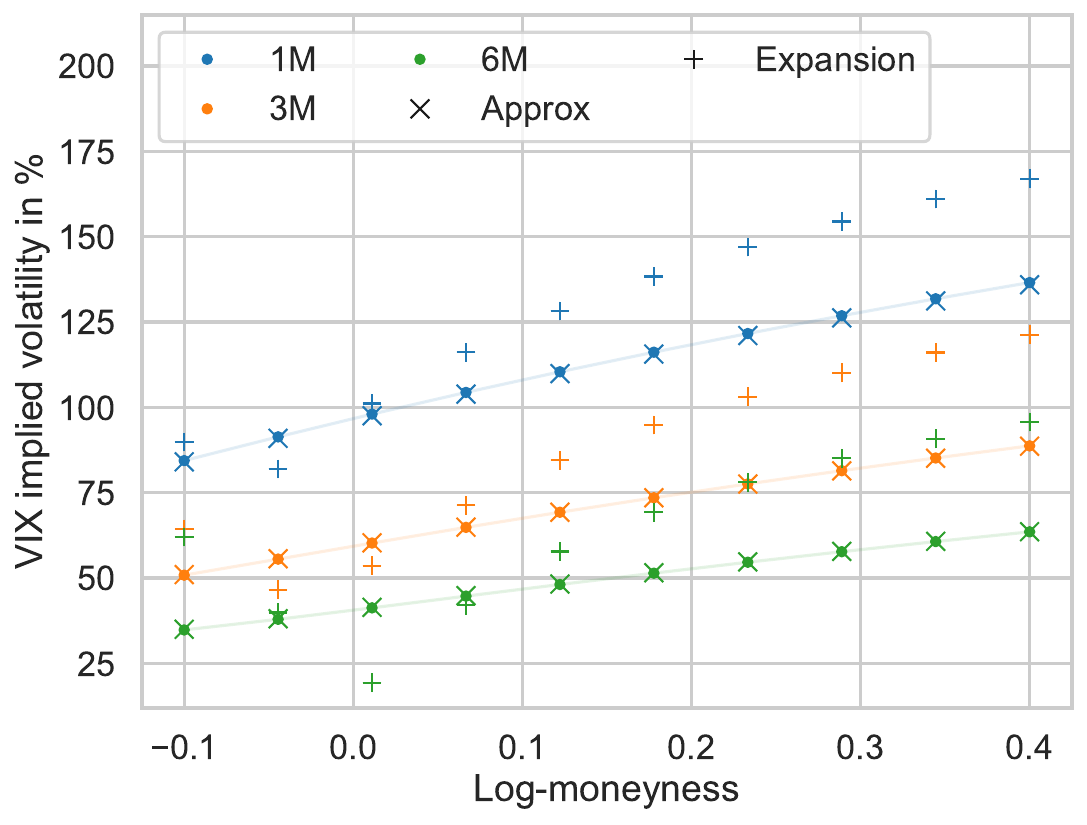}
    \includegraphics[width=0.5\linewidth]{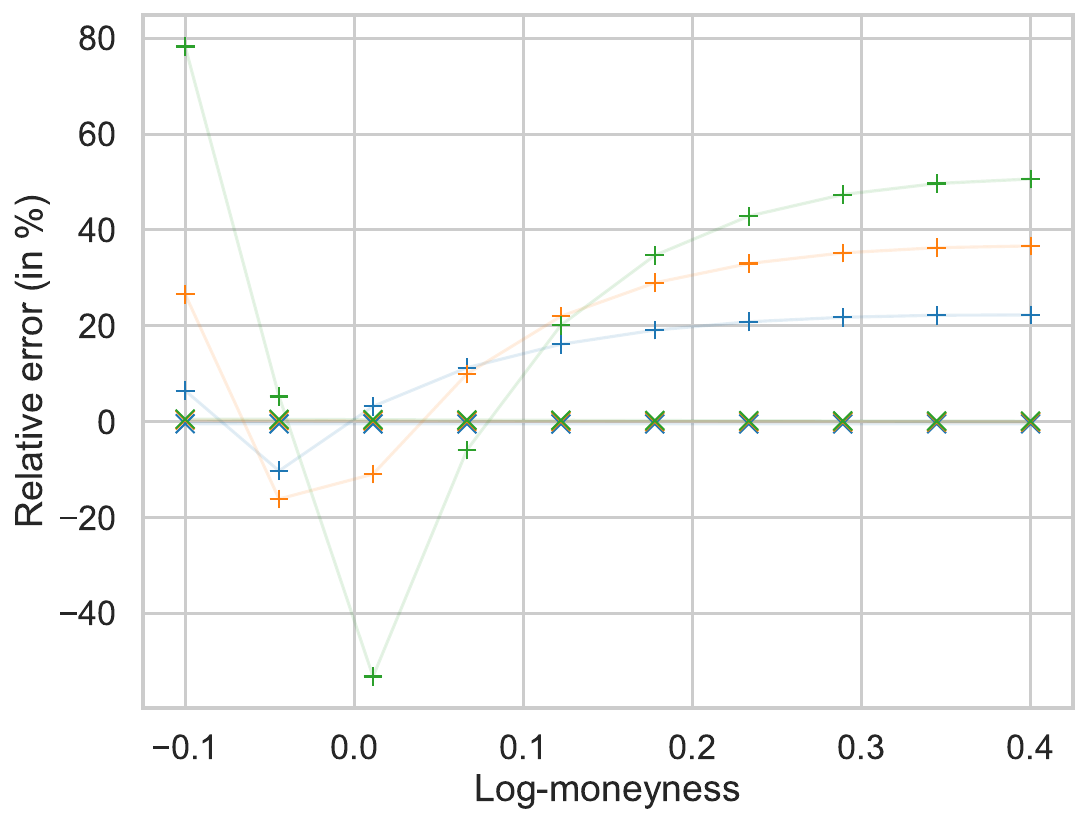}\\
    \hspace{-0.5cm}
    \includegraphics[width=0.5\linewidth]{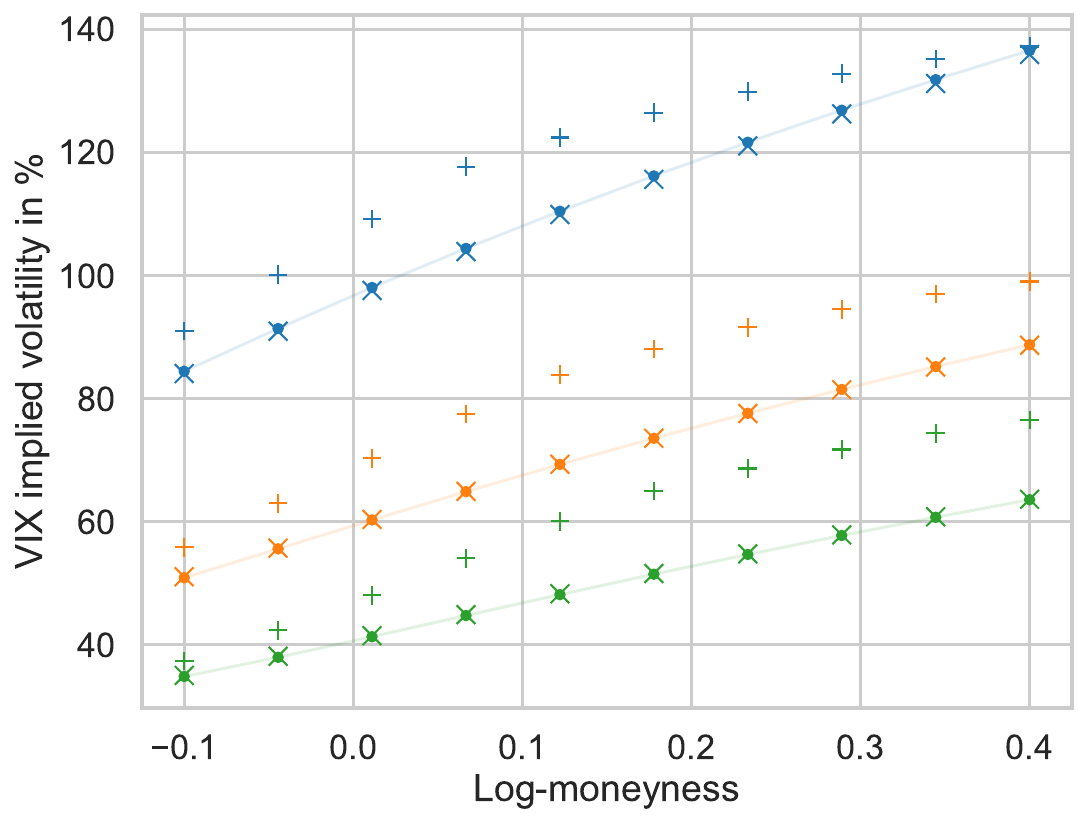}
    \includegraphics[width=0.5\linewidth]{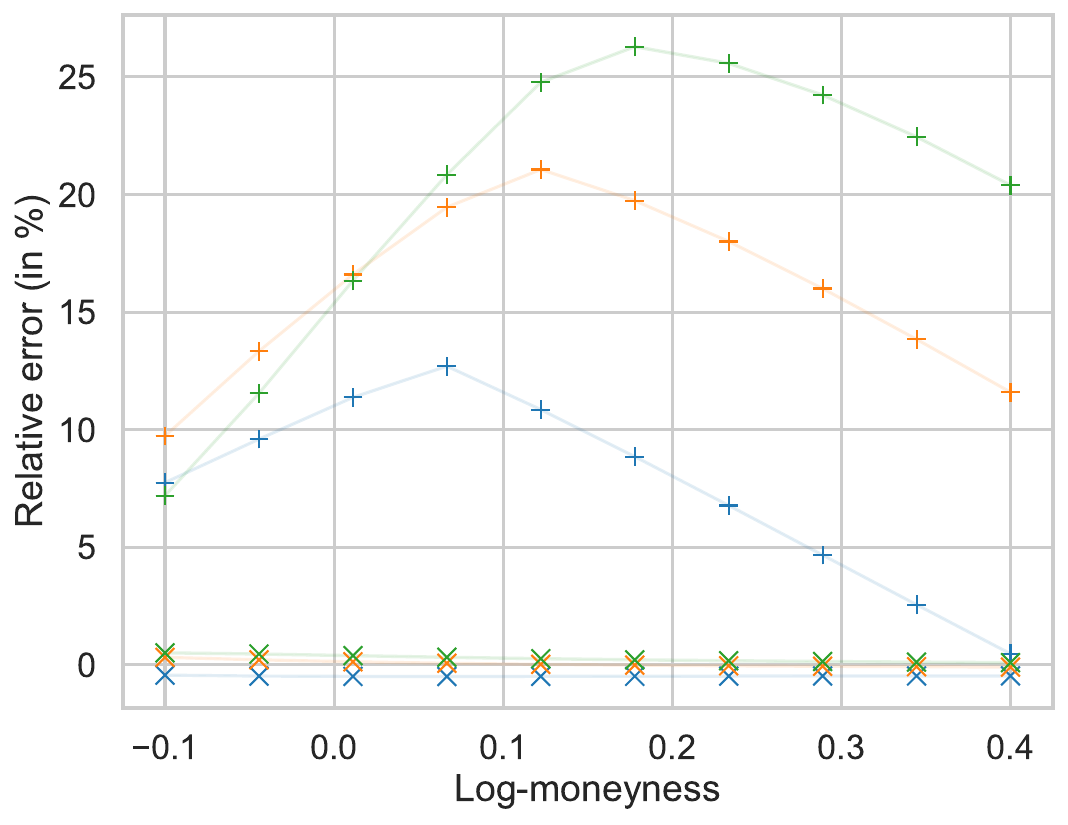}
    \caption{VIX smiles of Scenario 4 in the mixed rough Bergomi model. Top: Expansion around $\sigma_{P,1}$ (higher). Bottom: Expansion around $\sigma_{P,2}$ (lower). Reference VIX futures $F$ for the 1-, 3-, and 6-month maturities: 0.2236, 0.2130, 0.2061.}
    \label{fig:mixed_case_rbergomi_2}
\end{figure}

\subsection{Implied volatility expansion with implicit zeroth-order term}
\label{sec:vix_iv_proxy_mixed_kernel_rigorous}
As can be seen from the numerical tests in Section~\ref{sec:numerical_tests_mixed_kernel}, factoring out an explicit zeroth-order term as in \eqref{eq:vix_iv_proxy_mixed_kernel} leads to a significant loss of accuracy in the implied volatility expansion formula when $\sigma_{P,1}$ and $\sigma_{P,2}$ \eqref{eq:variance_mixed} appearing in the log-normal proxy \eqref{eq:proxy_mixed_kernel} are far apart. To alleviate this issue and at the cost of losing a fully explicit implied volatility expansion, we pursue another direction where only the zeroth-order term in the expansion is computed fully numerically and varies with the log-strike $k$. This approach also extends the idea of choosing between $\sigma_{P,1}$ and $\sigma_{P,2}$ depending on $k$ in Remark~\ref{re:choice_of_sigma_base}. As we show later, this approach achieves an accurate expansion for implied volatility, but at the cost of losing the explicit zeroth-order term. The correction terms in this expansion, on the other hand, are fully explicit in model parameters. In particular, we seek an expansion of the form
\eqlnostar{eq:mixed_proxy_iv_expansion_implicit}{
\widehat\Sigma(k,T):=\Sigma_0(k)+\Sigma_1(k),
}
where $\Sigma_0$ is the proxy implied volatility and $\Sigma_1$ provides the correction. Both components above are built from the decomposition \eqref{eq:weak_approx_gaussian_rep} of the weak approximations into the zeroth-order terms $C_{P,0}(k)$, $F_{P,0}$ and the first-order terms $C_{P,1}(k)$, $F_{P,1}$ given in \eqref{eq:zeroth_first_order_mixed}, which are one-dimensional Gaussian integrals.

\begin{theorem}
\label{thm:vix_iv_mixed_kernel_rigorous}
Let Assumption~\ref{assump:initial_instantaneous_forward_variance} hold, and let Assumptions~\ref{assump:kernel_function} and \ref{assump:no_degenerate} hold for $K_1$ and $K_2$ with common exponents $a_1, a_2$ and constant $c$, and let $\varepsilon := a_1 \wedge \frac{a_2}{2}$. Let $\Kc \subset \Rb$ be a compact set of log-strikes and assume that there exists $C_\Ec > 0$ such that, for all $k \in \Kc$ and $\Delta \le \Delta_0$, the remainders $\Ec_C := C - C_P$ and $\Ec_F := F - F_P$ satisfy
\eqlnostar{eq:remainder_assumption_mixed}{
|\Ec_C| \le C_\Ec\, \Delta^{3\varepsilon},
\qquad
|\Ec_F| \le C_\Ec\, \Delta^{3\varepsilon}.
}
Let $\Sigma_0(k)$ be the unique positive solution of
\eqlnostar{eq:sigma0_mixed}{
C^{\BS}\bigl(\ln(F_{P,0}), k, \Sigma_0(k)\bigr) = C_{P,0}(k),
}
and define the first-order correction
\eqlnostar{eq:sigma1_mixed}{
\Sigma_1(k) := \frac{C_{P,1}(k) - \Phi\bigl(d_0(k)\bigr)\, F_{P,1}}{\Vega\bigl(\ln(F_{P,0}), k, \Sigma_0(k)\bigr)}
= \frac{\Eb\bigl[ q(Z) \sqrt{h(Z)} \bigl( \ind_{Z > A} - \Phi(d_0(k)) \bigr) \bigr]}{2\, F_{P,0} \sqrt{T}\, \phi\bigl(d_0(k)\bigr)},
}
where $d_0(k) := d_1\bigl(\ln(F_{P,0}), k, \Sigma_0(k)\bigr)$ and $q$ is given in \eqref{eq:weights_q_mixed}. Then there exist constants $0 < s_{\min} \le s_{\max}$, $C > 0$, and $0 < \Delta_1 \le \Delta_0$, depending only on $\Kc$, $T$, $\lambda$, $C_\Ec$, and the constants in Assumptions~\ref{assump:initial_instantaneous_forward_variance}--\ref{assump:no_degenerate}, such that, for all $\Delta \le \Delta_1$ and $k \in \Kc$:
\begin{enumerate}
\item[(i)] $s_{\min} \le \Sigma_0(k) \le s_{\max}$ and $|\Sigma_1(k)| \le C \Delta^{2\varepsilon}$;
\item[(ii)] the VIX implied volatility satisfies
\eqlnostar{eq:vix_iv_error_bound_mixed_kernel}{
\bigl| \sbs(k,T) - \Sigma_0(k) - \Sigma_1(k) \bigr| \le C \Delta^{3\varepsilon}.
}
More precisely,
\eqlnostar{eq:vix_iv_error_structure_mixed_kernel}{
\sbs(k,T) = \Sigma_0(k) + \Sigma_1(k) + \frac{\Ec_C - \Phi\bigl(d_0(k)\bigr)\, \Ec_F}{\Vega\bigl(\ln(F_{P,0}), k, \Sigma_0(k)\bigr)} + R(k),
\quad
|R(k)| \le C \Delta^{4\varepsilon}.
}
\end{enumerate}
\end{theorem}
The proof is given in Appendix~\ref{sec:vix_iv_mixed_kernel_rigorous_proof}.

\begin{remark}
\label{re:choice_of_sigma1}
The form of $\Sigma_1$ in \eqref{eq:sigma1_mixed} is obtained by matching first-order terms. Recall that $\sbs(k,T)$ is defined by $C^{\BS}(\ln(F), k, \sbs(k,T)) = C$, see \eqref{eq:vix_iv}, and that, by \eqref{eq:weak_approx_gaussian_rep} and \eqref{eq:remainder_assumption_mixed},
\eqstar{
C = C_{P,0}(k) + C_{P,1}(k) + \Ec_C,
\qquad
F = F_{P,0} + F_{P,1} + \Ec_F,
}
where the first-order terms $C_{P,1}(k)$ and $F_{P,1}$ are of order $\Delta^{2\varepsilon}$ and the remainders are of order $\Delta^{3\varepsilon}$. At the zeroth order, the pair $(F_{P,0}, C_{P,0}(k))$ is matched exactly by the volatility $\Sigma_0(k)$ in \eqref{eq:sigma0_mixed}. At the first order, both the call price and the forward move: the call price by $C_{P,1}(k)$, and the log-forward by $\ln(F) - \ln(F_{P,0}) = F_{P,1}/F_{P,0}$, up to terms of order $\Delta^{3\varepsilon}$. Linearizing $(x,\sigma) \mapsto C^{\BS}(x,k,\sigma)$ around $(\ln(F_{P,0}), \Sigma_0(k))$, and using $\partial_x C^{\BS}(x,k,\sigma) = \ee^x \Phi(d_1(x,k,\sigma))$, the BS price with log-forward $\ln(F)$ and volatility $\Sigma_0(k) + \Sigma_1(k)$ equals, up to terms of second order in $(F_{P,1}, \Sigma_1(k))$ and to the remainder $\Ec_F$,
\eqstar{
C_{P,0}(k) + \Phi\bigl(d_0(k)\bigr)\, F_{P,1} + \Sigma_1(k)\, \Vega\bigl(\ln(F_{P,0}), k, \Sigma_0(k)\bigr).
}
Requiring this quantity to equal $C_{P,0}(k) + C_{P,1}(k)$ gives \eqref{eq:sigma1_mixed}. 
\end{remark}

Theorem~\ref{thm:vix_iv_mixed_kernel_rigorous} separates the mixed VIX smile into two parts. The proxy smile $\Sigma_0$ carries the positive skew generated by the mixture of two log-normal components. The correction $\Sigma_1$ is of order $\Delta^{2\varepsilon}$ and is explicit, given $\Sigma_0$, as a ratio of one-dimensional Gaussian integrals. It is linear in the coefficients $\gamma_{i,j}$, which contain the remaining dependence on the kernels and on the initial forward variance curve over $[T,T+\Delta]$. Moreover, $\Sigma_0$ depends on the model only through $\lambda$ and $(\mu_{P,1},\mu_{P,2},\sigma_{P,1},\sigma_{P,2})$, which lie in a fixed compact set for all $\Delta\le\Delta_0$. Therefore, it can be tabulated offline and the BS inversion can be removed from the online evaluation. Finally, \eqref{eq:vix_iv_error_structure_mixed_kernel} shows that the implied volatility error is the combination $\Ec_C-\Phi(d_0(k))\,\Ec_F$ of price remainders divided by the Vega. Thus, passing from prices to implied volatilities using the weak approximation only adds a term of order $\Delta^{4\varepsilon}$.


\subsubsection{Numerical tests}
\label{sec:rigorous_mixed_numerics}
By \eqref{eq:vix_iv_error_structure_mixed_kernel}, the error in the implied volatility expansion \eqref{eq:mixed_proxy_iv_expansion_implicit} (referred to as \textit{Expansion}) coincides, up to a term of order $\Delta^{4\varepsilon}$, with the error in the implied volatility obtained after using BS inversion of the weak price approximation (\textit{Approx}), so that we do not include the relative error plots for \textit{Expansion}.
In Figure~\ref{fig:mixed_cases}, we compare the performance of \textit{Expansion} and \textit{Approx} for the parameters in Table~\ref{tab:mixed_main_parameters}. We use the same reference methods, coefficient grid, maturities, and log-moneyness grid as in Section~\ref{sec:numerical_tests_single_kernel}. For computing the Gaussian integrals in \eqref{eq:zeroth_first_order_mixed} in both \textit{Approx} and \textit{Expansion}, we use $128$ Gauss--Legendre quadrature nodes on $[A, b]$ (calls) and $[-b, b]$ (futures), with $b \ge 14$.

\begin{table}[H]
\centering
\caption{Parameters used in Figure~\ref{fig:mixed_cases}.}
\label{tab:mixed_main_parameters}
\begin{tabular}{clcc}
\toprule
Scenario & Model & Vol-of-vol and weight & Kernel parameter \\
\midrule
1 & Mixed standard Bergomi & $(\omega_1,\omega_2,\lambda)=(10,2,0.2)$ & $\kappa=1$ \\
2 & Mixed standard Bergomi & $(\omega_1,\omega_2,\lambda)=(0.5,6,0.3)$ & $\kappa=1$ \\
3 & Mixed rough Bergomi & $(\eta_1,\eta_2,\lambda)=(1.4,0.7,0.3)$ & $H=0.1$ \\
4 & Mixed rough Bergomi & $(\eta_1,\eta_2,\lambda)=(2,0.2,0.4)$ & $H=0.1$ \\
\bottomrule
\end{tabular}
\end{table}
It is clear from the figure that \textit{Expansion} and \textit{Approx} closely match the reference implied volatility. In Scenarios~1--4, the maximum absolute relative errors of \emph{Expansion} are $8.8\times10^{-5}\%$, $4.1\times10^{-4}\%$, $0.331\%$, and $0.554\%$. Those of \emph{Approx} are $5.8\times10^{-5}\%$, $2.1\times10^{-4}\%$, $0.361\%$, and $0.510\%$, and those of the proxy $\Sigma_0$ are $0.046\%$, $0.121\%$, $1.281\%$, and $1.752\%$. Thus, the correction $\Sigma_1$ brings the proxy errors to the level of \emph{Approx}. Both \emph{Approx} and \emph{Expansion} take $0.44$--$0.46\,\ms$ for a ten-strike three-month smile (Table~\ref{tab:mixed_main_timings}), including the root-finding for $A$, with the timing protocol of Section~\ref{sec:numerical_tests_single_kernel}.

It is clear that the implied volatility expansion $\Sigma_0 + \Sigma_1$ is numerically superior to the approximation derived in Section~\ref{sec:vix_iv_proxy_mixed_kernel}, at the cost of numerically computing the zeroth-order term, with a runtime comparable to that of \textit{Approx}.
\begin{table}[H]
\centering
\caption{Median runtime in milliseconds for a ten-strike three-month smile.}
\label{tab:mixed_main_timings}
\begin{tabular}{crrr}
\toprule
Scenario & Approx & Expansion & Reference \\
\midrule
1 & 0.44 & 0.44 & 20.68 \\
2 & 0.46 & 0.45 & 20.14 \\
3 & 0.45 & 0.44 & 10669.95 \\
4 & 0.45 & 0.45 & 10660.82 \\
\bottomrule
\end{tabular}
\end{table}

\begin{figure}[H]
\centering
\includegraphics[width=0.49\linewidth]{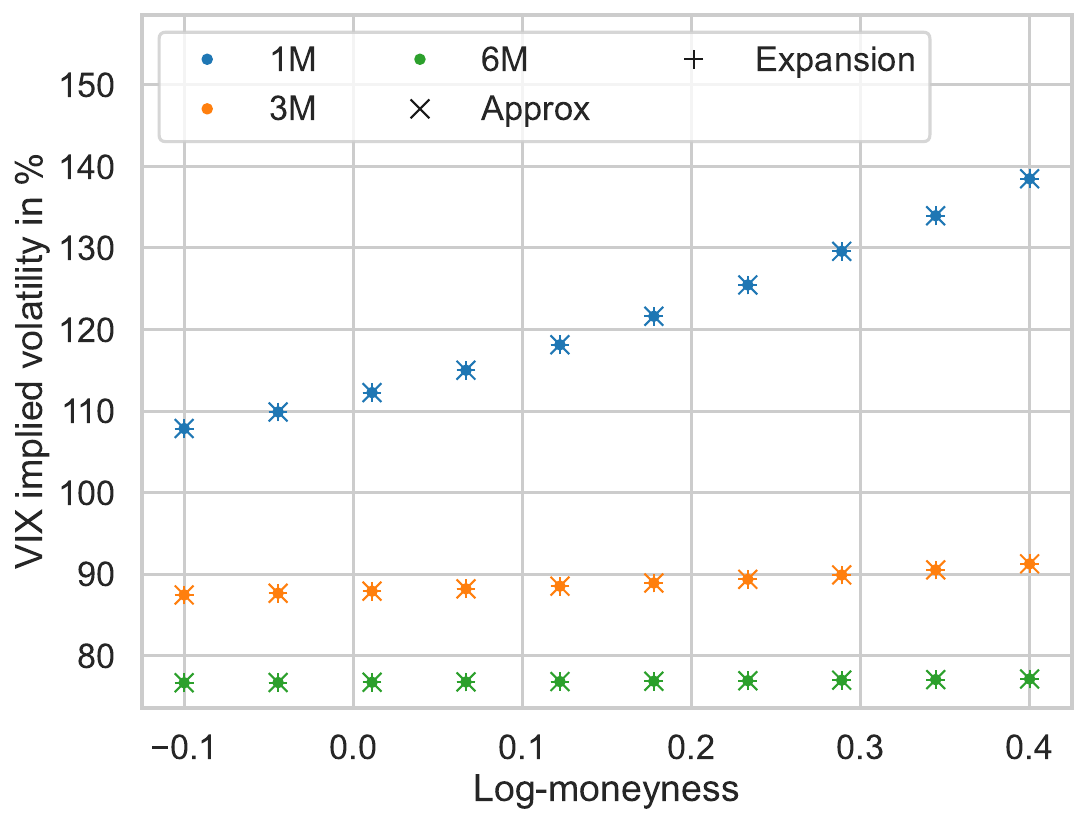}\hfill
\includegraphics[width=0.49\linewidth]{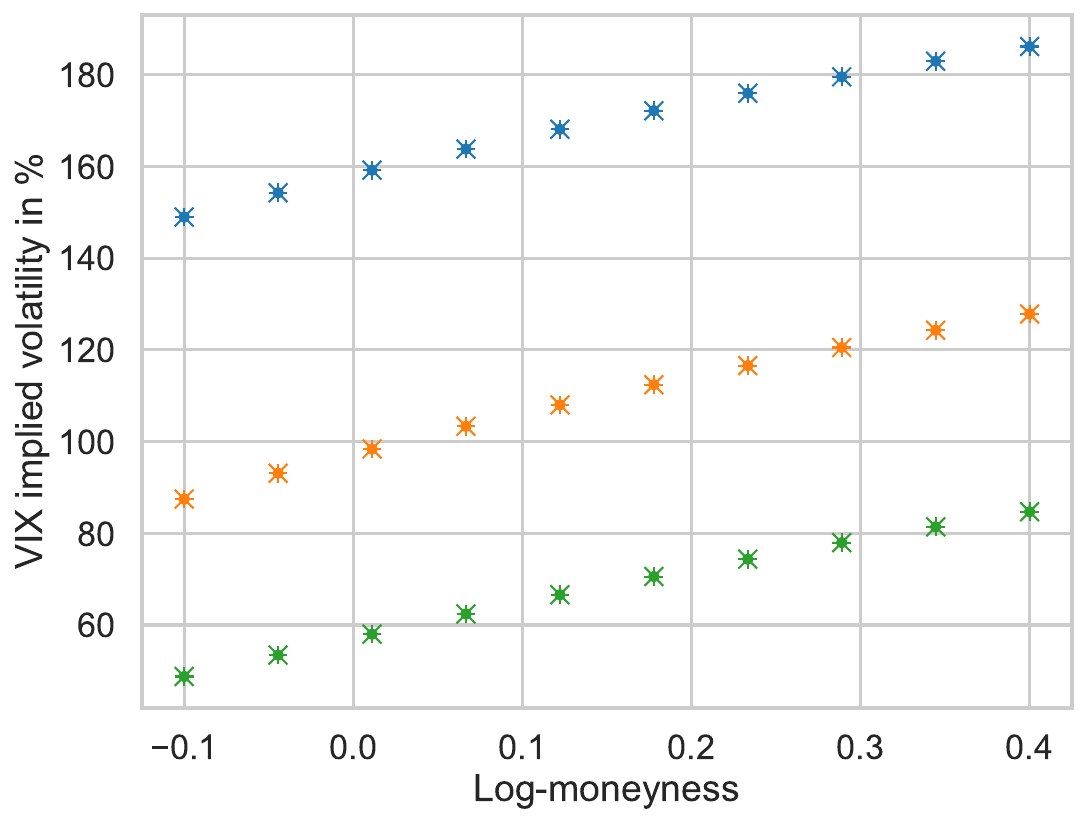}\\[1ex]
\includegraphics[width=0.49\linewidth]{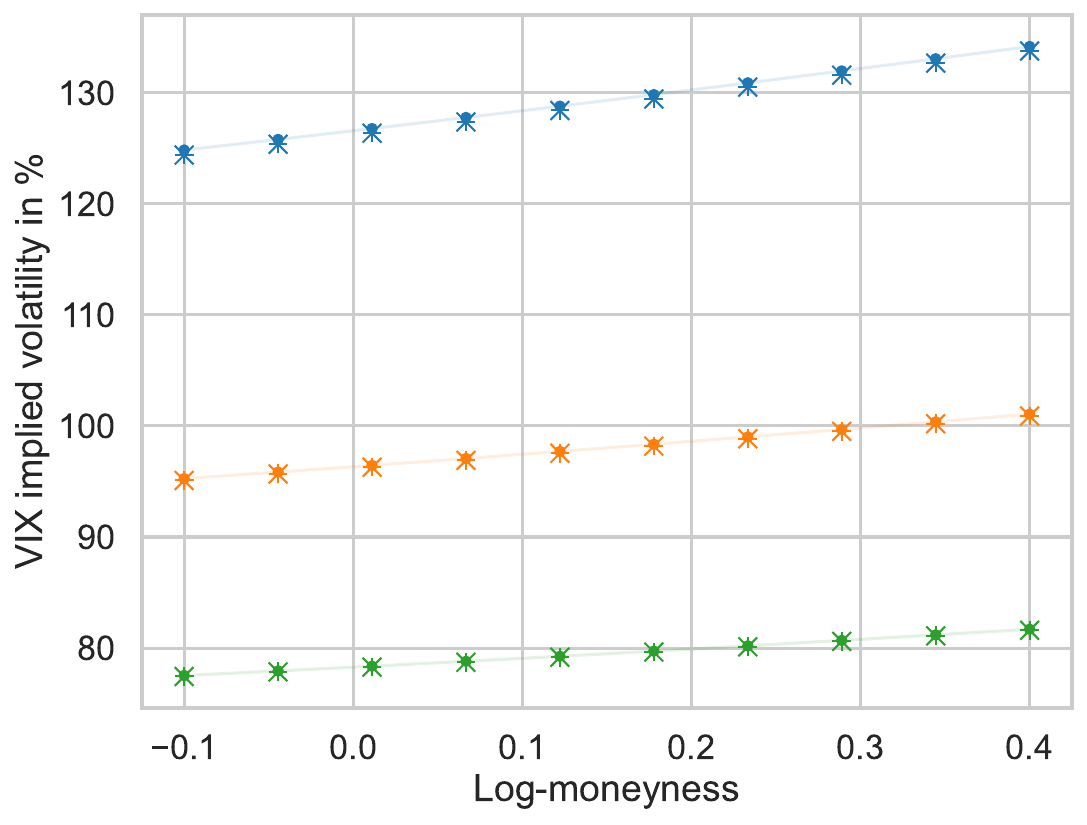}\hfill
\includegraphics[width=0.49\linewidth]{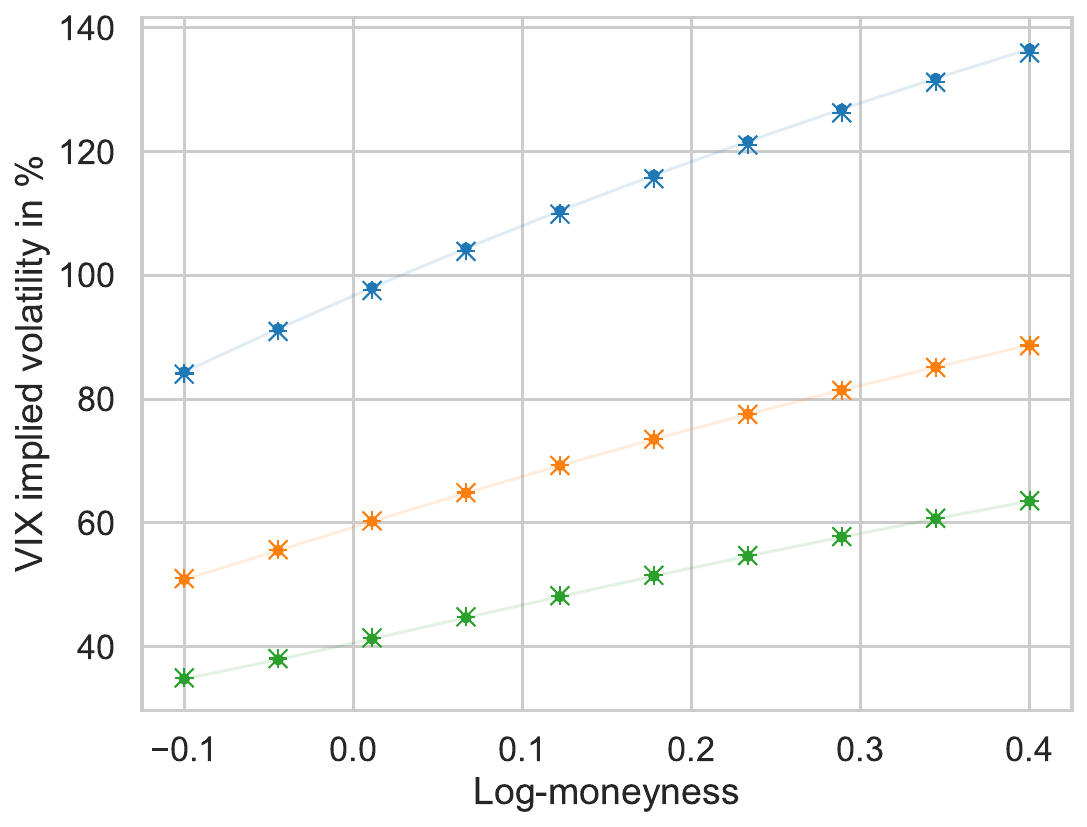}
\caption{VIX smiles for Scenarios~1--4, read from left to right and top to bottom. Top: mixed standard Bergomi; bottom: mixed rough Bergomi. Dots, crosses, and plus signs denote the reference, \emph{Approx}, and \emph{Expansion}, respectively. Reference VIX futures $F$ for the 1-, 3-, and 6-month maturities: 0.2178, 0.1986, 0.1862 (Scenario~1); 0.2072, 0.1752, 0.1566 (Scenario~2); 0.2233, 0.2106, 0.2010 (Scenario~3); 0.2236, 0.2130, 0.2061 (Scenario~4).}
\label{fig:mixed_cases}
\end{figure}
\section{Conclusion}
\label{sec:conclusion}

In this work, we derived VIX implied volatility expansions for single-kernel and mixed-kernel forward variance models, building on weak VIX option price approximations. In the single-kernel case, we proved that the explicit expansion inherits the order of accuracy of the price expansion, and that passing from prices to implied volatilities only adds a term of higher order. The formulas are highly accurate numerically and are about four times faster than the BS inversion of the price approximation. In the mixed-kernel case, we followed two routes. First, we showed that the Hermite approximation of the price is exactly a Black--Scholes price with a shifted volatility, up to explicit residual terms, which yields a fully explicit expansion. Provided the two proxy volatilities are not too far apart, this expansion reproduces the pronounced positive skew of mixed-kernel VIX smiles, whereas single-kernel models only produce nearly flat smiles with a slight positive skew. In more extreme parameter regimes, the neglected volatility and shift residuals are no longer small, and the expansion loses accuracy. Second, we derived an expansion whose zeroth-order term is the implied volatility of the mixed log-normal proxy, computed numerically, and whose first-order correction is explicit and linear in the coefficients of the price expansion. We proved that this expansion has the same order of accuracy as in the single-kernel case. Numerically, its errors are of the same magnitude as those of the weak price approximation in all scenarios considered, including those where the explicit expansion fails. A rigorous bound on the weak price approximation error for the non-smooth VIX call, put, and futures payoffs in mixed-kernel models, on which our error bound relies, as well as an error bound for the fully explicit mixed-kernel expansion, are left for future research.

\section*{Disclosure statement}
No potential conflict of interest was reported by the authors.

\section*{Declaration: AI use}
The authors used Claude Code (Anthropic) and Codex (OpenAI) to assist with manuscript review, coding and writing. The authors assume responsibility for all content.

\printbibliography
\appendix

\section{Verification of Assumption~\ref{assump:no_degenerate} for the exponential and power-law kernels}
\label{sec:sufficient_conditions}
We verify Assumption~\ref{assump:no_degenerate} for the exponential and power-law kernels \eqref{eq:standard_kernel}--\eqref{eq:power_kernel}. First, write $w_u := \xi_0^u / \VIXsqExp$ for the weights appearing in \eqref{eq:gamma}-\eqref{eq:delta}, which are non-negative and satisfy $\frac{1}{\Delta}\int_T^{T+\Delta} w_u \du = 1$, so that $\frac{1}{\Delta}\int_T^{T+\Delta} w_v (\cdot) \dv$ is an average. Suppose that there exist $\beta_1, \beta_2 > 0$ and $C > 0$ such that, for all $u, v \in [T, T+\Delta_0]$, we have the following
\eqlnostar{eq:sufficient_condition_kernel}{
\Bigl| \int_0^T \bigl( K^u(t)^2 - K^v(t)^2 \bigr) \dt \Bigr| \le C |u-v|^{\beta_1},
\qquad
\int_0^T \bigl( K^u(t) - K^v(t) \bigr)^2 \dt \le C |u-v|^{\beta_2}.
}
Since 
\eqstar{
K^u(t)^2 - \frac{1}{\Delta}\int_T^{T+\Delta} w_v K^v(t)^2 \dv = \frac{1}{\Delta}\int_T^{T+\Delta} w_v \bigl(K^u(t)^2 - K^v(t)^2\bigr) \dv,
}
the first condition gives $\Gamma_{\Delta,T,p} \le C\Delta^{\beta_1}$ for every $p>0$. 

Similarly, by Jensen's inequality applied to the average in $v$, 
\eqstar{
\int_0^T \bigl(K^u(t) - \frac{1}{\Delta}\int_T^{T+\Delta} w_v K^v(t) \dv\bigr)^2 \dt \le \frac{1}{\Delta}\int_T^{T+\Delta} w_v \int_0^T (K^u(t) - K^v(t))^2 \dt \dv \le C\Delta^{\beta_2},
}
so that $\Lambda_{\Delta,T,p} \le C\Delta^{\beta_2}$. Thus \eqref{eq:gamma}-\eqref{eq:delta} hold with $(a_1, a_2) = (\beta_1, \beta_2)$. Moreover, if the first condition in Assumption \ref{assump:kernel_function} holds uniformly in $u$, that is,
\eqstar{
\sup_{u \in [T,T+\Delta_0]} \int_0^T K^u(t)^2 \dt < \infty,
}
and $K^u(t) \ge \underline{K} > 0$ for all $t \le T \le u \le T+\Delta_0$, then $|\mu_P|$ is bounded (using Assumption~\ref{assump:initial_instantaneous_forward_variance} for the term $\ln \VIXsqExp$). In addition, $\sigma_P^2 \le \sup_u \int_0^T K^u(t)^2 \dt$ by Jensen's inequality, and $\sigma_P \ge \underline{K}\sqrt{T}$, uniformly in $\Delta \le \Delta_0$.

Both kernels considered in this work satisfy these conditions. For the exponential kernel \eqref{eq:standard_kernel}, 
\eqstar{
&\int_0^T (K^u(t)^2 - K^v(t)^2)\dt = \omega^2 (\ee^{-2\kappa u} - \ee^{-2\kappa v}) \int_0^T \ee^{2\kappa t}\dt,\\  &\int_0^T (K^u(t) - K^v(t))^2 \dt = \omega^2 (\ee^{-\kappa u} - \ee^{-\kappa v})^2 \int_0^T \ee^{2\kappa t}\dt,
}
so that $(\beta_1, \beta_2) = (1, 2)$ and the error exponent in Theorem~\ref{thm:vix_opt_price_approx_single_kernel} is $3(a_1 \wedge \frac{a_2}{2}) = 3$. For the power-law kernel \eqref{eq:power_kernel} with $H \in (0,\frac12)$, we get
\eqstar{
\int_0^T (K^u(t)^2 - K^v(t)^2)\dt = \frac{\eta^2}{2H}\bigl[ (u^{2H} - v^{2H}) - ((u-T)^{2H} - (v-T)^{2H}) \bigr].
}
Since $u \mapsto u^{2H}$ is $2H$-H\"older continuous on $[0,\infty)$, the above is bounded by $C|u-v|^{2H}.$ Moreover, for $T \le u < v \le T+\Delta_0$, the change of variable $t = T - (v-u) r$ gives
\eqstar{
\int_0^T \bigl( K^u(t) - K^v(t) \bigr)^2 \dt 
&\le \eta^2 (v-u)^{2H} \int_0^\infty \Bigl( (c_0 + r)^{H-\frac12} - (c_0 + 1 + r)^{H-\frac12} \Bigr)^2 \dd r \\
&\le \eta^2 (v-u)^{2H} \int_0^\infty \Bigl( r^{H-\frac12} - (1 + r)^{H-\frac12} \Bigr)^2 \dd r,
}
with $c_0 := (u-T)/(v-u) \ge 0$. In the last inequality we have used the fact that $r \mapsto r^{H-\frac12} - (1+r)^{H-\frac12}$ is positive and decreasing. Moreover, the last integral above is finite. First we make the following split
\eqstar{
&\int_0^\infty \Bigl( r^{H-\frac12} - (1 + r)^{H-\frac12} \Bigr)^2 \dd r\\
&= \int^1_0 \Bigl( r^{H-\frac12} - (1 + r)^{H-\frac12} \Bigr)^2 \dd r + \int_1^\infty \Bigl( r^{H-\frac12} - (1 + r)^{H-\frac12} \Bigr)^2 \dd r.
}
For the first integral, we have 
\eqstar{
|r^{H-\frac12} - (1 + r)^{H-\frac12}|^2 \leq 2 + 2r^{2H-1}, \qquad r \in (0,1).
}
Thus, 
\eqstar{
& \int^1_0 \Bigl( r^{H-\frac12} - (1 + r)^{H-\frac12} \Bigr)^2 \dd r\leq 2  + \frac1{H}.
}
For the second integral, let $f(r) = r^{-\alpha} - (1 + r)^{-\alpha}$ with $\alpha: = \frac12 - H,$ and $g(x) = x^{-\alpha}$. Then, by the mean value theorem, we have for some $\zeta \in (r, r+1)$
\eqstar{
f(r) = g(r) - g(r+1) = -g^\prime(\zeta) = \alpha \zeta^{-\alpha -1} \leq \alpha r^{-\alpha -1}.
}
Thus, 
\eqstar{
\int^\infty_1 f^2(r) \dd r \leq \int^\infty_1 \alpha^2 r^{-2\alpha -2} \dd r = \frac{\alpha^2}{2\alpha +1} < \infty.
}
Hence $(\beta_1, \beta_2) = (2H, 2H)$ and the error exponent is $3(a_1 \wedge \frac{a_2}{2}) = 3H$.

For the other conditions to hold, we note that for exponential and power-law kernels, $K^u(t)$ is bounded below by $\omega\ee^{-\kappa(T+\Delta_0)}$ and $\eta (T+\Delta_0)^{H-\frac12}$, respectively, and $\int_0^T K^u(t)^2\dt$ is bounded uniformly in $u \in [T, T+\Delta_0]$ in both cases, since
\eqstar{
\int_0^T \omega^2 \ee^{-2\kappa(u-t)} \dt &= \frac{\omega^2}{2\kappa} \bigl( \ee^{-2\kappa(u-T)} - \ee^{-2\kappa u} \bigr) \le \frac{\omega^2}{2\kappa},
\\
\int_0^T \eta^2 (u-t)^{2H-1} \dt &= \frac{\eta^2}{2H} \bigl( u^{2H} - (u-T)^{2H} \bigr) \le \frac{\eta^2}{2H} (T+\Delta_0)^{2H}.
}
In particular, the exponential integrability condition in Assumption~\ref{assump:kernel_function} holds for both kernels, and Assumption~\ref{assump:no_degenerate} is satisfied with $(a_1, a_2) = (1,2)$ for the exponential kernel and $(a_1, a_2) = (2H, 2H)$ for the power-law kernel.

\section{Proofs}
\subsection{Proof of Lemma~\ref{lem:order_of_coefficients}}
\label{sec:proof of lem:order_of_coefficients}
As in Appendix~\ref{sec:sufficient_conditions}, we write $w_u := \xi_0^u / \VIXsqExp$ and $\Kbar(t) := \frac{1}{\Delta}\int_T^{T+\Delta} w_v K^v(t) \dv$. Set $a_u := \int_0^T \bigl(K^u(t)^2 - \frac{1}{\Delta}\int_T^{T+\Delta} w_v K^v(t)^2 \dv\bigr)\dt$ and $b_u := \int_0^T \Kbar(t)\bigl(K^u(t) - \Kbar(t)\bigr)\dt$, so that, by \eqref{eq:coefficients},
\eqstar{
\gamma_1 &= \frac{1}{8\Delta}\int_T^{T+\Delta} w_u\, a_u^2 \du + \frac{1}{2\Delta}\int_T^{T+\Delta} w_u\, \|K^u - \Kbar\|^2_{L^2(0,T)} \du\\
&= \frac18 \Gamma_{\Delta,T,2}^2 + \frac12 \Lambda_{\Delta,T,1}, \\
\gamma_2 &= -\frac{1}{2\Delta}\int_T^{T+\Delta} w_u\, a_u b_u \du, 
\qquad
\gamma_3 = \frac{1}{2\Delta}\int_T^{T+\Delta} w_u\, b_u^2 \du.
}
The bound on $\gamma_1$ follows from \eqref{eq:gamma}--\eqref{eq:delta}. Note that $\sigma_P^2 = \int_0^T \Kbar(t)^2 \dt$ by \eqref{eq:variance}. By the Cauchy--Schwarz inequality on $[0,T]$, $|b_u| \le \sigma_P \|K^u - \Kbar\|_{L^2(0,T)}$. Hence, 
\eqstar{
\gamma_3 \le \frac{\sigma_P^2}{2}\Lambda_{\Delta,T,1} \le \frac{c^3}{2}\Delta^{a_2}.
}
By the Cauchy--Schwarz inequality, we get
\eqstar{
|\gamma_2| \le \frac{1}{2} \Bigl( \frac{1}{\Delta}\int_T^{T+\Delta} w_u\, a_u^2 \du \Bigr)^{\frac12} \Bigl( \frac{1}{\Delta}\int_T^{T+\Delta} w_u\, b_u^2 \du \Bigr)^{\frac12}
\le \frac{\sigma_P}{2}\, \Gamma_{\Delta,T,2}\, \Lambda_{\Delta,T,1}^{1/2} \le \frac{c^{5/2}}{2}\Delta^{a_1 + \frac{a_2}{2}}.
}
The last assertion follows since $\Delta_0 \le 1$, $2a_1 \ge 2\varepsilon$, $a_2 \ge 2\varepsilon$, and $a_1 + \frac{a_2}{2} \ge 2\varepsilon$.
\qed

\subsection{Proof of Theorem~\ref{thm:vix_iv_single_kernel}}
\label{sec:proof of thm:vix_iv_single_kernel}
Throughout the proof, $C$ denotes a positive constant depending only on $\Kc$, $T$ and the constants in Assumptions \ref{assump:initial_instantaneous_forward_variance}--\ref{assump:no_degenerate}, whose value may change from line to line. Recall $\delta$ from \eqref{eq:vix_fut_proxy_single_kernel}, and set
\eqlnostar{eq:delta_and_dsigma}{
\Delta\sigma := \widehat{\sigma}_P(k,T) - \frac{\wtsp}{2}
= \frac{\gamma_2}{2 \wtsp T} + \frac{3\gamma_3}{8 \wtsp T} - \frac{\gamma_3}{\wtsp^3 T^2} (x_P-k).
}
By Lemma \ref{lem:order_of_coefficients}, we have the following bound, uniformly in $k \in \Kc$ and $\Delta \le \Delta_0$,
\eqlnostar{eq:delta_dsigma_bounds_1}{
|\delta| + |\Delta\sigma| \le C \Delta^{2\varepsilon}. 
}
Moreover, from Assumption \ref{assump:no_degenerate}, we get that 
\eqlnostar{eq:delta_dsigma_bounds_2}{
x_P \in [x_{\min}, x_{\max}] := \Bigl[-\frac{c}{2}, \frac{c}{2} + \frac{c^2}{8}\Bigr],
\quad
\frac{\wtsp}{2} \in [\sigma_{\min}, \sigma_{\max}] := \Bigl[\frac{1}{2c\sqrt{T}}, \frac{c}{2\sqrt{T}}\Bigr].
}
Let us choose $\Delta_1 \le \Delta_0$ small enough so that $|\delta| \le \frac14$ and $|\Delta\sigma| \le \frac{\sigma_{\min}}{4}$ for all $\Delta \le \Delta_1$. Next, we introduce the compact set
\eqstar{
\Dc :=  \Xc \times \Kc \times J, 
\qquad \Xc: = [x_{\min} - 1, x_{\max} + 1], \quad
J := \Bigl[\frac{\sigma_{\min}}{2}, 2\sigma_{\max}\Bigr].
}
Since $(x,k,\sigma) \mapsto C^\BS(x,k,\sigma)$ is smooth on $\Rb\times\Rb\times(0,\infty)$, its first- and second-order partial derivatives are bounded on $\Dc$ by a constant $M$. Moreover, since $\Vega$ is continuous and positive on the compact set $\Dc$, it is bounded from below on $\Dc$ by a constant $v_0 > 0$.
\par

\emph{Step 1.} Write
\eqstar{
\partial_x^2= (\partial_x^2 - \partial_x) + \partial_x, \quad \partial_x^3 = (\partial_x^3 - \partial_x^2) + (\partial_x^2 - \partial_x) + \partial_x. 
}
Then, by \eqref{eq:vix_call_proxy_single_kernel}, and using the identities \eqref{eq:identities_bs_price} with $\sigma = \wtsp/2$, we obtain the exact identity
\eqlnostar{eq:exact_rewriting_CP}{
C_P = C^{\BS} \bigl(x_P, k, \tfrac{\wtsp}{2}\bigr)
+ \delta\,\partial_x C^{\BS} \bigl(x_P, k, \tfrac{\wtsp}{2}\bigr)
+ \Delta\sigma\,\Vega \bigl(x_P, k, \tfrac{\wtsp}{2}\bigr).
}
Moreover, by \eqref{eq:vix_fut_proxy_single_kernel}, we have $\ln(F_P) = x_P + \ln(1+\delta)$.
\par

\emph{Step 2.} Since $|\ln(1+\delta)| \le 2|\delta| \le \frac12$ and $|\Delta\sigma| \le \sigma_{\min}/4$, the segment joining $(x_P, \wtsp/2)$ to $(\ln(F_P), \widehat{\sigma}_P)$ lies in $\Xc \times J$. A first-order Taylor expansion of $(x,\sigma) \mapsto C^\BS(x,k,\sigma)$ around $(x_P, \wtsp/2)$ with increments $(\ln(1+\delta), \Delta\sigma)$ therefore gives
\eqstar{
C^{\BS}\bigl(\ln(F_P), k, \widehat{\sigma}_P\bigr)
&= C^{\BS} \bigl(x_P, k, \tfrac{\wtsp}{2}\bigr)
+ \ln(1+\delta)\,\partial_x C^{\BS} \bigl(x_P, k, \tfrac{\wtsp}{2}\bigr)
+ \Delta\sigma\,\Vega \bigl(x_P, k, \tfrac{\wtsp}{2}\bigr)
+ R_2, \\
|R_2| &\le \frac{M}{2}\bigl(|\ln(1+\delta)| + |\Delta\sigma|\bigr)^2.
}
From \eqref{eq:exact_rewriting_CP}, and using $|\ln(1+\delta) - \delta| \le \delta^2,$ for $|\delta| \le \frac12$, we obtain the following bound from \eqref{eq:delta_dsigma_bounds_1}
\eqlnostar{eq:step2_bound}{
\bigl| C^{\BS}\bigl(\ln(F_P), k, \widehat{\sigma}_P\bigr) - C_P \bigr| 
\le M \delta^2 + \frac{M}{2}\bigl(2|\delta| + |\Delta\sigma|\bigr)^2 
\le C \Delta^{4\varepsilon}.
}
\par

\emph{Step 3.} We get the following bound by applying Theorem \ref{thm:vix_opt_price_approx_single_kernel} and Remark \ref{re:payoff_func} to the call and futures payoffs
\eqlnostar{eq:call and futures bounds}{
|C - C_P| \le C\Delta^{3\varepsilon}, 
\qquad
|F - F_P| \le C\Delta^{3\varepsilon}.
}
Since $x_P \in [x_{\min}, x_{\max}]$, $F_P = \ee^{x_P}(1+\delta) \ge \frac12 \ee^{x_{\min}}$. By choosing $\Delta_1$ small enough, we get from the above error bound in \eqref{eq:call and futures bounds} that $F \ge \frac14 \ee^{x_{\min}}$. Therefore, 
\eqstar{
|\ln(F) - \ln(F_P)| \le 4\ee^{-x_{\min}}|F - F_P| \le C\Delta^{3\varepsilon}.
}
Reducing $\Delta_1$ further so that $|\ln(F) - \ln(F_P)| \le \frac12$, we get $|\ln(F) - x_P| \le |\ln(1+\delta)| + \frac12 \le 1$, and hence $\ln(F) \in \Xc$. Next, combining \eqref{eq:step2_bound} and the bound $|\partial_x C^\BS| \le M$ on $\Dc$, we get
\eqlnostar{eq:step3_bound}{
\bigl| C^{\BS}\bigl(\ln(F), k, \widehat{\sigma}_P\bigr) - C \bigr|
&\le \bigl| C^{\BS}\bigl(\ln(F), k, \widehat{\sigma}_P\bigr) - C^{\BS}\bigl(\ln(F_P), k, \widehat{\sigma}_P\bigr) \bigr| \nonumber\\
&\quad + \bigl| C^{\BS}\bigl(\ln(F_P), k, \widehat{\sigma}_P\bigr) - C_P \bigr| + |C_P - C| \nonumber\\
&\le M\, |\ln(F) - \ln(F_P)| + C\Delta^{4\varepsilon} + C\Delta^{3\varepsilon} 
\le C\Delta^{3\varepsilon}.
}
\par

\emph{Step 4.} 
Let $r := \sigma_{\min}/4$ which is the upper bound for $| \Delta \sigma |.$ Also let $J' := [\frac{3\sigma_{\min}}{4}, \sigma_{\max} + \frac{\sigma_{\min}}{4}]$, 
so that $[\sigma_0 - r, \sigma_0 + r] \subset J$ for every $\sigma_0 \in J'$. 
We claim that, for all $(x,k) \in \Xc \times \Kc$, $\sigma_0 \in J'$ and $\sigma > 0$,
\eqlnostar{eq:coercivity_bs}{
\bigl| C^{\BS}(x,k,\sigma) - C^{\BS}(x,k,\sigma_0) \bigr| \ge v_0 \min\{ |\sigma - \sigma_0|,\, r \}.
}
\begin{itemize}
\item Case $|\sigma - \sigma_0| \le r$: By the mean value theorem, there exists $\bar\sigma = \lambda \sigma + (1-\lambda) \sigma_0$, with $\lambda \in (0,1)$, such that
\eqstar{
C^{\BS}(x,k,\sigma) - C^{\BS}(x,k,\sigma_0) = \partial_\sigma C^{\BS}(x,k,\bar\sigma)(\sigma - \sigma_0).
}
Next, since $\Vega \ge v_0$ on $\Dc$, we get the desired result.
\item Case $\sigma > \sigma_0 + r$: Since $\sigma \mapsto C^{\BS}(x,k,\sigma)$ is increasing, 
$C^{\BS}(x,k,\sigma) - C^{\BS}(x,k,\sigma_0) \ge C^{\BS}(x,k,\sigma_0 + r) - C^{\BS}(x,k,\sigma_0) \ge v_0 r$ 
by the first case. The case $\sigma < \sigma_0 - r$ is symmetric. 
\end{itemize}
Now apply the result in \eqref{eq:coercivity_bs} with $x = \ln(F) \in \Xc$ (concluded in Step 3), 
$\sigma_0 = \widehat{\sigma}_P(k,T) \in J'$ (by \eqref{eq:delta_dsigma_bounds_2} and $|\Delta\sigma| \le \frac{\sigma_{\min}}{4}$) and $\sigma = \sbs(k,T)$, 
for which $C^{\BS}(\ln(F), k, \sbs(k,T)) = C$ by \eqref{eq:vix_iv}. Together with \eqref{eq:step3_bound}, 
\eqstar{
v_0 \min\bigl\{ |\sbs(k,T) - \widehat{\sigma}_P(k,T)|,\, r \bigr\} \le C \Delta^{3\varepsilon}.
}
Reducing $\Delta_1$ so that $C \Delta_1^{3\varepsilon} < v_0 r$, the minimum is attained by the first term, 
and \eqref{eq:vix_iv_error_bound_single_kernel} follows. 
\qed

\subsection{Proof of Proposition~\ref{prop:rewrite_vix_opt_price_approx_mixed_kernel}}
\label{sec:proof of prop:rewrite_vix_opt_price_approx_mixed_kernel}
Fix $j \in \{1,2\}$ and let $j' := 3-j$. Since $\Delta\sigma_{P,j} = \sigma_{P,j'} - \sigma_{P,j}$, the definition of $C_j$ gives, for all $y \in \Rb$,
\eqstar{
\lambda_{j'}\, \ee^{\mu_{P,j'} + \sigma_{P,j'} (y + \frac{\sigma_{P,j}}{2})}
= C_j\, \ee^{\Delta\sigma_{P,j} y}\, \lambda_j\, \ee^{\mu_{P,j} + \sigma_{P,j} (y + \frac{\sigma_{P,j}}{2})}.
}
Hence, by \eqref{eq:func_h} and \eqref{eq:func_g},
\eqlnostar{eq:h_shifted}{
\begin{gathered}
h\Bigl(y + \frac{\sigma_{P,j}}{2}\Bigr) = \lambda_j\, \ee^{\mu_{P,j} + \sigma_{P,j} y + \frac{\sigma_{P,j}^2}{2}}\, g^j_0(y)^2,
\\
\pi_j\Bigl(y + \frac{\sigma_{P,j}}{2}\Bigr) = \frac{1}{g^j_0(y)^2},
\qquad
\pi_{j'}\Bigl(y + \frac{\sigma_{P,j}}{2}\Bigr) = 1 - \frac{1}{g^j_0(y)^2}.
\end{gathered}
}
Since $\phi(y + \frac{\sigma_{P,j}}{2}) = \ee^{-\frac{\sigma_{P,j}}{2} y - \frac{\sigma_{P,j}^2}{8}}\, \phi(y)$, the first identity in \eqref{eq:h_shifted} gives $\sqrt{h(y + \frac{\sigma_{P,j}}{2})}\, \phi(y + \frac{\sigma_{P,j}}{2}) = \sqrt{\lambda_j}\, \ee^{x_{P,j}} g^j_0(y)\, \phi(y)$. Therefore, for every measurable function $f$ of polynomial growth and every interval $B \subset \Rb$, the change of variables $z = y + \frac{\sigma_{P,j}}{2}$ yields
\eqlnostar{eq:change_of_variables_mixed}{
\Eb\Bigl[ f(Z) \sqrt{h(Z)}\, \ind_{Z \in B} \Bigr]
= \sqrt{\lambda_j}\, \ee^{x_{P,j}} \int_{B - \frac{\sigma_{P,j}}{2}} f\Bigl(y + \frac{\sigma_{P,j}}{2}\Bigr)\, g^j_0(y)\, \phi(y) \dy,
}
where both sides are absolutely convergent since $g^j_0(y) \le 1 + \sqrt{C_j}\, \ee^{\Delta\sigma_{P,j} y/2}$.

Since $h(Z) > \ee^{2k}$ if and only if $Z > A$, we have $C_{P,0}(k) = \Eb[\sqrt{h(Z)}\, \ind_{Z > A}] - \ee^k \Phi(-A)$. Applying \eqref{eq:change_of_variables_mixed} with $f \equiv 1$ and $B = (A, \infty)$, resp.\ $B = \Rb$, gives the representations of $C_{P,0}(k)$ and $F_{P,0}$. The function $q$ has polynomial growth, since $|q(z)| \le |q_1(z)| + |q_2(z)|$ by $0 \le \pi_l \le 1$. Applying \eqref{eq:change_of_variables_mixed} with $f = q/2$ and the same intervals to \eqref{eq:zeroth_first_order_mixed} gives the representations of $C_{P,1}(k)$ and $F_{P,1}$, with $g^j_1 = \frac12 q(\cdot + \frac{\sigma_{P,j}}{2})\, g^j_0$. Similarly, $P_P = \ee^k \Phi(A) - \Eb[(1 + \frac12 q(Z)) \sqrt{h(Z)}\, \ind_{Z < A}]$, and \eqref{eq:change_of_variables_mixed} with $f = 1 + q/2$ and $B = (-\infty, A)$ gives \eqref{eq:vix_put_main_term_mixed_kernel}. Finally, substituting the last two identities of \eqref{eq:h_shifted} into $q = \pi_j q_j + \pi_{j'} q_{j'}$ gives the second expression of $g^j_1$ in \eqref{eq:func_g}.
\qed

\subsection{Proof of Proposition~\ref{prop:vix_opt_price_hermite_expan}}
\label{sec: proof of prop:vix_opt_price_hermite_expan}
By \eqref{eq:weak_approx_gaussian_rep}, \eqref{eq:vix_call_main_term_mixed_kernel}, and \eqref{eq:Integral_0_exact} with $x = A_j$,
\eqstar{
C_P = \sqrt{\lambda_j}\, \ee^{x_{P,j}} \sum_{i=0}^1 \bigl( I^j_{i,N}(A_j) + \rho^j_{i,N}(A_j) \bigr) - \ee^k \Phi(-A) = C^j_{P,N} + \rho^j_N.
}
Since $\He_0 = 1$, \eqref{eq:weight_func} gives $\int_\Rb g^j_i(y) \phi(y) \dy = \omega^{(i)}_{0,j}$ for $i \in \{0,1\}$. Hence, \eqref{eq:vix_fut_hermite} follows from \eqref{eq:weak_approx_gaussian_rep} and \eqref{eq:vix_fut_main_term_mixed_kernel}, and, by \eqref{eq:Integral_0_exact}, $\int_{-\infty}^{A_j} g^j_i(y) \phi(y) \dy = \omega^{(i)}_{0,j} - I^j_{i,N}(A_j) - \rho^j_{i,N}(A_j)$. Substituting the latter into \eqref{eq:vix_put_main_term_mixed_kernel} gives $P_P = P^j_{P,N} + \rho^j_N$. The bound on $\rho^j_N$ follows from \eqref{eq:Integral_0_exact}, and the bound tends to zero as $N \to \infty$ by \eqref{eq:hermite_truncation_error}. Finally, \eqref{eq:vix_call_hermite_expan}, \eqref{eq:vix_put_hermite_expan}, and \eqref{eq:vix_fut_hermite} give $C^j_{P,N} - P^j_{P,N} = \sqrt{\lambda_j}\, \ee^{x_{P,j}} (\omega^{(0)}_{0,j} + \omega^{(1)}_{0,j}) - \ee^k = F_P - \ee^k$.
\qed

\subsection{Proof of Lemma~\ref{lem:gaussian_rep_corrections}}
\label{sec:gaussian_rep_corrections}
By \eqref{eq:mean_mixed}--\eqref{eq:variance_mixed}, $\sigma_{P,j} = \alpha_j \varsigma$ and $L - \mu_{P,j} = \frac{1}{2} \alpha_j^2 \tau_0$, where
\eqstar{
\varsigma^2 := \int_0^T \Bigl( \frac{1}{\Delta}\int_T^{T+\Delta} \frac{\xi_0^u}{\VIXsqExp} K_0^u(t) \du \Bigr)^2 \dt,
\qquad
\tau_0 := \int_0^T \frac{1}{\Delta}\int_T^{T+\Delta} \frac{\xi_0^u}{\VIXsqExp} K_0^u(t)^2 \du \dt,
}
do not depend on $j$. Hence, for $x = X_1 + \epsilon = \mu_{P,1} - L + \sigma_{P,1} Z + \epsilon$, the exponent of the second term in \eqref{eq:S1} equals
\eqstar{
\Bigl(\frac{\alpha_1}{\alpha_2} - 1\Bigr) \frac{\alpha_2^2 \tau_0}{2} - \frac{\alpha_2}{\alpha_1} \frac{\alpha_1^2 \tau_0}{2} + \frac{\alpha_2}{\alpha_1}\bigl(\sigma_{P,1} Z + \epsilon\bigr)
= \mu_{P,2} - L + \sigma_{P,2}\Bigl(Z + \frac{\epsilon}{\sigma_{P,1}}\Bigr),
}
so that $\VIXsqExp\, S_1(X_1 + \epsilon) = h(Z + \epsilon/\sigma_{P,1})$. The same computation with \eqref{eq:S2} for $j = 2$, together with the definition of $X_j$, gives, for $j \in \{1,2\}$,
\eqstar{
\VIXsqExp\, S_j(X_j + \epsilon) = h\Bigl(Z + \frac{\epsilon}{\sigma_{P,j}}\Bigr),
\qquad
\lambda_j\, \ee^{X_j + \epsilon}\, \VIXsqExp = \lambda_j\, \ee^{\mu_{P,j} + \sigma_{P,j}(Z + \epsilon/\sigma_{P,j})}.
}
Hence
\eqstar{
\Psi_j(X_j + \epsilon) = G_j\Bigl(Z + \frac{\epsilon}{\sigma_{P,j}}\Bigr),
\qquad
G_j(y) := \lambda_j\, \ee^{\mu_{P,j} + \sigma_{P,j} y}\, \varphi'\bigl(h(y)\bigr).
}
Since $h(y) \ge \lambda_j \ee^{\mu_{P,j} + \sigma_{P,j} y}$, the growth condition on $\varphi'$ gives $|G_j(y)| \le c_\varphi \bigl(\lambda_j \ee^{\mu_{P,j} + \sigma_{P,j} y} + (\lambda_j \ee^{\mu_{P,j} + \sigma_{P,j} y})^{1/2}\bigr) \le C \ee^{C|y|}$. Hence
\eqstar{
\Eb\bigl[\Psi_j(X_j + \epsilon)\bigr] = \int_\Rb G_j(y)\, \phi\Bigl(y - \frac{\epsilon}{\sigma_{P,j}}\Bigr) \dy.
}
By \eqref{eq:probabilist's_hermite_polynomial}, $\partial_\epsilon^n \phi(y - \epsilon/\sigma) = \sigma^{-n} \He_n(y - \epsilon/\sigma)\, \phi(y - \epsilon/\sigma)$, and for $|\epsilon| \le 1$ this is bounded by $C_n\, \ee^{-y^2/8}$, uniformly in $\epsilon$. Differentiation under the integral sign is thus justified, and taking $\epsilon = 0$ yields \eqref{eq:gaussian_rep_P}. Finally, $\lambda_j \ee^{\mu_{P,j} + \sigma_{P,j} z} = \pi_j(z) h(z)$, $\He_0 = 1$, $\He_1(z) = z$, and $\He_2(z) = z^2 - 1$, which gives \eqref{eq:gaussian_rep_corrections}.
\qed

\subsection{Proof of Theorem~\ref{thm:vix_iv_mixed_kernel_rigorous}}
\label{sec:vix_iv_mixed_kernel_rigorous_proof}
Throughout, $C$ denotes a positive constant depending only on the quantities listed in the statement, whose value may change from line to line.

\textit{Step 1 (zeroth-order terms).}
Let $\Theta := [-c, c]^2 \times [c^{-1}, c]^2$. For $\vartheta = (m_1, m_2, s_1, s_2) \in \Theta$, set $h_\vartheta(z) := \lambda_1 \ee^{m_1 + s_1 z} + \lambda_2 \ee^{m_2 + s_2 z}$,
\eqstar{
F_0(\vartheta) := \Eb\Bigl[ \sqrt{h_\vartheta(Z)} \Bigr],
\qquad
C_0(\vartheta, k) := \Eb\Bigl[ \bigl(\sqrt{h_\vartheta(Z)} - \ee^k\bigr)^+ \Bigr].
}
By Assumption~\ref{assump:no_degenerate}, $\vartheta_\Delta := (\mu_{P,1}, \mu_{P,2}, \sigma_{P,1}, \sigma_{P,2}) \in \Theta$ for all $\Delta \le \Delta_0$, and $F_{P,0} = F_0(\vartheta_\Delta)$, $C_{P,0}(k) = C_0(\vartheta_\Delta, k)$. Hence $\Sigma_0$ depends on $\Delta$ only through $\vartheta_\Delta$, which ranges in a fixed compact set. Since $\lambda_1 \ee^{m_1 + s_1 z} \le h_\vartheta(z) \le \ee^{c (1+ |z|)}$, we get
\eqlnostar{eq:F0_bounds_mixed}{
f_{\min} := \sqrt{\lambda_1}\, \ee^{-\frac{c}{2}} \le \sqrt{\lambda_1}\, \ee^{\frac{m_1}{2} + \frac{s_1^2}{8}} \le F_0(\vartheta) \le \ee^{\frac{c}{2}}\, \Eb\bigl[\ee^{\frac{c}{2}|Z|}\bigr] =: f_{\max}.
}
Since $h_\vartheta$ is a continuous increasing bijection from $\Rb$ onto $(0,\infty)$, the random variable $X_\vartheta := \sqrt{h_\vartheta(Z)}$ satisfies $\mathbb{P}(X_\vartheta > \ee^k) > 0$ and $\mathbb{P}(X_\vartheta < \ee^k) > 0$. Hence $C_0 - (F_0 - \ee^k) = \Eb[(\ee^k - X_\vartheta)^+] > 0$, $C_0 > 0$, and $F_0 - C_0 = \Eb[X_\vartheta \wedge \ee^k] > 0$, that is,
\eqstar{
\bigl(F_0(\vartheta) - \ee^k\bigr)^+ < C_0(\vartheta, k) < F_0(\vartheta).
}
For fixed $(x,k)$, the map $\sigma \mapsto C^\BS(x,k,\sigma)$ is a strictly increasing bijection from $(0,\infty)$ onto $((\ee^x - \ee^k)^+, \ee^x)$. Thus $\Sigma_0(k)$ in \eqref{eq:sigma0_mixed} is well defined. Moreover, by the implicit function theorem ($\Vega > 0$), the inverse map $I(x,k,p)$, defined by $C^\BS(x,k,I(x,k,p)) = p$, is continuous on the open set $\{(x,k,p) : (\ee^x - \ee^k)^+ < p < \ee^x\}$. By dominated convergence, with dominating function $\ee^{\frac{c}{2}(1+|z|)} + \sup_{k\in\Kc} \ee^k$, the maps $\vartheta \mapsto F_0(\vartheta)$ and $(\vartheta, k) \mapsto C_0(\vartheta, k)$ are continuous on $\Theta \times \Kc$. Hence $(\vartheta, k) \mapsto I\bigl(\ln F_0(\vartheta), k, C_0(\vartheta, k)\bigr)$ is continuous and positive on the compact set $\Theta \times \Kc$. We let $s_{\min} > 0$ and $s_{\max} < \infty$ denote its minimum and maximum, which proves the first part of (i).

As in the proof of Theorem~\ref{thm:vix_iv_single_kernel}, we introduce the compact sets
\eqstar{
\Xc := [\ln f_{\min} - 1, \ln f_{\max} + 1],
\qquad
J := \Bigl[\frac{s_{\min}}{2}, 2 s_{\max}\Bigr],
\qquad
\Dc := \Xc \times \Kc \times J,
}
we let $M$ bound the first- and second-order partial derivatives of $C^\BS$ with respect to $(x,\sigma)$ on $\Dc$, and we set $v_0 := \min_\Dc \Vega > 0$. Note that $(\ln F_{P,0}, k, \Sigma_0(k)) \in \Dc$ for every $k \in \Kc$ and $\Delta \le \Delta_0$.

\textit{Step 2 (first-order terms).} By Lemma~\ref{lem:order_of_coefficients} applied to each kernel $K_j$, $|\gamma_{i,j}| \le C_\gamma \Delta^{2\varepsilon}$. Since $\sigma_{P,j} \ge c^{-1}$, \eqref{eq:weights_q_mixed} gives $|q_j(z)| \le C \Delta^{2\varepsilon}(1 + z^2)$. Since $\pi_1, \pi_2 \ge 0$ and $\pi_1 + \pi_2 = 1$, it follows that $|q(z)| \le \max_j |q_j(z)| \le C \Delta^{2\varepsilon} (1 + z^2)$. Combined with $\sqrt{h(z)} \le \ee^{\frac{c}{2}(1 + |z|)}$, this yields
\eqlnostar{eq:first_order_bounds_mixed}{
|C_{P,1}(k)| + |F_{P,1}| \le \Eb\bigl[ |q(Z)| \sqrt{h(Z)} \bigr] \le C \Delta^{2\varepsilon}.
}
Since $\Vega(\ln F_{P,0}, k, \Sigma_0(k)) \ge v_0$, we get $|\Sigma_1(k)| \le (|C_{P,1}(k)| + |F_{P,1}|)/v_0 \le C\Delta^{2\varepsilon}$, which completes the proof of (i). The second expression in \eqref{eq:sigma1_mixed} follows from \eqref{eq:zeroth_first_order_mixed} and \eqref{eq:vega}. The heuristic behind the form of $\Sigma_1$ is explained in Remark~\ref{re:choice_of_sigma1}.

\textit{Step 3 (the futures).}
By \eqref{eq:weak_approx_gaussian_rep}, \eqref{eq:remainder_assumption_mixed}, and \eqref{eq:first_order_bounds_mixed}, $|F - F_{P,0}| \le |F_{P,1}| + |\Ec_F| \le C \Delta^{2\varepsilon}$. We choose $\Delta_1 \le \Delta_0$ such that, for all $\Delta \le \Delta_1$, $|F - F_{P,0}| \le f_{\min}/2$ and $|\Sigma_1(k)| \le s_{\min}/4$. Then $u := (F - F_{P,0})/F_{P,0}$ satisfies $|u| \le \frac{1}{2}$, so that $F > 0$ and $\ell := \ln F - \ln F_{P,0} = \ln(1+u)$ satisfies $|\ell| \le 2|u| \le C \Delta^{2\varepsilon}$. Reducing $\Delta_1$ if necessary, $|\ell| \le 1$, hence $\ln F \in \Xc$. Moreover, using $|\ln(1+u) - u| \le u^2$ for $|u| \le \frac{1}{2}$,
\eqlnostar{eq:log_futures_mixed}{
\bigl| F_{P,0}\, \ell - (F - F_{P,0}) \bigr| \le F_{P,0}\, u^2 \le \frac{(F - F_{P,0})^2}{f_{\min}} \le C \Delta^{4\varepsilon}.
}
\textit{Step 4 (BS price of the candidate volatility).}
Set $\widehat{\Sigma}(k) := \Sigma_0(k) + \Sigma_1(k)$, which lies in $J' := [\frac{3}{4} s_{\min}, s_{\max} + \frac{1}{4} s_{\min}] \subset J$. The segment joining $(\ln F_{P,0}, \Sigma_0)$ to $(\ln F, \widehat{\Sigma})$ lies in $\Xc \times J$. A first-order Taylor expansion of $(x,\sigma) \mapsto C^\BS(x,k,\sigma)$ around $(\ln F_{P,0}, \Sigma_0)$, together with $\partial_x C^\BS(x,k,\sigma) = \ee^x \Phi(d_1(x,k,\sigma))$, gives
\eqstar{
C^\BS\bigl(\ln F, k, \widehat{\Sigma}\bigr)
&= C_{P,0} + F_{P,0}\, \ell\, \Phi(d_0) + \Sigma_1\, \Vega\bigl(\ln F_{P,0}, k, \Sigma_0\bigr) + R_2,\\
|R_2| & \le \frac{M}{2} \bigl(|\ell| + |\Sigma_1|\bigr)^2 \le C \Delta^{4\varepsilon}.
}
Since $F - F_{P,0} = F_{P,1} + \Ec_F$, \eqref{eq:log_futures_mixed} and the definition \eqref{eq:sigma1_mixed} give
\eqstar{
F_{P,0}\, \ell = F_{P,1} + \Ec_F + r_F, \quad |r_F| \le C \Delta^{4\varepsilon},
\qquad
\Sigma_1\, \Vega(\ln F_{P,0}, k, \Sigma_0) = C_{P,1} - \Phi(d_0)\, F_{P,1}.
} The terms in $F_{P,1}$ cancel, and we obtain, using $C_{P,0} + C_{P,1} = C - \Ec_C$,
\eqlnostar{eq:candidate_price_mixed}{
C^\BS\bigl(\ln F, k, \widehat{\Sigma}\bigr) = C - \Ec_C + \Phi(d_0)\, \Ec_F + R_4,
\qquad
|R_4| \le C \Delta^{4\varepsilon}.
}
In particular, by \eqref{eq:remainder_assumption_mixed}, $|C^\BS(\ln F, k, \widehat{\Sigma}) - C| \le C \Delta^{3\varepsilon}$.

\textit{Step 5 (inversion).}
Let $r := s_{\min}/4$, so that $[\sigma_0 - r, \sigma_0 + r] \subset J$ for every $\sigma_0 \in J'$. The argument leading to \eqref{eq:coercivity_bs} uses only $\Vega \ge v_0$ on $\Dc$ and the monotonicity of $\sigma \mapsto C^\BS(x,k,\sigma)$. It therefore applies verbatim with the present sets $\Xc$, $J$, $J'$, $\Dc$. Applying it with $x = \ln F \in \Xc$, $\sigma_0 = \widehat{\Sigma}(k) \in J'$, and $\sigma = \sbs(k,T)$, for which $C^\BS(\ln F, k, \sbs(k,T)) = C$, we get
\eqstar{
v_0 \min\bigl\{ |\sbs(k,T) - \widehat{\Sigma}(k)|,\, r \bigr\} \le \bigl| C - C^\BS(\ln F, k, \widehat{\Sigma}) \bigr| \le C \Delta^{3\varepsilon}.
}
Reducing $\Delta_1$ so that $C \Delta_1^{3\varepsilon} < v_0 r$, the minimum is attained by the first term, which proves \eqref{eq:vix_iv_error_bound_mixed_kernel}.

\textit{Step 6 (structure of the error).}
By Step 5, $\sbs(k,T) \in [\widehat{\Sigma} - r, \widehat{\Sigma} + r] \subset J$. By the mean value theorem, there exists $\bar\sigma$ between $\widehat{\Sigma}$ and $\sbs(k,T)$ such that
\eqstar{
C - C^\BS\bigl(\ln F, k, \widehat{\Sigma}\bigr) = \Vega(\ln F, k, \bar\sigma) \bigl(\sbs(k,T) - \widehat{\Sigma}\bigr).
}
Write $V_0 := \Vega(\ln F_{P,0}, k, \Sigma_0)$, $\bar V := \Vega(\ln F, k, \bar\sigma)$, and $\delta_C := C - C^\BS(\ln F, k, \widehat{\Sigma}) = \Ec_C - \Phi(d_0)\Ec_F - R_4$, by \eqref{eq:candidate_price_mixed}. Both $(\ln F, k, \bar\sigma)$ and $(\ln F_{P,0}, k, \Sigma_0)$ lie in $\Dc$, so both Vegas are bounded below by $v_0$, and
\eqstar{
\bigl| \bar V - V_0 \bigr| \le M \bigl( |\ell| + |\Sigma_1| + |\sbs(k,T) - \widehat{\Sigma}| \bigr) \le C \Delta^{2\varepsilon}.
}
Hence
\eqstar{
\sbs(k,T) - \widehat{\Sigma}
= \frac{\delta_C}{V_0} + \delta_C\, \frac{V_0 - \bar V}{V_0\, \bar V}
= \frac{\Ec_C - \Phi(d_0) \Ec_F}{V_0} + R(k),
\qquad
R(k) := - \frac{R_4}{V_0} + \delta_C\, \frac{V_0 - \bar V}{V_0\, \bar V}.
}
Since $|\delta_C| \le C\Delta^{3\varepsilon}$ by Step~4 and $V_0, \bar V \ge v_0$, we get $|R(k)| \le C\Delta^{4\varepsilon} + C\Delta^{5\varepsilon} \le C\Delta^{4\varepsilon}$, which proves \eqref{eq:vix_iv_error_structure_mixed_kernel}.
\qed

\section{Additional numerical results}
\label{sec:additional_numerical_results}

In this section, we first present the numerical results of the $N$th-order Hermite expansion for the VIX call price in Proposition~\ref{prop:vix_opt_price_hermite_expan}, and then report the accuracy of the implied volatility expansion given in Remark~\ref{re:extended_vix_iv_mixed_kernel}. 
For both numerical tests, we consider the parameter scenarios as given in Tables \ref{tab:parameters_order_mb} and \ref{tab:parameters_order_mrb} under both the mixed standard and mixed rough Bergomi models, with maturities of 1, 3, and 6 months. Moreover, we assume the initial forward variance is flat, that is, $\xi_0 = \xi_0^u = 0.24^2$, for all $u \in [T, T+\Delta]$, and $\kappa = 1$ (resp. $H=0.1$) in the mixed standard Bergomi (resp. rough Bergomi) model.

\subsection{Option price}
\label{sec:numerical_tests_vix_opt_mixed_kernel}

The reference VIX call prices, with their $95\%$ confidence intervals in the mixed rough Bergomi model, and the strikes are those of Section~\ref{sec:numerical_tests_mixed_kernel}. While computing the $N$th-order Hermite expansion of the VIX call price in \eqref{eq:vix_call_hermite_expan} (referred to as \textit{Hermite}), we factor out the component $j$ with the lower proxy volatility, and use the orders $N$ of Tables~\ref{tab:parameters_order_mb} and \ref{tab:parameters_order_mrb}. Moreover, we compare it with the weak option price approximation in \eqref{eq:vix_opt_price_approx_mixed_kernel} (referred to as \textit{Approx}), computed as in Section~\ref{sec:rigorous_mixed_numerics}. The results are presented in Figures~\ref{fig:mixed_case_bergomi_call_1}--\ref{fig:mixed_case_rbergomi_call_2}. The reported results show that the $N$th-order Hermite expansion for the VIX call price is accurate with respect to the reference price, with relative errors below $1.4\%$ in all cases, and differs from \textit{Approx} by less than $0.1\%$.\par

\begin{figure}[H]
    \centering
    \hspace{-0.5cm}
    \includegraphics[width=0.5\linewidth]{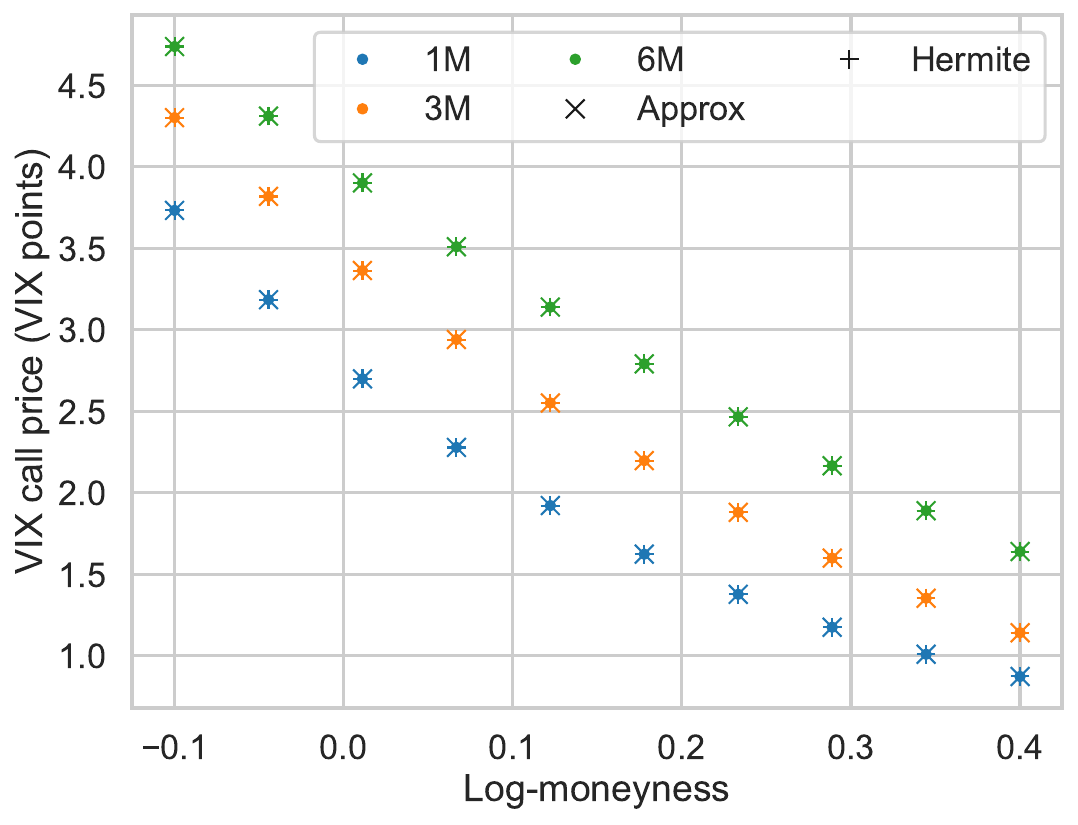}
    \includegraphics[width=0.5\linewidth]{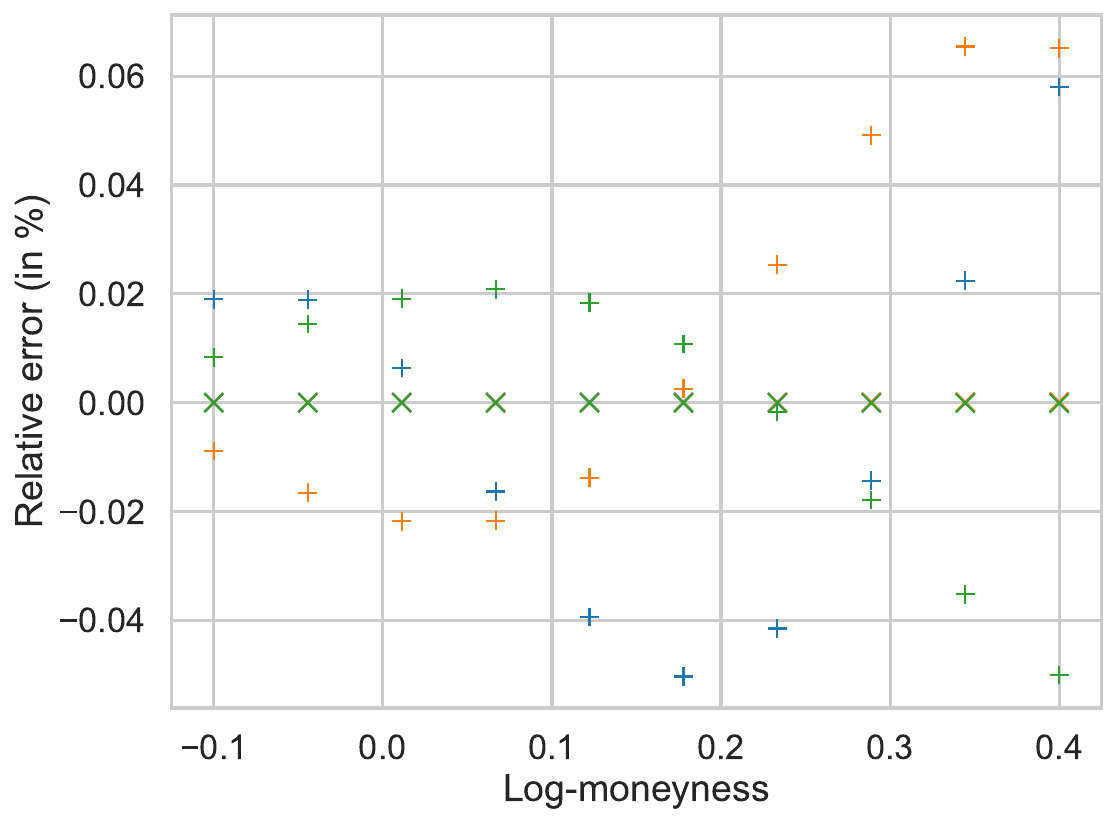}
    \caption{VIX call prices of Scenario 1 in the mixed standard Bergomi model. Reference VIX futures $F$ for the 1-, 3-, and 6-month maturities: 0.2178, 0.1986, 0.1862.}
    \label{fig:mixed_case_bergomi_call_1}
\end{figure}

\begin{figure}[H]
    \centering
    \hspace{-0.5cm}
    \includegraphics[width=0.5\linewidth]{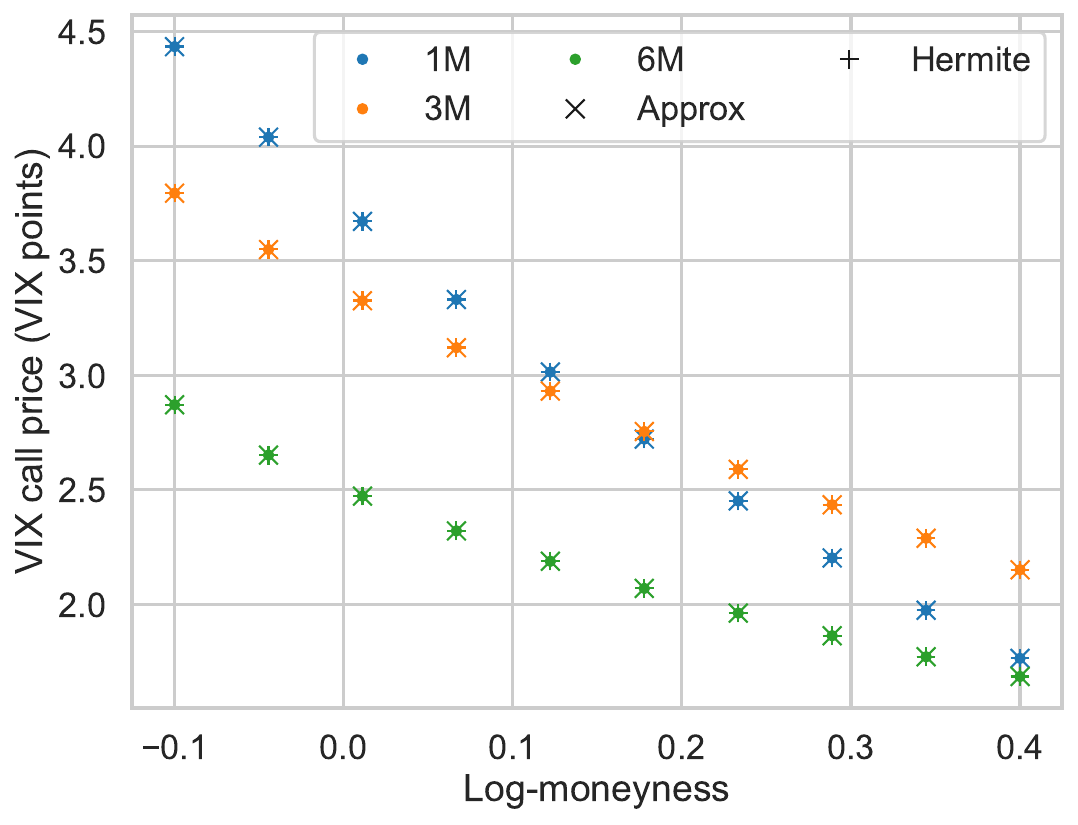}
    \includegraphics[width=0.5\linewidth]{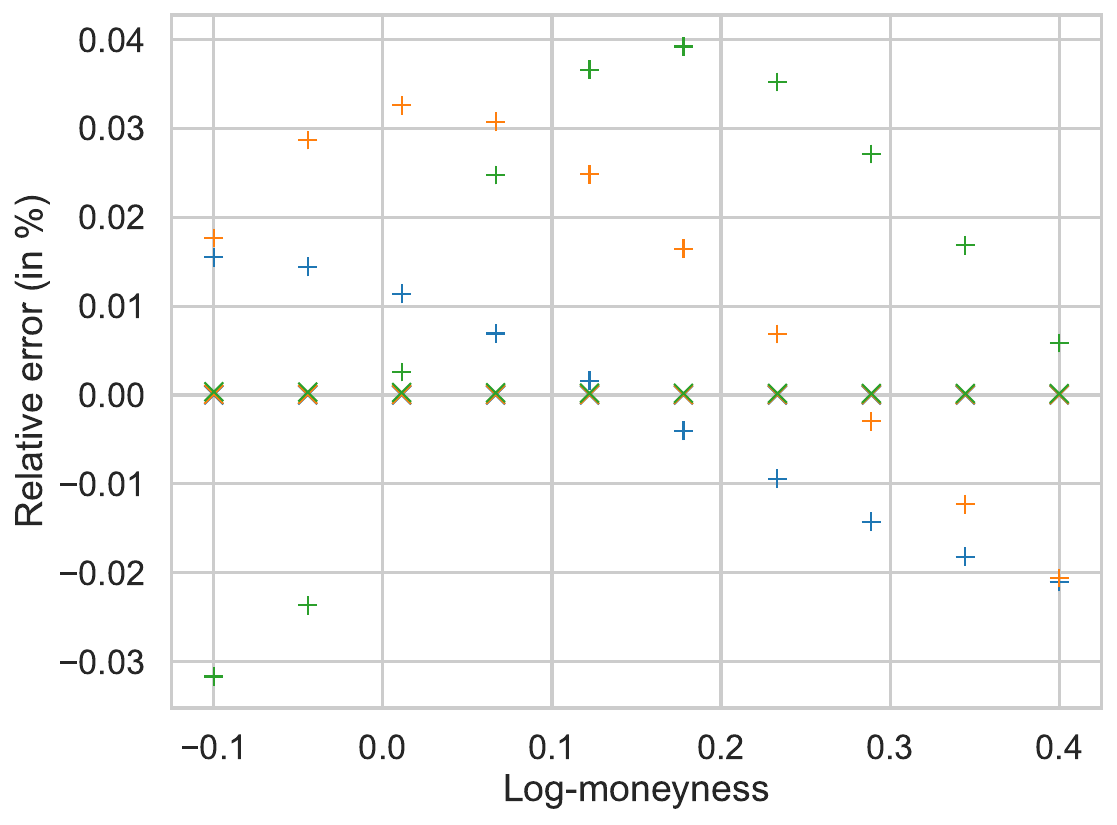}\\
    \caption{VIX call prices of Scenario 2 in the mixed standard Bergomi model. Reference VIX futures $F$ for the 1-, 3-, and 6-month maturities: 0.2072, 0.1752, 0.1566.}
    \label{fig:mixed_case_bergomi_call_2}
\end{figure}

\begin{figure}[H]
    \centering
    \hspace{-0.5cm}
    \includegraphics[width=0.5\linewidth]{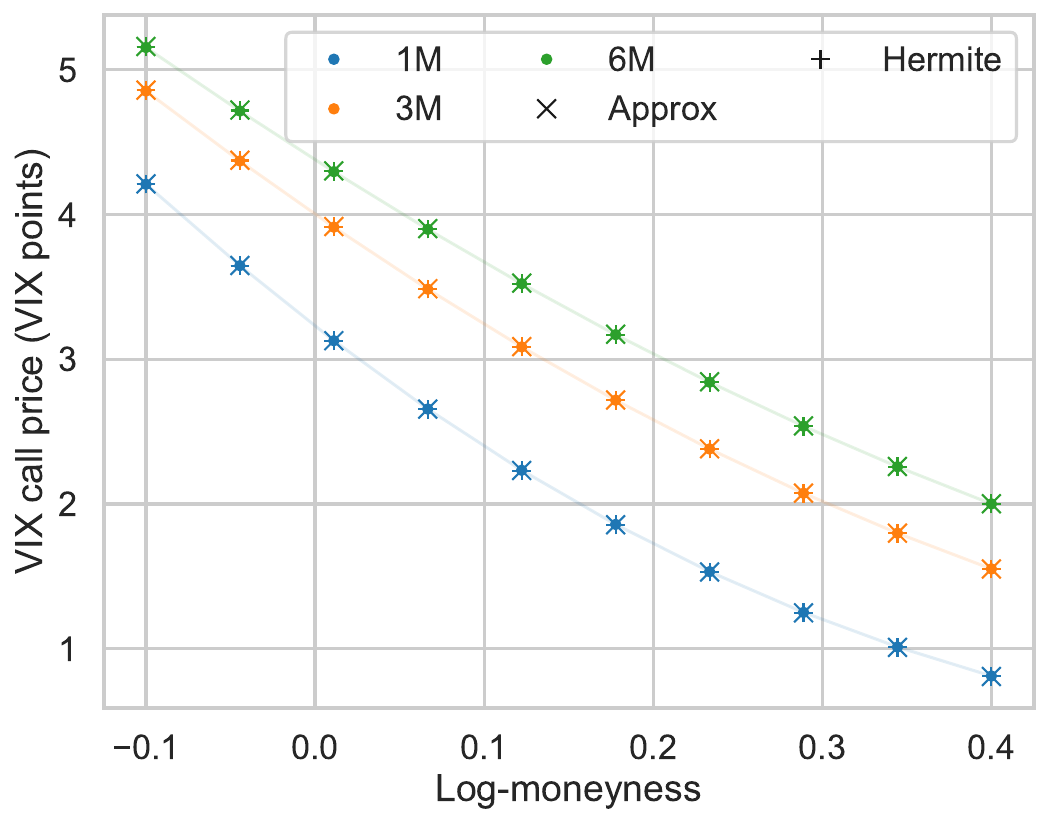}
    \includegraphics[width=0.5\linewidth]{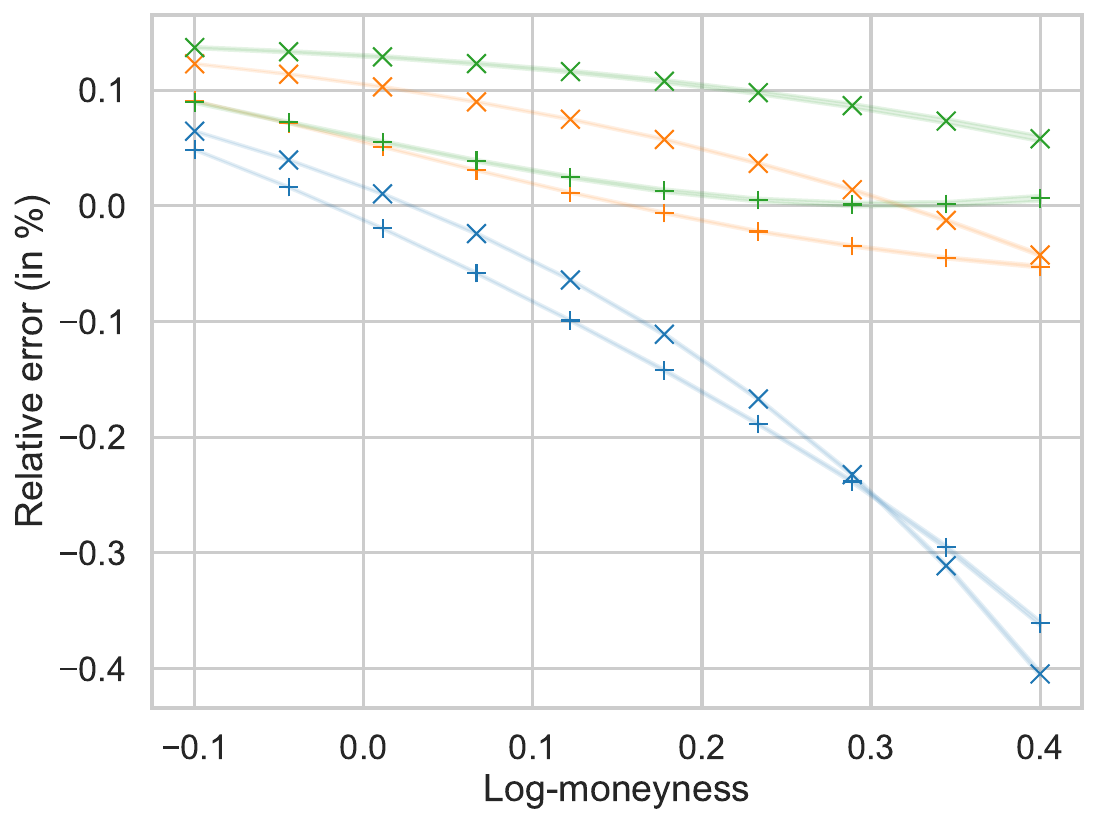}\\
    \caption{VIX call prices of Scenario 3 in the mixed rough Bergomi model. Reference VIX futures $F$ for the 1-, 3-, and 6-month maturities: 0.2233, 0.2106, 0.2010.}
    \label{fig:mixed_case_rbergomi_call_1}
\end{figure}

\begin{figure}[H]
    \centering
    \hspace{-0.5cm}
    \includegraphics[width=0.5\linewidth]{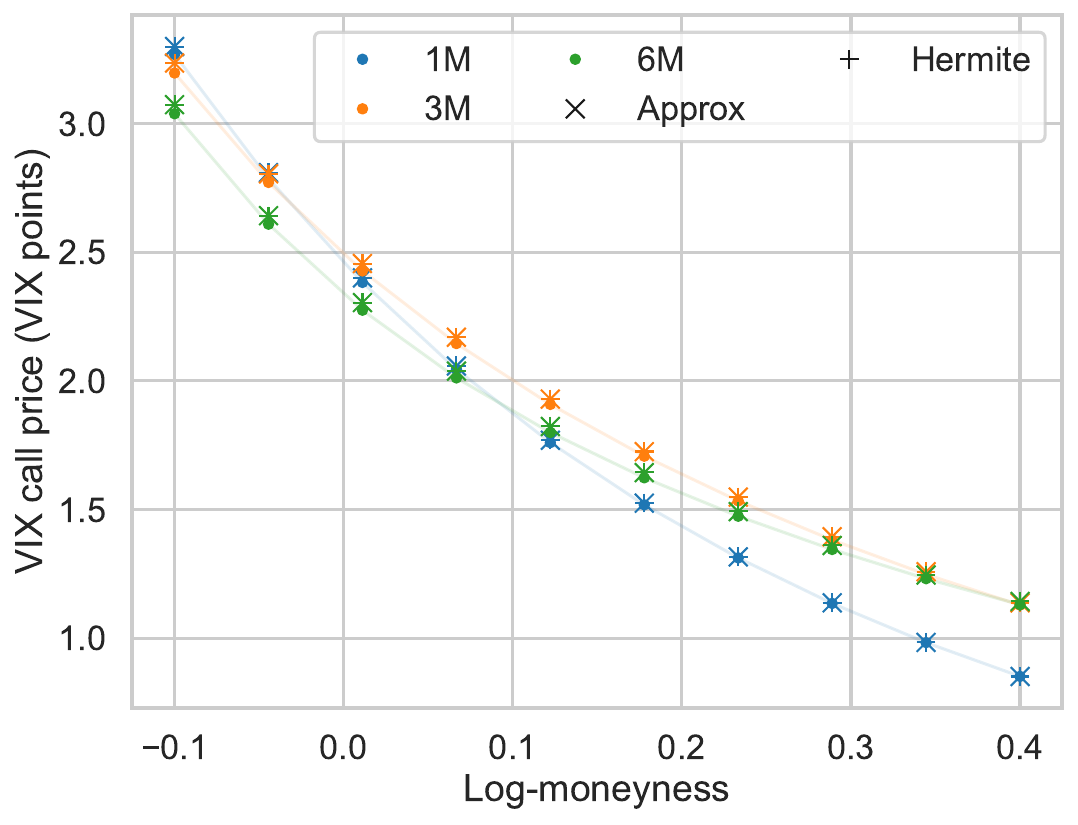}
    \includegraphics[width=0.5\linewidth]{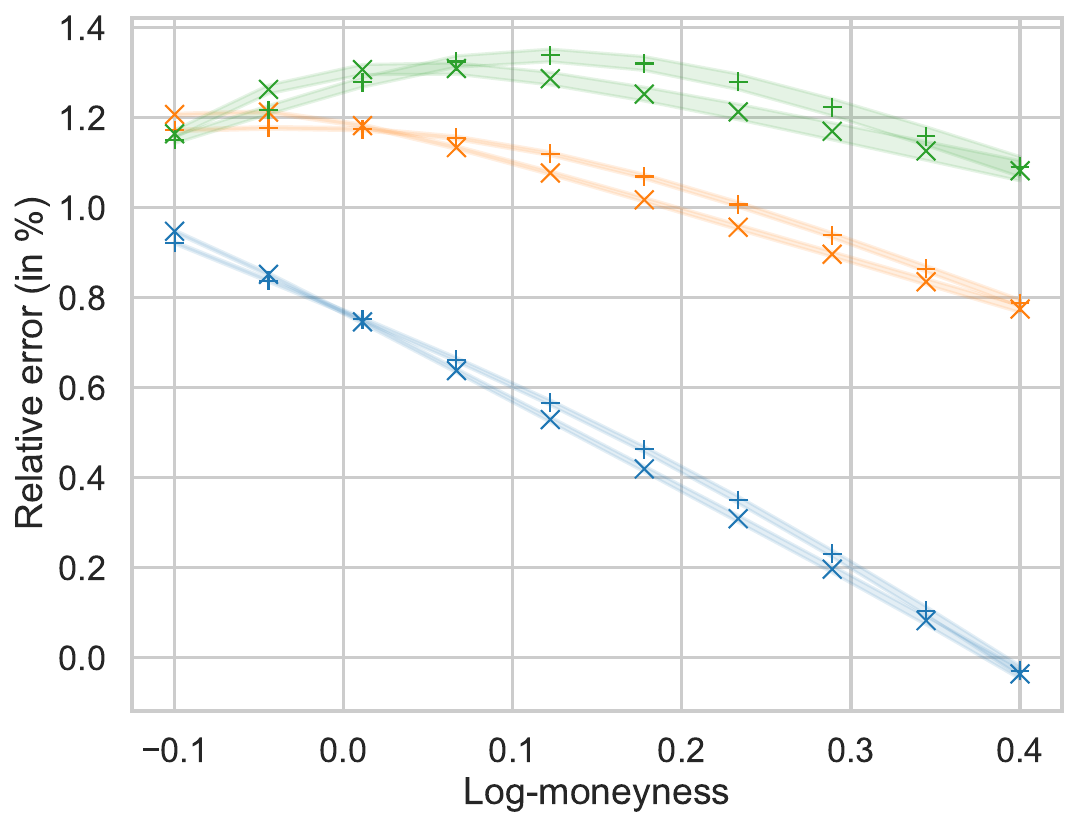}\\
    \caption{VIX call prices of Scenario 4 in the mixed rough Bergomi model. Reference VIX futures $F$ for the 1-, 3-, and 6-month maturities: 0.2236, 0.2130, 0.2061.}
    \label{fig:mixed_case_rbergomi_call_2}
\end{figure}

\subsection{Implied volatility}
\label{sec:numerical_tests_vix_iv_extended_mixed_kernel}

In this section, we test the accuracy of the VIX implied volatility expansion \eqref{eq:vix_iv_proxy_mixed_kernel_extended} in Remark~\ref{re:extended_vix_iv_mixed_kernel}, which is labeled as \textit{Expansion (extended)} in the plots. This expansion is derived by absorbing the strike shift in the $N$th-order Hermite expansion of the VIX call price \eqref{eq:vix_call_hermite_expan} through a Taylor expansion in the log-strike, rather than by discarding it.
The reference implied volatility values are computed using the same methods described in Section \ref{sec:numerical_tests_mixed_kernel}. Moreover, we also report the values obtained from the implied volatility expansion in \eqref{eq:vix_iv_mixed_kernel_approx}, whose derivation discards the shift residual. 
The numerical results under the mixed standard Bergomi model are presented in Figures~\ref{fig:mixed_case_bergomi_1_re} and \ref{fig:mixed_case_bergomi_2_re}. Both expansions use the component $j$ with the lower proxy volatility, and values outside the plotted range are not shown. We observe that the performance of \eqref{eq:vix_iv_proxy_mixed_kernel_extended} is worse than that of \eqref{eq:vix_iv_mixed_kernel_approx} in Scenario 1 (maximum relative errors of $16.8\%$ against $6.2\%$), where the ratio $\sigma_{P,1}/\sigma_{P,2} = \omega_1/\omega_2 = 5$ is moderate compared with Scenario 2, in which the larger proxy volatility is $12$ times the smaller one. In Scenario 2, although the implied volatility expansion \eqref{eq:vix_iv_proxy_mixed_kernel_extended} gives a closer approximation for some values of log-moneyness, its overall performance is worse, with relative errors above $100\%$ at $16$ of the $30$ points.
The same pattern can be observed for the mixed rough Bergomi model in Figures~\ref{fig:mixed_case_rbergomi_1_re} and \ref{fig:mixed_case_rbergomi_2_re}. Therefore, we can conclude that the implied volatility expansion \eqref{eq:vix_iv_proxy_mixed_kernel_extended} generally performs worse than the expansion \eqref{eq:vix_iv_mixed_kernel_approx}.

\begin{figure}[H]
    \centering
    \hspace{-0.5cm}
    \includegraphics[width=0.5\linewidth]{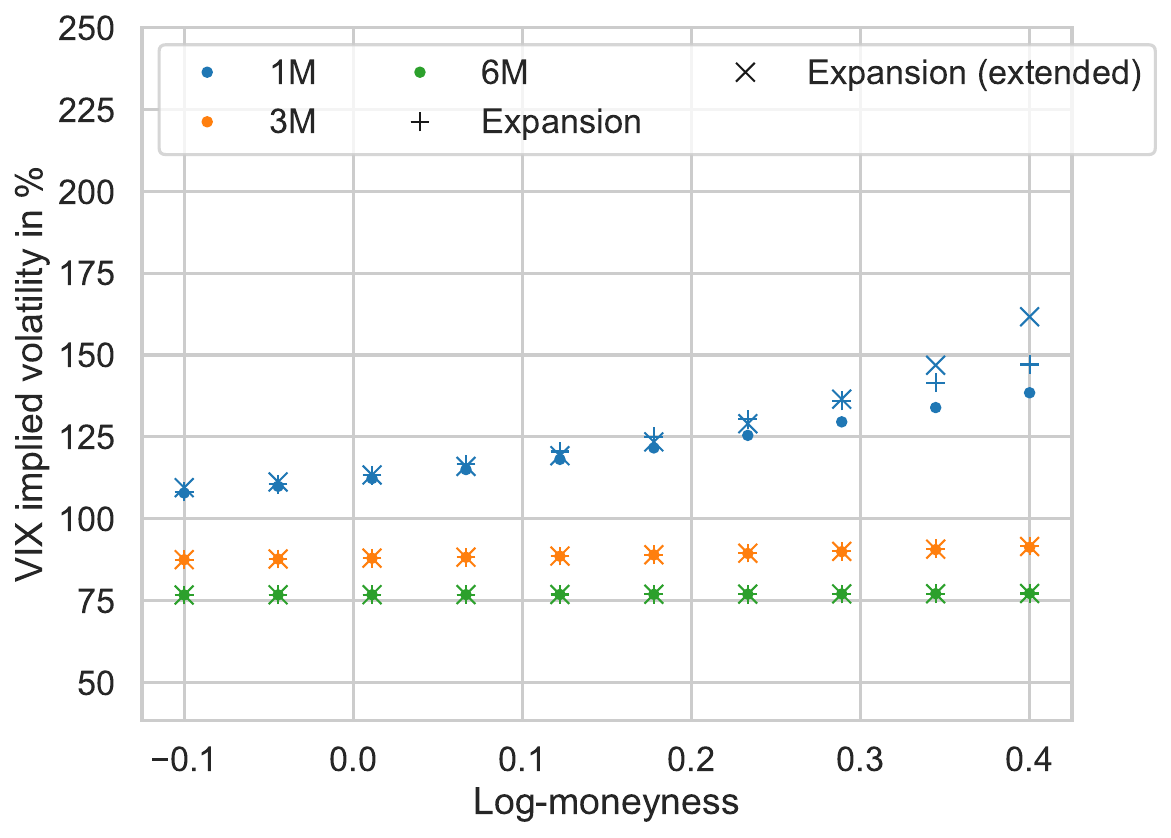}
    \includegraphics[width=0.5\linewidth]{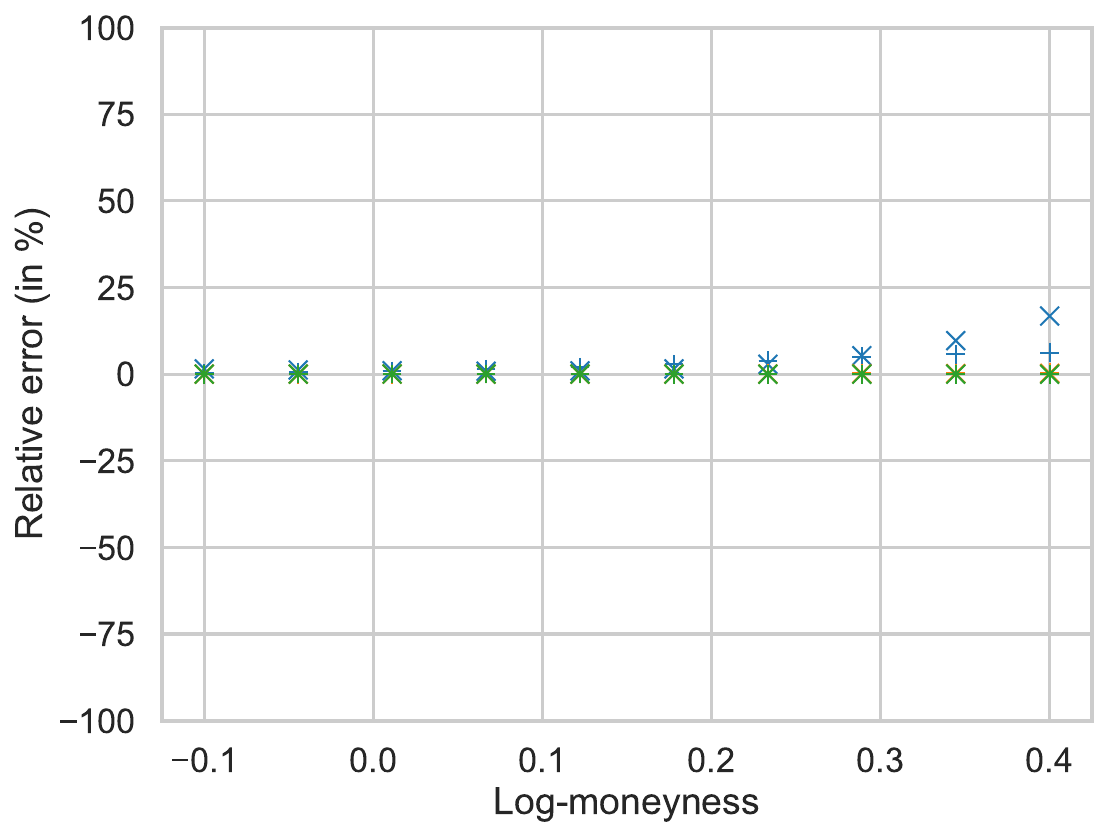}
    \caption{VIX smiles of Scenario 1 in the mixed standard Bergomi model. Expansion \eqref{eq:vix_iv_mixed_kernel_approx} and extended expansion \eqref{eq:vix_iv_proxy_mixed_kernel_extended}, both around $\sigma_{P,2}$ (lower). Reference VIX futures $F$ for the 1-, 3-, and 6-month maturities: 0.2178, 0.1986, 0.1862.}
    \label{fig:mixed_case_bergomi_1_re}
\end{figure}

\begin{figure}[H]
    \centering
    \hspace{-0.5cm}
    \includegraphics[width=0.5\linewidth]{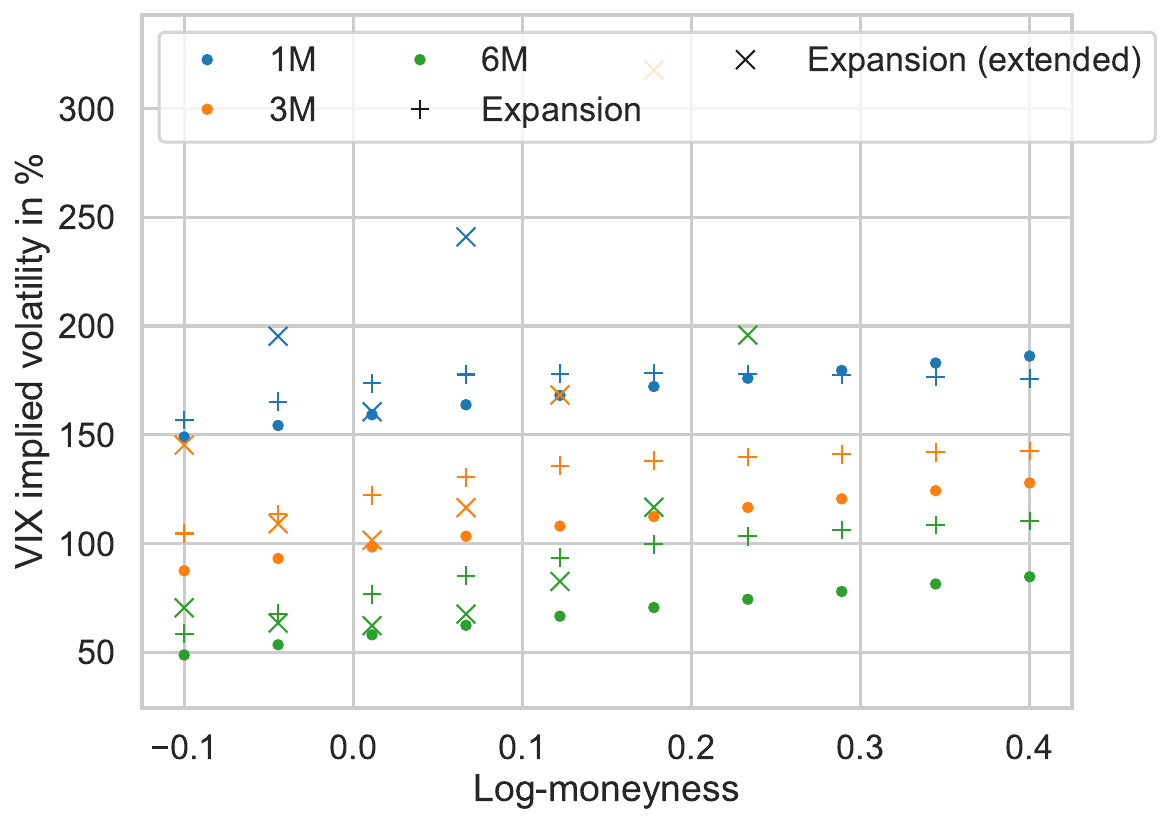}
    \includegraphics[width=0.5\linewidth]{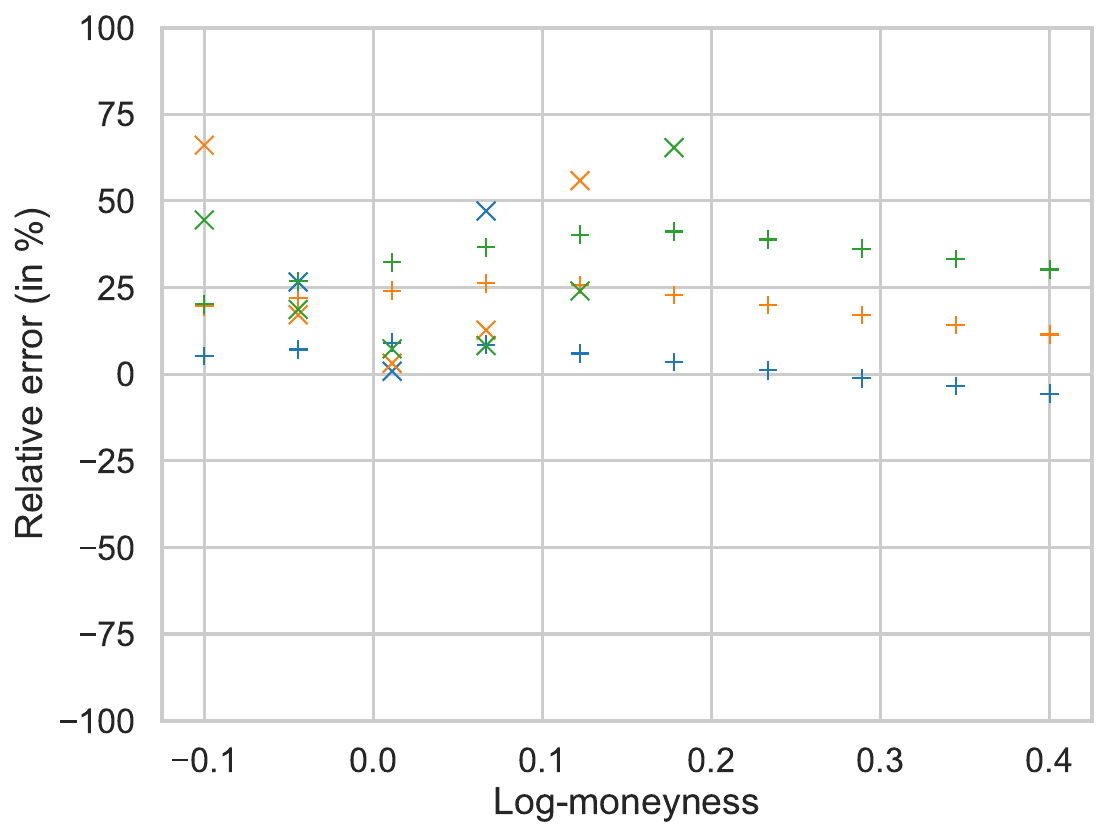}
    \caption{VIX smiles of Scenario 2 in the mixed standard Bergomi model. Expansion \eqref{eq:vix_iv_mixed_kernel_approx} and extended expansion \eqref{eq:vix_iv_proxy_mixed_kernel_extended}, both around $\sigma_{P,1}$ (lower). Reference VIX futures $F$ for the 1-, 3-, and 6-month maturities: 0.2072, 0.1752, 0.1566.}
    \label{fig:mixed_case_bergomi_2_re}
\end{figure}

\begin{figure}[H]
    \centering
    \hspace{-0.5cm}
    \includegraphics[width=0.5\linewidth]{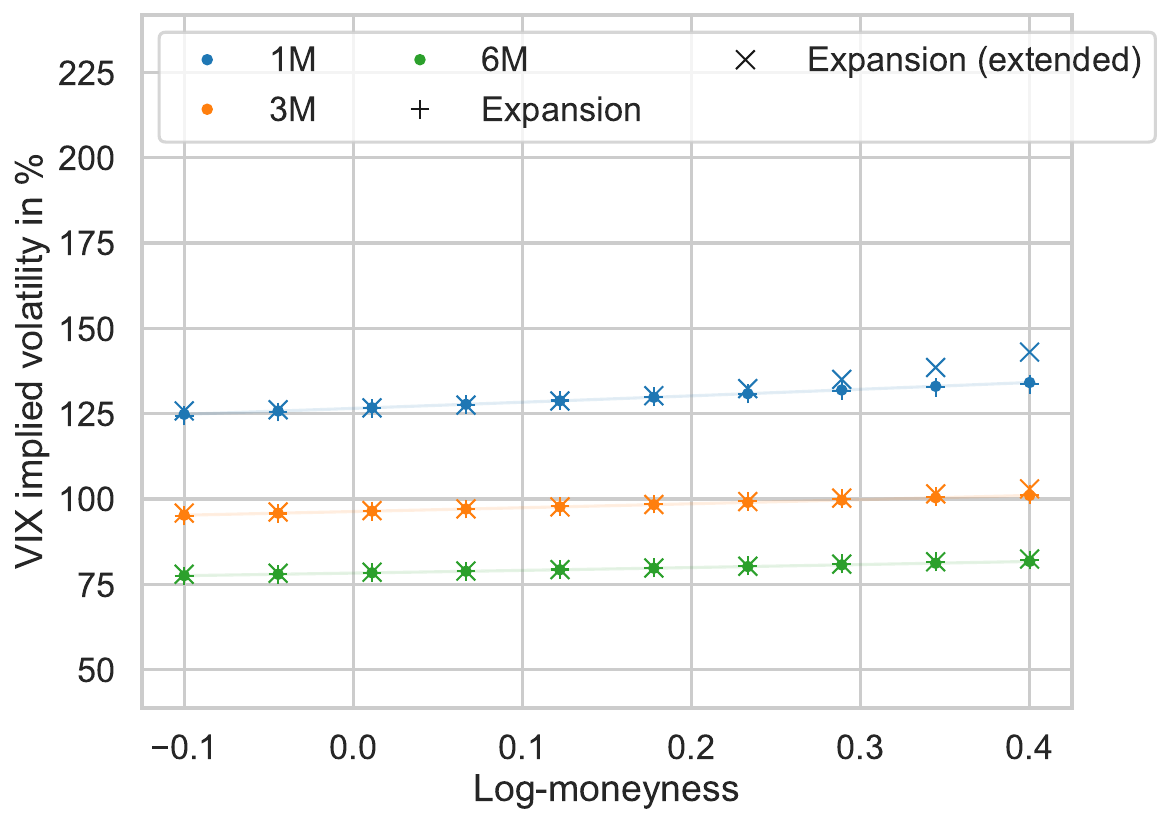}
    \includegraphics[width=0.5\linewidth]{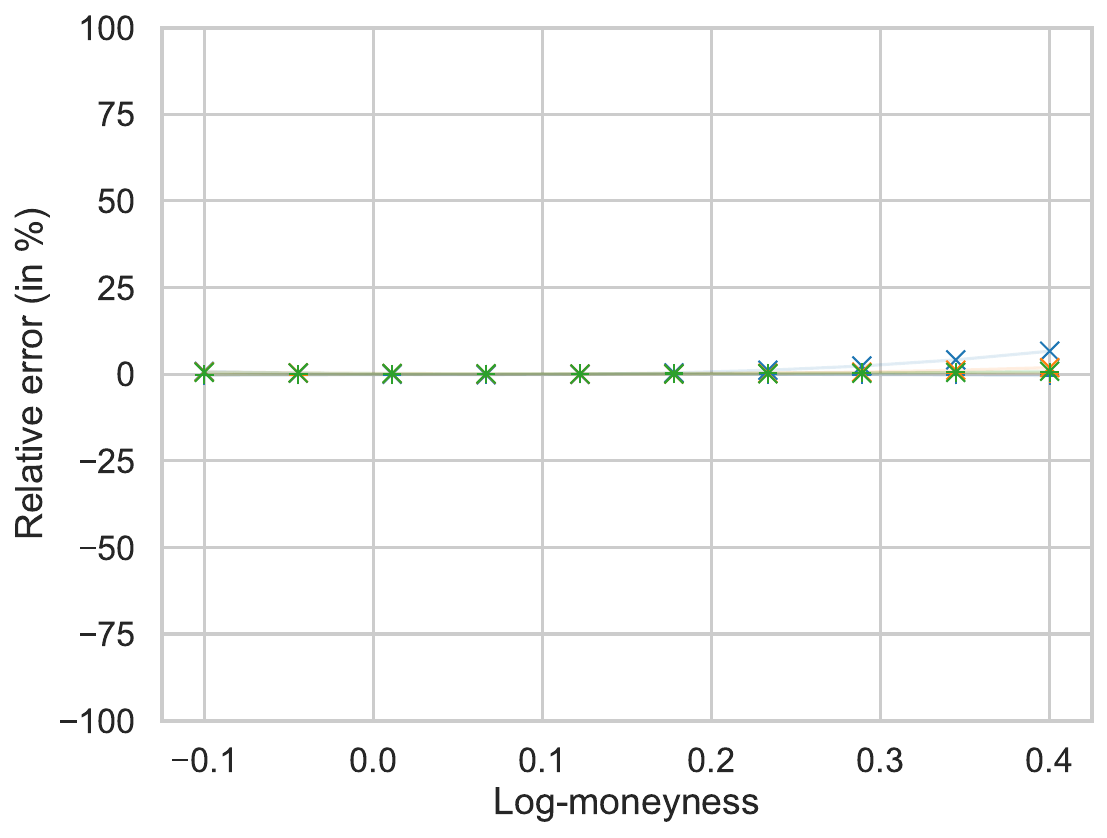}
    \caption{VIX smiles of Scenario 3 in the mixed rough Bergomi model. Expansion \eqref{eq:vix_iv_mixed_kernel_approx} and extended expansion \eqref{eq:vix_iv_proxy_mixed_kernel_extended}, both around $\sigma_{P,2}$ (lower). Reference VIX futures $F$ for the 1-, 3-, and 6-month maturities: 0.2233, 0.2106, 0.2010.}
    \label{fig:mixed_case_rbergomi_1_re}
\end{figure}

\begin{figure}[H]
    \centering
    \hspace{-0.5cm}
    \includegraphics[width=0.5\linewidth]{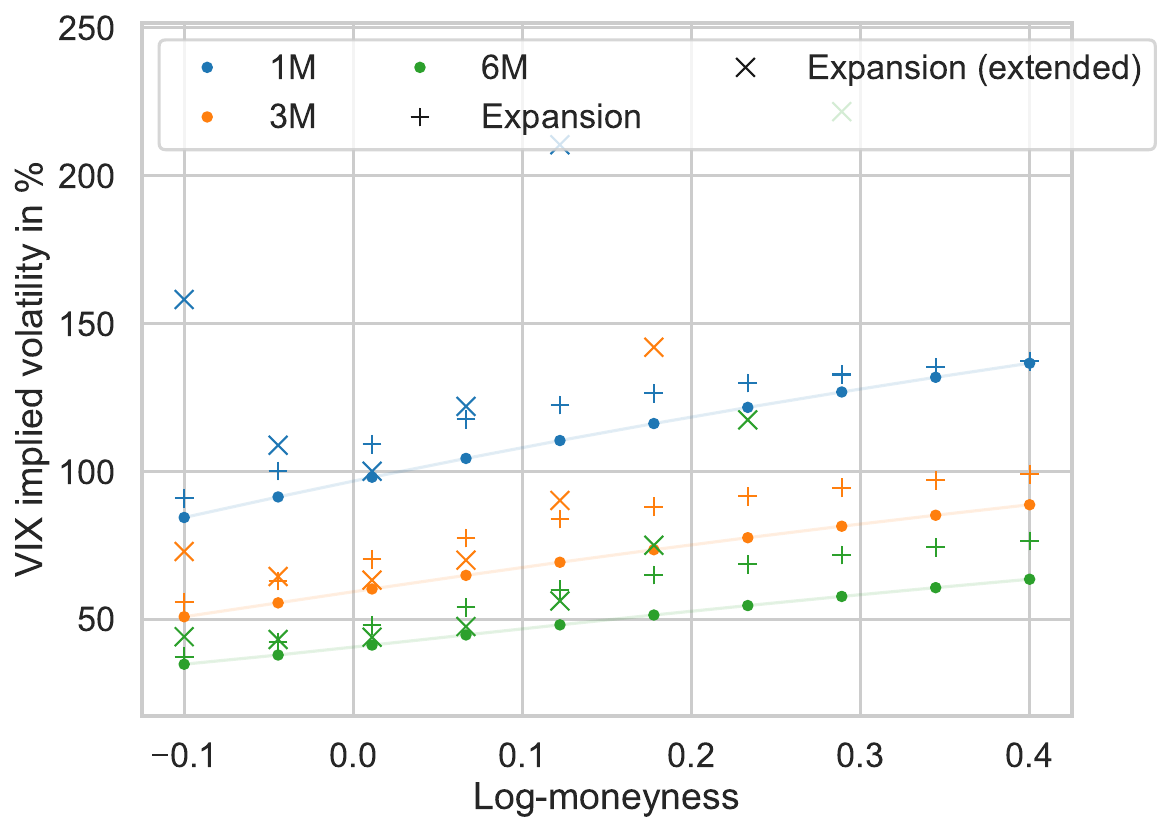}
    \includegraphics[width=0.5\linewidth]{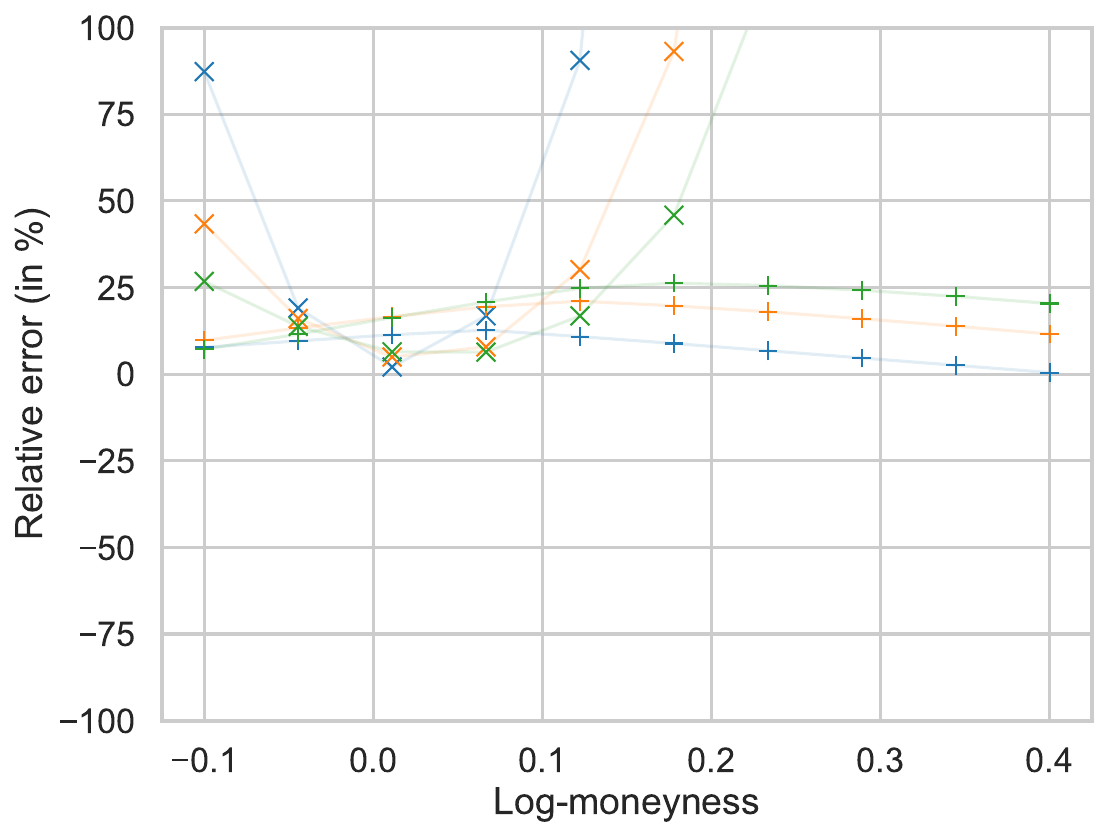}
    \caption{VIX smiles of Scenario 4 in the mixed rough Bergomi model. Expansion \eqref{eq:vix_iv_mixed_kernel_approx} and extended expansion \eqref{eq:vix_iv_proxy_mixed_kernel_extended}, both around $\sigma_{P,2}$ (lower). Reference VIX futures $F$ for the 1-, 3-, and 6-month maturities: 0.2236, 0.2130, 0.2061.}
    \label{fig:mixed_case_rbergomi_2_re}
\end{figure}

\end{document}